\documentclass[11pt,a4paper]{article}
\pdfoutput=1 
\usepackage{fullpage}

\usepackage{algorithm}
\usepackage[noend]{algpseudocode}
\usepackage{amsmath}
\usepackage{amsfonts}
\usepackage{amssymb}
\usepackage{amsthm}
\allowdisplaybreaks
\usepackage{mathtools}

\usepackage[T1]{fontenc}
\usepackage[tt=false]{libertine}
\usepackage[libertine]{newtxmath}

\usepackage{hyperref}
\usepackage[svgnames]{xcolor}
\hypersetup{colorlinks={true},urlcolor={blue},linkcolor={DarkBlue},citecolor=[named]{DarkGreen}}
\usepackage[authoryear,square]{natbib}

\usepackage{enumitem}

\usepackage{subcaption}

\usepackage{tikz}
\usepackage{mdframed}
\mdfdefinestyle{MyFrame}{%
    linecolor=black,
    outerlinewidth=2pt,
    innertopmargin=4pt,
    innerbottommargin=4pt,
    innerrightmargin=4pt,
    innerleftmargin=4pt,
        leftmargin = 4pt,
        rightmargin = 4pt
        }

\usepackage{doi}
\usepackage{bm}
\usepackage{bbm}
\usepackage[capitalise,nameinlink,noabbrev]{cleveref}
\usepackage{xspace}

\def \R{\mathbb R}

\newcommand{\sset}[1]{\left\{ #1\right\}}
\newcommand{\ssets}[1]{\{ #1\}}
\newcommand{\fwh}[1]{\; \left| \; #1 \right.}
\newcommand{\fwhs}[1]{\; | \; #1 }
\newcommand{\card}[1]{\left| #1 \right|}
\newcommand{\cards}[1]{| #1 |}
\newcommand{\union}{\cup}
     
\newcommand{\map}{\to}
 
\newcommand{\inters}{\cap}

\newcommand{\perm}[1]{\textup{perm}(#1)}

\DeclareMathOperator*{\expectation}{\mathbb E}
\newcommand{\expect}[2][]{\expectation_{#1}\nolimits\left[#2\right]}

\DeclareMathOperator*{\probability}{\mathrm{Pr}}
\newcommand{\prob}[1]{\probability\left[#1\right]}

\DeclareMathOperator*{\argmax}{argmax}

\DeclareMathOperator*{\supportdistro}{\mathrm{supp}}
\newcommand{\support}[1]{\supportdistro\left(#1\right)}

\renewcommand\vec{\bm}

\newcommand{\sat}{\textup{\textsc{2/3,3-SAT}}\xspace}

\newcommand*{\poly}{\operatorname{poly}}

\newcommand{\NOT}{\textup{\textsf{NOT}}\xspace}
\newcommand{\OR}{\textup{\textsf{OR}}\xspace}

\newcommand{\PROJ}{\textup{\textsf{PROJ}}\xspace}
\newcommand{\OUT}{\textup{\textsf{OUT}}\xspace}

\newcommand{\priorupperbound}{\ensuremath{\overline{\phi}}}
\newcommand{\priorlowerbound}{\ensuremath{\underline{\phi}}}
\newcommand{\interiorset}[1]{\ensuremath{{#1}^{\circ}}}

\theoremstyle{definition}
\newtheorem{definition}{Definition}

\theoremstyle{plain}
\newtheorem{theorem}{Theorem}[section]
\newtheorem{lemma}[theorem]{Lemma}

\newtheorem{proposition}[theorem]{Proposition}

\newtheorem{question}{General Question}

\newtheorem{inftheorem}{Informal Theorem}
\newtheorem*{opproblem}{Open Problem}

\theoremstyle{definition}
\newtheorem{remark}{Remark}
\newtheorem{example}{Example}

\newcommand*{\myproofname}{Proof}

\usepackage{xcolor}

\title{Equilibrium Computation in First-Price Auctions with Correlated Priors\thanks{Aris Filos-Ratsikas was supported by the UK Engineering and Physical Sciences Research Council (EPSRC) grant EP/Y003624/1. Charalampos Kokkalis was supported by an EPSRC DTA Scholarship (Reference EP/W524384/1).}} 
\author{
\begin{tabular}{c c}
& \\ \textbf{Aris Filos-Ratsikas} & \textbf{Yiannis Giannakopoulos}\\
\small{University of Edinburgh, United Kingdom} & \small{University of Glasgow, United Kingdom} \\
\href{mailto:Aris.Filos-Ratsikas@ed.ac.uk}{\small{\texttt{aris.filos-ratsikas@ed.ac.uk}}} & \href{mailto:yiannis.giannakopoulos@glasgow.ac.uk}{\small{\texttt{yiannis.giannakopoulos@glasgow.ac.uk}}}\\
& \\
\textbf{Alexandros Hollender} & \textbf{Charalampos Kokkalis}\\
\small{University of Oxford, United Kingdom} & \small{University of Edinburgh, United Kingdom} \\
\href{mailto:alexandros.hollender@cs.ox.ac.uk}{\small{\texttt{alexandros.hollender@cs.ox.ac.uk}}} & \href{mailto:charalampos.kokkalis@ed.ac.uk}{\small{\texttt{charalampos.kokkalis@ed.ac.uk}}}
\end{tabular}}

\date{June 5, 2025}

\begin{document} 
\maketitle

\begin{abstract}
We consider the computational complexity of computing Bayes-Nash equilibria in
first-price auctions, where the bidders' values for the item are drawn from a
general (possibly correlated) joint distribution. We show that when the values
and the bidding space are discrete, determining the existence of a pure
Bayes-Nash equilibrium is NP-hard. This is the first hardness result in the
literature of the problem that does not rely on assumptions of subjectivity of
the priors, or convoluted tie-breaking rules. We then present two main
approaches for achieving positive results, via \emph{bid sparsification} and via
\emph{bid densification}. The former is more combinatorial and is based on
enumeration techniques, whereas the latter makes use of the continuous theory of
the problem developed in the economics literature. Using these approaches, we
develop polynomial-time approximation algorithms for computing equilibria in
symmetric settings or settings with a fixed number of bidders, for different
(discrete or continuous) variants of the auction. 
\end{abstract}

\newpage
\vspace{1cm}
\setcounter{tocdepth}{2} 
\tableofcontents
\newpage

\section{Introduction}

The study of the first-price auction has been in the epicentre of auction theory
since the inception of the field \citep{vickrey1961counterspeculation}.
Motivated by the fact that the bidders in this auction have incentives to
\emph{underbid}, the literature in economics since the early 1960s has studied
the game-theoretic aspects of the auction extensively, aiming to understand and
characterize its \emph{equilibria}. These questions are now as relevant as ever,
since first-price auctions and their variants are widely used in practice, e.g.,
in the sale of \emph{ad impressions} on major online platforms
\citep{paes2020competitive,despotakis2021first,conitzer2022multiplicative,aggarwal2024auto}.

When choosing whether to underbid and by how much, bidders base their decisions
on their value for the item for sale (their ``willingness to buy'' the item) and
the \emph{beliefs} that they have about the values of the other bidders. If a
bidder expects that her competitors will not be very interested in purchasing
the item, she might underbid significantly, attempting to win the item at a
rather low price, whereas if she expects stiffer competition, she may choose to
bid closer to her true value. In game-theoretic terms, this situation is most
accurately modelled as a \emph{Bayesian game of incomplete information}
\citep{harsanyi1967games}, and the aforementioned beliefs are modelled by means
of \emph{value distributions}, also known as \emph{probability priors}. Indeed,
in his seminal paper in 1961, \citeauthor{vickrey1961counterspeculation} studied
the case when the value distributions are all identical and uniform, and showed
that a \emph{Bayes-Nash equilibrium} of the auction always exists, and can be
described via a closed form expression. Since then, a series of works have
considered the same type of questions for different assumptions on the
distributions, and produced mainly existence or uniqueness results, and
descriptions of the equilibria only in limited cases, e.g., see
\citep{griesmer1967toward,riley1981optimal,plum1992characterization,marshall1994numerical,maskin1985auction,maskin2000equilibrium,maskin2003uniqueness,lebrun1996existence,lebrun1999first,lebrun2006uniqueness,lizzeri2000uniqueness,Athey2001,athey2007nonparametric,reny2004existence,chawla2013auctions,bergemann2017first}.

Without a doubt, a highlight of this literature is the work of \citet{MW82}, who
considered auctions with \emph{correlated} (or \emph{interdependent}) values. In
this setting, the values of the agents are drawn from a joint distribution which
assigns a probability to each possible tuple of values, and can capture
situations in which the value of a bidder depends on the values of its
competitors. For example, if the item for sale might potentially be resold in a
future auction, this might have an effect on the values of the bidders for the
item, creating dependencies between them \citep{eden2021poa}. Another classic
example is that of \emph{auctions for mineral rights}
\citep{wilson1967competitive}, where the bidders' values come from estimates
about whether a certain oil site contains oil or not, which are based, e.g., on
their own geological surveys. In such a case, the value of a bidder is clearly
affected by the values of the competitors, as them having a larger value
indicates increased likelihood of the presence of oil in the site; see also
\citep{MW82,roughgarden2016optimal}. \citet{MW82} showed that under a particular
form of positive correlation called \emph{affiliation} (which subsumes
\emph{independent} value priors), and under a certain symmetry condition, the
first-price auction always has a (symmetric) equilibrium, and provided a closed
form expression that describes it. In the broader literature of Bayesian games,
correlation was very much present in the original definition of these games in
\citeauthor{harsanyi1967games}'s seminal trilogy
\citep{harsanyi1967games,harsanyi1968games,harsanyiIII}, referred to as
``C-games'', see also \citep{myerson2004comments}.

In recent years, the interest in the equilibria of the first-price auction has
rekindled in the literature of computer science via the prism of computational
complexity. Concretely, the goal of the associated investigations is to either
design polynomial time algorithms for computing these equilibria, or to prove
hardness results for the appropriate computational problems. To this end,
\citet{fghlp2021_sicomp} studied the setting when the value distributions are
continuous, independent and \emph{subjective}, and provided a PPAD-hardness
result for the problem of computing pure Bayes-Nash equilibria  of the auction.
In follow-up work, \citet{fghk24} provided similar hardness results, namely an
NP-hardness result for pure equilibria and a PPAD-hardness result for mixed
equilibria, when the distributions are still independent and subjective, but
discrete. 

The subjectivity assumption imposed in the aforementioned works implies that the
distributions are different from the perspective of different bidders, i.e.,
there is a distribution $F_{ij}$ for the values of bidder $j$, from the
perspective of each bidder $i$. Auctions with subjective priors are hence very
general, and the hardness results about their equilibrium computation, while
important, are relatively weak. At the same time, while
\citet{harsanyi1967games,harsanyi1968games,harsanyiIII} originally defined
Bayesian games in the context of subjective priors, such priors have only been
considered in a handful of works in the literature of auctions, certainly much
fewer than the plethora of works that study auctions with correlated
(non-subjective) values. This can possibly be explained by the fact that
\citeauthor{harsanyi1967games} in his original work argued that the subjective
priors should be \emph{consistent}, i.e., they need to be derivable from a
common prior by applying Bayes's rule. The rationale behind this assertion,
which became known as the \emph{Harsanyi doctrine}\footnote{The term was
seemingly coined by \citet{aumann1976agreeing}.}, is that, with consistent
priors, the differences in beliefs are due to differences in information, which
is typically the case in reality. It turns out that once one imposes this
consistency condition, the resulting Bayesian games are C-games in the language
of \citet{harsanyi1967games}, i.e., games with correlated values, see
\citep{myerson2004comments}. 

The need to remove the subjectivity assumption was highlighted in
\citep{fghlp2021_sicomp,fghk24}, where settling the complexity of the setting
with \emph{independent private values} was posed as a major open problem. In the
quest to establish computational hardness, the most sensible intermediate step
would be to attempt to prove hardness results for computing equilibria in
first-price auctions with correlated priors. Following the discussion above,
such results would be quite important in their own right, given the prevalence
of correlation in auction theory and the associated applications in practice.
This brings us to our first general question.

\begin{question}\label{que:hardness} Can we prove hardness results for computing
    equilibria in the first-price auction with correlated values, without
    imposing any subjectivity assumptions?
\end{question}

\noindent This question was in fact implicitly asked by
\citep{fghlp2021_sicomp}, who stated the complexity of the auction with
consistent subjective priors as an open problem, seemingly unaware of its
connection with the setting with correlated priors that we mentioned above.  

On the other end of the spectrum, we are also interested in obtaining positive
results, i.e., polynomial time algorithms for computing (approximate) equilibria
of the auction. The sensible approach in this case is to start from auctions
with simple distributions and gradually generalize them as much as possible.
Indeed, in this spirit, \citet{fghk24} designed a polynomial time approximation
scheme (PTAS) for computing symmetric mixed equilibria of the auction with
discrete values which are drawn iid from the same distribution. For the
continuous-value variant, \citet{fghlp2021_sicomp} designed an approximation
algorithm when the number of bidders is fixed. Could we hope to extend these
results to the case of correlated priors, or at least for some reasonable forms
of correlation? This motivates our second general question.

\begin{question}\label{que:positive} Can we design polynomial time algorithms
    for computing approximate equilibria in the first-price auction in symmetric
    settings, or settings with a few bidders, when the values exhibit some
    reasonable form of correlation?
\end{question}

\noindent In this paper, we make significant progress on both of these questions. 

\subsection{Our Results and Techniques}
In this work, we study first-price auctions in which the bids come from a
discrete set, and the value distributions are either discrete or continuous.
Both of these settings have been studied in the literature of the problem before
\citep{fghlp2021_sicomp,chen2023complexity,fghk24} and the corresponding
auctions were coined the \emph{Discrete First-Price Auction (DFPA)} and the
\emph{Continuous First-Price Auction (CFPA)}, respectively.  Contrary to all
previous work, which only considered independent private values (IPV), we allow
the values of the bidders to be \emph{correlated}, i.e., to come from a joint
distribution. In the DFPA, we will generally be interested in both pure
Bayes-Nash equilibria (PBNE) and mixed Bayes-Nash equilibria (MBNE), whereas in
the CFPA we will be interested in PBNE only.\footnote{The rationale for this
decision is grounded on corresponding existence theorems \citep{Athey2001} as
well as issues of representation; we are not aware of an appropriate way to
represent MBNE in the CFPA, see also \citep{fghk24}.} 

\subsubsection{A hardness result for correlated values}\label{sec:intro-hardness}

We first aim to provide answers to \cref{que:hardness} above. To this end, we
first consider the question of deciding the existence of a PBNE of the DFPA with
correlated priors. Note that PBNE of the auction in this case are not guaranteed
to exist, even in the simple case of two bidders with iid values, e.g., see
\citep{fghk24}. We provide the following computational hardness result.

\begin{inftheorem}\label{infthm:hardness} The problem of deciding whether an
(approximate) PBNE of the DFPA with correlated values exists is strongly NP-hard. 
\end{inftheorem}

\noindent In fact, the formal version (\cref{thm:NP-completeness}) of this informal theorem shows a stronger
statement, namely that there exists an $\varepsilon$ of size inversely
polynomial in the description of the input, such that the decision problem is
NP-hard, even for $\varepsilon$-approximate equilibria. \cref{infthm:hardness}
is the first computational hardness result in the literature of the problem
(either for the DFPA or the CFPA) that does not require subjectivity assumptions
or unnatural tie-breaking rules.\footnote{\citet{chen2023complexity} showed a
PPAD-hardness result for computing PBNE in the CFPA, but their construction
crucially requires a rather convoluted tie-breaking rule, rather than the
standard uniform tie-breaking rule.}

From a technical perspective, the proof of the NP-hardness result in
\cref{infthm:hardness} is significantly more involved than the proof of the
corresponding NP-hardness result in \citep{fghk24} for the case of subjective
priors. The high-level idea in the previous reduction is to simulate each
operator of an instance of SAT by a gadget involving a pair or a triplet of
bidders, depending on the fan-in of the operator in question, one bidder for the
output of the operator, and one or two bidders for the input. At an equilibrium
of the auction, the bidders will play strategies that encode the boolean values
``true'' and ``false'' in a way that correctly simulates the semantics of the
operators. For this to be achievable, the equilibrium strategies of the output
bidder should only depend on the strategies of the input bidder(s), and
crucially, none of the other bidders. To achieve that, \citet{fghk24} heavily
exploit the subjectivity assumption and set the value of any other bidder to be
$0$ \emph{only from the perspective of the output bidder}. 

We would like to achieve a similar effect without the subjectivity assumption.
First, we observe that we can construct a joint distribution that assigns
positive probability only to tuples in which the input and output bidders of an
operator have positive values, therefore ``localizing'' the issue to the fact
that a bidder can appear with positive value as an output bidder of one operator
and as an input bidder to another. To resolve this issue, we construct our
instance in \emph{layers}, with appropriate \emph{discounting factors} between
layers. Intuitively, the points of the distribution that correspond to some
operator of a lower layer will appear with significantly smaller probability. As
a result, when computing the best response of a bidder which is both an output
to an operator and an input to another, her best response will be primarily
affected by the former, with the latter practically being absorbed in the
approximation error. Once we have simulated the SAT operators properly, we embed
an instance showing the non-existence of PBNE to the construction, in a way such
that an equilibrium exists if and only if the SAT formula is satisfiable.

Here, one might wonder about potential computational hardness results for
computing MBNE of the DFPA or PBNE of the CFPA. Proving such results seems quite
challenging and is beyond the scope of our work. Nonetheless, to this end, our
\cref{lem:discrete-to-continuous-and-back-for-existence} establishes that a
hardness result for the former would imply a hardness result for the latter as
well. We provide some additional discussion in \cref{sec:conclusion}.

\subsubsection{Approximation algorithms via bid sparsification or bid
densification.}\label{sec:intro-results-positive}

We next turn to \cref{que:positive}, and the design of polynomial time
algorithms for computing approximate equilibria for important special cases.
Similarly to previous work, we will focus on auctions with a fixed number of
bidders \citep{fghlp2021_sicomp}, or symmetric auctions (for any number of
bidders) \citep{fghk24}, and we will consider the problem of finding
\emph{monotone} equilibria of the auction. Roughly speaking, a monotone
equilibrium is one in which the bidding functions assign (weakly) higher bids to
higher values. These equilibria, besides being quite natural, are also typically
the only ones for which existence results have been obtained in the literature
of the problem, e.g., see \citep{Athey2001,MW82}. Additionally, monotonicity is
rather integral for the CFPA, as it allows the expression of the equilibrium
strategies by means of a set of ``jump points'', see
\citep{Athey2001,fghlp2021_sicomp} and \cref{sec:CFPA-represent} of our paper
for more details.   

In terms of correlation, we will consider auctions with \emph{affiliated private
values}, the class introduced in the seminal work of \citet{MW82} discussed
earlier. Affiliation is a form of positive correlation, which stipulates that,
when the value of a bidder for the item increases, it is more likely that the
values of the other bidders will be higher as well. It is no exaggeration to say
that, since the work of \citeauthor{MW82}, and driven by their (and subsequently
\citeauthor{Athey2001}'s [\citeyear{Athey2001}]) existence results, affiliation
has become the ``canonical'' form of correlation studied in the auction
literature, e.g., see
\citep{mcadams2007monotonicity,de2010testing,castro2007affiliation,pinkse2005affiliation,kagel1987information,campo2003asymmetry,li2010testing}.

Despite this rich literature, the existence of monotone MBNE of the DFPA for
affiliated values was not known. Our aforementioned
\cref{lem:discrete-to-continuous-and-back-for-existence} effectively allows us
to translate existence results for the PBNE of the CFPA to MBNE of the DFPA, so
at first glance it seems that the desired existence result would follow ``for
free'' from \citep{Athey2001}. This is not the case however, because
\citeauthor{Athey2001}'s result applies only to auctions with strictly positive
density (full support), and the auction that
\cref{lem:discrete-to-continuous-and-back-for-existence} generates
\emph{fundamentally} does not have full support. To circumvent this obstacle, we
first strengthen \citeauthor{Athey2001}'s classic result to auctions without
full support which fall in a natural class appropriate for computational
representations (those with \emph{piecewise-constant} densities).
\cref{lem:discrete-to-continuous-and-back-for-existence} then applies and yields
the following result:

\begin{inftheorem}
A (symmetric) monotone mixed Bayes-Nash equilibrium of the DFPA with (symmetric)
affiliated private values always exists. 
\end{inftheorem}

We discuss the main challenges of obtaining this new existence theorem
in \cref{sec:existence}, together with some discussion about the inherent
connection between monotonicity and affiliation. With the appropriate existence
theorem at hand, we now turn to the design of polynomial time algorithms. We
will design our algorithms using the following two, conceptually polar opposite,
approaches:
\begin{itemize}[leftmargin=*]
    \item[-] \emph{Bid sparsification:} We will substitute the bidding space of
    the DFPA with a smaller subset, at the expense of some approximation error.
    This will allow us to express the equilibrium computation problem as a
    system of polynomial inequalities of manageable size, which can be solved
    using standard techniques from the literature \citep{grigor1988solving}. We
    will use this approach to obtain results for two cases: (a) for auctions
    with a fixed number of bidders, and (b) for auctions with \emph{symmetric
    affiliated private values}, like those studied in \citep{MW82}. 
    \item[-] \emph{Bid densification:} We will substitute the bidding space of
    the DFPA with a \emph{continuous} bidding space, and employ the closed form
    expressions for equilibria devised in the classic economic theory. Since
    such expressions only exist for symmetric settings \citep{MW82}, we will use
    this approach to obtain results for symmetric affiliated private values. 
\end{itemize}
We provide more details about these two approaches, and state our main
results obtained via each of them below.

\paragraph{Bid sparsification.} 

Starting from an instance $G$ of either the CFPA or the DFPA, the idea of bid
sparsification is to create an instance $G'$ of the same auction with bidding
space $B'$ which is a smaller subset of the original bidding space $B$, and
argue that a (monotone) $\varepsilon$-approximate equilibrium on $G'$ is also a
(monotone) $(\varepsilon+\gamma)$-approximate equilibrium of $G$, where $\gamma$
is a parameter that can be chosen as a function of the size of the smaller
bidding space. This idea was introduced first by \citet{chen2023complexity} in
the context of computing PBNE of the CFPA with independent private values (IPV),
and was latter adapted to the computation of MBNE of the DFPA in the IPV setting
by \citet{fghk24}; in the latter work, the associated lemma was coined the
\emph{shrinkage lemma}, see \citep[Lemma 5.1]{fghk24}. 

To design a polynomial time algorithm, one can then devise a system of
polynomial inequalities, a solution to which is an equilibrium of the auction.
From the related literature (e.g., see \citep{grigor1988solving}), it is known
that such a system can be solved to within accuracy $\eta > 0$, in time which is
polynomial in $1/\eta$ and exponential in the number of variables. When
translated to the ``naive'' formulation of the equilibrium computation problem,
this translates to the size of the bidding space and the number of bidders
appearing in the exponent of the running time. For a fixed number of bidders, by
an application of the shrinkage lemma one can (with an additional \emph{rounding
step} to make sure the approximate solution to the system corresponds to a
monotone non-overbidding approximate equilibrium) obtain a PTAS for the problem.
This approach was first taken by \citet{fghlp2021_sicomp} to compute approximate
PBNE of the CFPA in the IPV setting, for a fixed number of bidders.\footnote{We
remark that the corresponding result in \citep{fghlp2021_sicomp} is in fact
stated with the additional assumption that the bidding space has fixed size; as we show in
\cref{sec:sparsification}, it is possible to
apply the shrinkage lemma (\cref{lem:shrinkage}) to obtain a PTAS without this additional assumption.} When the number of bidders is not fixed, \citet{fghk24} showed
for the DFPA that this exponential dependency on the number of bidders can still
be circumvented if we consider symmetric equilibria in settings where the values
are drawn iid from a common distribution. In particular, they proposed a more
succinct way of representing the equilibrium strategies, which takes advantage
of the symmetry, and coined that \emph{the support representation}. They used
all the aforementioned tools to design a PTAS for the computation of MBNE in the
DFPA with iid values.

We extend the previous results to the following statement about
auctions with affiliated values. Note that this result applies to both PBNE of
the CFPA and MBNE of the DFPA.

\begin{inftheorem}\label{infthm:PTAS-sparsification} The problem of computing a
monotone pure (resp. mixed) Bayes-Nash equilibrium in the CFPA (resp. DFPA) with
affiliated private values admits a PTAS when either
\begin{itemize}[leftmargin=*]
\item[-] there is a fixed number of bidders, or 
\item[-] the affiliated private values are symmetric. In that case, the
equilibrium that we compute is also symmetric.
\end{itemize}
\end{inftheorem}

The approach that we employ for obtaining the PTAS of
\cref{infthm:PTAS-sparsification} follows closely those of the previous
literature, but we still perform the technical work required to carefully adapt
and generalize them from the IPV and iid settings to the case of (symmetric)
affiliated values. From a technical perspective, we only need to prove the
theorem for the monotone PBNE of the CFPA, and then invoke our
\cref{lem:discrete-to-continuous-and-back-for-existence} that connects the two
settings to translate those to the monotone MBNE of the DFPA. On the conceptual
side, we observe that the design principles behind these techniques can be seen
as a general bid sparsification approach, to go alongside our newly proposed bid
densification approach, described in the next paragraph. 

\paragraph{Bid densification.} The sparsification approach described above is
sufficient to obtain a polynomial time algorithm for computing
$\varepsilon$-approximate equilibria for any constant $\varepsilon$. Still,
there is something somewhat unsatisfactory about appealing to the system of
polynomial inequalities to find an equilibrium. Indeed, if one were to code a
sparsification-based algorithm on a computer, they would have to ``unbox'' the
highly technical algorithm of \citet{grigor1988solving}, and perform a rather
involved rounding step described in \citep{fghlp2021_sicomp,fghk24}.  On the
other hand, for certain cases of interest, e.g., for symmetric auctions, the
classic auction theory in economics over the past 65 years has managed to
describe the equilibria of the auction via means of appropriate closed form
expressions. Could we make use of this elegant theory to design algorithms for
the equilibrium computation problem? 

There are some inherent challenges associated with this endeavour. First, these
results from economics typically apply to a variant of the auction where both
the value distributions and the bidding space are continuous; we henceforth
refer to this setting as the \emph{Continuous-Continuous First Price-Auction
(CCFPA)}.\label{page:CCFPA_intro_1} In a related manner, translating a closed
form expression (which may contain integrals and algebraic expressions) to an
algorithm for computing a bidding strategy is not straightforward.

The ``bid densification'' approach therefore aims to connect the (monotone,
symmetric) PBNE of the CFPA with those of the CCFPA. In particular, we consider
a CCFPA with exactly the same value distribution as the CFPA. We are now
operating on a continuous bidding space (which we may assume to be the unit
interval $[0,1]$ without loss of generality) and hence we can invoke the closed
form expressions that we mentioned above; for symmetric affiliated values in
particular, we can use the formula devised by \cite{MW82}. However, the bidding
function $\beta$ of the CCFPA might prescribe bids in $[0,1]$ to certain values
that are not a part of the original bidding space $B$ of our CFPA. To circumvent
this, we apply to the monotonicity of the equilibrium and approximate $\beta$ by
a non-decreasing piecewise-constant bidding function $\hat{\beta}$. This
function ``jumps'' from one bid to the next at the values prescribed by the
inverse bidding function $\beta^{-1}$ when taking the points of the discrete
bidding space $B$ as input. Here is where the second complication kicks in:
since $\beta$ is not a discrete object, we can only compute its inversion at the
points of interest approximately, absorbing an extra error factor in the
equilibrium approximation. 

The last step is to show that the function $\hat{\beta}$ is a good approximation
of the equilibrium strategy in the CFPA. We of course cannot hope this to be the
case in general, as the discrete bidding space $B$ might very far from a
continuous bidding space. As long as the granularity of the bidding space is
sufficiently fine, however, the equilibrium approximation should be reasonably
bounded. While this approach seems sensible, there are several intricacies
associated with formally bounding the error in the approximation of $\beta$ by
$\hat{\beta}$, which depend on parameters of the joint distribution. It turns
out that when the density of the distribution is bounded away from zero, or in
the case of iid values (without the positive density assumption), it is possible
to impose an appropriate bound on this error. We state the corresponding theorem
informally below.

\begin{inftheorem}\label{infthm:densification} Assuming that the bidding space
is sufficiently granular, we can compute an approximate PBNE of the CFPA in
polynomial time, when
\begin{itemize}
    \item[-] the values are symmetric affiliated with full support, or
    \item[-] the values are drawn iid from a distribution (that does not need to have full support).
\end{itemize}
\end{inftheorem}

In the formal statement of the theorem (\cref{th:approximate-PBNE-CFPA-to-CCFPA}), the following parameters
appear: the number of bidders $n$, the upper and lower bounds on the density of
the distribution, and the granularity of the bidding space $\delta$ (i.e., the
maximum distance between any two consecutive bids). Assuming constant upper and
lower bounds on the density, it suffices to have sufficiently more bids than
bidders to obtain very strong approximations. In particular, when $\delta <
1/n^2$, we can obtain an algorithm with equilibrium approximation error
$1/\poly(n)$; given that in reality the bidding space of an auction would
typically be \emph{much} larger than the number of bidders (e.g., all multiples
of 5 cents), the algorithm suggested by \cref{infthm:densification} would
compute a rather strong approximation. 

An interesting question here is whether we can obtain similar results to that of
\cref{infthm:densification} for the MBNE of the DFPA as well. While this is
conceivable, it poses some critical technical challenges; we provide a related
discussion in \cref{sec:conclusion}. Finally, we remark that
\cref{infthm:densification} can be also viewed from the perspective of a
platform designer, as instructing how finely to discretize the bidding space
(which is a necessary assumption in reality) in order to obtain strong
equilibrium approximations.      

\subsection{Further Related Work}

The study of equilibria of the first-price auction was initiated by the seminal
work of \citet{vickrey1961counterspeculation}, who proved the existence and
uniqueness of a symmetric equilibrium for the case when all bidders' values are
drawn from a uniform distribution, and provided a closed form formula that
describes it. Since then, a plethora of works in economics studied different
variants of the auction extensively, e.g., see
\citep{Athey2001,athey2007nonparametric,lebrun1996existence,lebrun1999first,lebrun2006uniqueness,maskin1985auction,maskin2000equilibrium,maskin2003uniqueness,griesmer1967toward,plum1992characterization,chwe1989discrete}
and references therein, aiming to produce similar existence results, as well as
reasonable descriptions of the equilibrium bidding functions. As we mentioned
earlier, a highlight of this literature is the work of \citet{MW82} who provided
such results for the case of symmetric affiliated values; we make use of these
results in the development of our polynomial time algorithms in
\cref{sec:densification}. In computer science, auctions with correlated values
have been considered in a plethora of works
\citep{roughgarden2016optimal,eden2018interdependent,eden2021poa,eden2024combinatorial,fu2014optimal,dobzinski2011optimal,cai2012optimal,papadimitriou2011optimal},
but not from the perspective of computational complexity, and not particularly
for the first-price auction. 

The computational complexity of equilibrium computation was first studied by
\citet{fghlp2021_sicomp}, who provided a PPAD-completeness result for computing
pure Bayes-Nash equilibria of the auction in the case of subjective priors, as
well as some positive results for a fixed number of bidders. In the same
setting, \citet{chen2023complexity} managed to prove a PPAD-hardness result
without the subjectivity assumption, but with the addition of a rather
convoluted tie-breaking rule, rather than the standard uniform tie-breaking.
\citet{fghk24} considered the same problem of discrete values and discrete bids
that we do, but, crucially, their hardness results only apply to the case of
subjective priors. \citet{wang2020bayesian} proposed an efficient algorithm for
first price auctions with discrete values and continuous bidding spaces, when
the tie breaking rule employed is a non-uniform rule due to
\citep{maskin2000equilibrium}.

\section{Preliminaries}\label{sec:preliminaries}

In this section we introduce the settings that we will study in this paper. We
will be interested in first-price auctions with discrete bids\footnote{The
literature of the problem in economics often assumes continuous bidding spaces,
but this is primarily a matter of mathematical convenience in order to be able
to obtain closed-form solutions; see, e.g., the discussion
of~\citet[Section~II]{Vickrey1962}. Furthermore, discrete bids are much more
amenable to computational complexity analysis; this is a point made also in all
previous works that study the same setting, namely
\citep{fghlp2021_sicomp,fghk24} and \citep{chen2023complexity}. Finally, note
that even for continuous bidding settings where closed-form solutions are
known~\citep{MW82}, evaluating these formulas, from a computational perspective,
is not straightforward; as a matter of fact, we address this
in~\cref{lemma:inverter-poly-time-oracle-beta-MW}. See also our discussion
in~\cref{remark:discrete-approx-CCFPA}.} and either discrete value distributions
(DFPA) (presented in \cref{sec:DFPA-model}) or continuous value distributions
(CFPA) (presented in \cref{sec:CFPA-model}). In \cref{sec:representation} we
discuss issues related to the representation of inputs and outputs for both
settings, to make them amenable to computational complexity analysis. Finally,
in \cref{sec:expected-utility} we discuss the computation of expected utilities
in the auction, which is essential for our algorithms and membership results. We
start with some general notation that will be useful throughout the paper.  

\subsection{General Notation}
Throughout our paper we use $\support{F}$ to denote the support of a (continuous
or discrete) distribution, with cumulative distribution function (cdf) $F$, and
probability mass function (pmf) $f$ (for the discrete case) and probability
density function (pdf) $f$ (for the continuous case). With a slight abuse of
notation, we will use $F$ to denote both the probability distribution and its
cdf. Also, for a random variable $X$ distributed according to a cdf $F$, we will
use the notation $X\sim F$. For a positive integer $n$, we use
$[n]\coloneqq\ssets{1,2,\dots,n}$. Let $\R$ be the set of real numbers. For any
vector $\vec{x}=(x_1,x_2,\ldots,x_n)\in {\R}^n$ and $i\in[n]$, we let
$\vec{x}_{-i} \coloneq (x_1,x_2,\ldots,x_{i-1},x_{i+1},\ldots,x_n) \in
{\R}^{n-1}$. Given this, we will also denote $\vec{x}=(x_i,\vec{x_{-i}})$.
For a finite set $X$ of cardinality $n$, we use $\Delta(X)$ to denote the set of
all possible distributions over $X$; that is, the $(n-1)$-dimensional unit simplex
$$
\Delta(X) \coloneqq \ssets{\vec{y}\in [0,1]^{n}\fwhs{\sum\nolimits_{i\in[n]} y_i=1}}.
$$
For a positive integer $n$ we use $\perm{n}$ to denote the set of all possible
permutations of indices $[n]$, that is
$\perm{n}\coloneqq\sset{\pi:[n]\to[n]\fwh{\pi\;\;\text{is a bijection}}}$. For a
set $X$, we let $\mathbbm{1}_X$ denote its indicator function, i.e.,
$\mathbbm{1}_X(x)=1$ if $x\in X$, and $\mathbbm{1}_X(x)=0$ otherwise. Finally,
for any $X\subseteq \R^n$ of a Euclidean space, we use
$\interiorset{X}$ to denote its interior (with respect to the standard Euclidean
metric).

\subsection{Discrete First-Price Auctions}\label{sec:DFPA-model} 

In a (discrete, Bayesian) \emph{first-price auction (DFPA)}, there is a set
$N=[n]$ of \emph{bidders} and one item for sale. Each bidder $i$ has a
\emph{value} $v_i\in V_i$ for the item and submits a \emph{bid} $b_i \in B$. The
sets $V_1,V_2,\dots,V_n,B$ are \emph{finite} subsets of $[0,1]$ and are called
the \emph{value spaces} of the bidders and the \emph{bidding space} of the
auction, respectively. Similarly to previous work
\citep{fghk24,chen2023complexity}, we assume that $0 \in B$ (which can be seen
as abstaining from the auction). 

The item is allocated to the highest bidder, who has to submit a payment equal
to her bid. In case of a tie for the highest bid, the winner is determined
according to the \emph{uniform tie-breaking} rule. That is, for a \emph{bid
profile} $\vec{b}=(b_1,\ldots,b_n)$, the \emph{ex-post utility} of bidder $i$
with value $v_i$ is defined as
\begin{equation}
\label{eq:ex_post_utilities}
\tilde{u}_i(\vec{b};v_i) \coloneq 
\begin{cases}
\frac{1}{\cards{W(\vec{b})}}(v_i-b_i), & \text{if}\;\; i\in W(\vec{b}), \\
0, & \text{otherwise}, 
\end{cases}
\qquad\text{where}\;\; W(\vec{b})=\argmax_{j\in N} b_j
\end{equation}

\paragraph{Correlated values.}
In the Bayesian setting, the information bidders have about how much the other
bidders value the item is modelled by a \emph{joint} distribution $F$ over
$\vec{V} := V_1 \times V_2 \times \ldots \times V_n$. \medskip We will also be
interested in special cases of these Bayesian priors, namely:

\begin{itemize}[leftmargin=*]
    \item [--] \emph{\textbf{Affiliated Private Values (APV)}}, where $f$
    satisfies the affiliation condition:
    \begin{equation}
    \forall \vec{v},\vec{v'} \in \vec{V}: f(\vec{v}\vee \vec{v'})\cdot f(\vec{v}\wedge\vec{v'}) \geq f(\vec{v})\cdot f(\vec{v'}) \label{eq:affiliation}
    \end{equation}
    where $\vec{v}\vee\vec{v'}$ denotes the component-wise maximum (join) and
    $\vec{v}\wedge\vec{v'}$ denotes the component-wise minimum (meet) of
    $\vec{v}$ and $\vec{v'}$. A function $f$ that satisfies
    Condition~\eqref{eq:affiliation} is also commonly known in the mathematics
    literature as multivariate totally positive of order 2 (MTP$_2$)
    \citep{karlin1980classes}, and is also often referred to as log-supermodular
    \citep{Athey2001}. Affiliation is a form of positive correlation: the higher
    the value for one bidder, the more likely it is that the values of the other
    bidders will be higher as well. See \citep{MW82,de2010testing} for a more
    elaborate discussion.  

    \item [--] \emph{\textbf{Group-Symmetric Affiliated Private Values
    ($k$-GSAPV)}},\label{page:TSAPV-def} where $f$ satisfies the affiliation
    condition \eqref{eq:affiliation}, and is also \emph{$k$-group-symmetric} in
    the following sense: the set of bidders is partitioned into $k$ distinct
    groups $N_1,N_2,\ldots,N_k$ such that $f$ is symmetric in the arguments
    corresponding to bidders in the same group. That is, $\forall k' \in [k],
    \forall i,j \in N_{k'}$ it holds that $V_i=V_j$ and
    \begin{equation}
    f(v_i,v_j,\vec{v}_{-(i,j)}) = f(v_j,v_i,\vec{v}_{-(i,j)}), \qquad \forall \vec{v}\in \vec{V}, \label{eq:group-symmerty}
    \end{equation}
    where, for $x,y\in R$, $(x,y,\vec{v}_{-(i,j)})$ denotes the vector
    resulting from $\vec{v}$, if we replace its $i$-th coordinate with $x$,
    and its $j$-th coordinate with $y$.

    \item [--] \emph{\textbf{Symmetric Affiliated Private Values
    (SAPV)}},\label{page:SAPV-def} where $f$ satisfies the affiliation condition
    \eqref{eq:affiliation}, and is also symmetric in all of its
    arguments\footnote{In the literature of probability, this property is
    typically referred to as \emph{exchangeability} e.g., see \citep[Section
    6.8., p.~282]{ross2010}.} meaning that 
    \begin{equation} 
    V_1=V_2=\ldots=V_n \ \  \text{ and } \ \ 
    \forall \vec{v} \in \vec{V}, \forall \pi \in \perm{n}: f(\vec{v})=f(v_{\pi(1)},v_{\pi(2)},\ldots,v_{\pi(n)}). \label{eq:symmerty}
    \end{equation}
    
    \item [--] \emph{\textbf{Independent Private Values (IPV)}}, where for each
    bidder $i$ there is a distribution $F_i$ over the values $V_i$, and the
    joint distribution $F$ is a product distribution i.e., $F=F_1\times
    F_2\times\dots\times F_n$.
    
    \item [--] \emph{\textbf{Identical Independent Private Values (IID)}}, which
    is defined as in the case of IPV above, and additionally $V_i=V_{i'}$ and
    $F_i=F_{i'}$ for all $i,i' \in N$. In other words, bidder values are iid
    according to a common distribution $F$.\footnote{We remark that some works
    (e.g., \citep{athey2007nonparametric,de2010testing}) refer to this setting
    as ``symmetric independent private values''; while this would fit well the
    naming conventions of our taxonomy here, we elect to use the term iid, which
    is consistent with the bulk of the literature.}
\end{itemize}

We explain the connection between the different classes of priors below. It can
be observed \citep{MW82} that IPV is a special case of APV, in which the
affiliation condition \eqref{eq:affiliation} holds with equality. Additionally,
IID is both a special case of IPV (in which symmetry is added to independence)
and of SAPV (in which independence is added to symmetry). All of the above are
obviously contained in the class of general correlated values. We have the
inclusion relation diagram shown in \cref{fig:bayesian_priors}. Observe, also,
that $k$-GSAPV captures both APV ($k=n$) and SAPV ($k=1$). 
\begin{figure}[ht!]
\begin{mdframed}[style=MyFrame,nobreak=true,align=center,userdefinedwidth=22em]
\begin{center}
    \tikzset{every picture/.style={line width=0.75pt}} 

    \begin{tikzpicture}[x=0.75pt,y=0.75pt,yscale=-1,xscale=1]

    \draw (190,140) node [anchor=north west][inner sep=0.75pt]   [align=left] {IID};
    \draw (187,91) node [anchor=north west][inner sep=0.75pt]   [align=left] {IPV};
    \draw (252,91) node [anchor=north west][inner sep=0.75pt]   [align=left] {APV};
    \draw (310,90) node [anchor=north west][inner sep=0.75pt]   [align=left] {General Correlation};
    \draw (250,140) node [anchor=north west][inner sep=0.75pt]   [align=left] {SAPV};
    \draw (225,91) node [anchor=north west][inner sep=0.75pt]   [align=left] {$\displaystyle \subseteq $};
    \draw (225,140) node [anchor=north west][inner sep=0.75pt]   [align=left] {$\displaystyle \subseteq $};
    \draw (290,91) node [anchor=north west][inner sep=0.75pt]   [align=left] {$\displaystyle \subseteq $};
    \draw (191,130) node [anchor=north west][inner sep=0.75pt]  [rotate=-270] [align=left] {$\displaystyle \subseteq $};
    \draw (258,130) node [anchor=north west][inner sep=0.75pt]  [rotate=-270] [align=left] {$\displaystyle \subseteq $};
    \end{tikzpicture} 
\end{center}
\end{mdframed}
\caption{The inclusion relation between the different classes of Bayesian priors.}
\label{fig:bayesian_priors}
\end{figure}
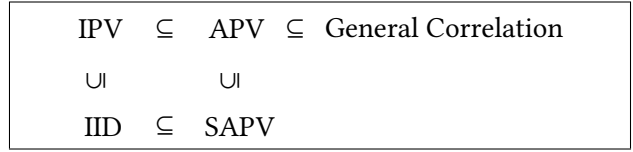

 \subsubsection{The Discrete First-Price Auction Game}
    \label{sec:DFPA-game}

    We now proceed to the definition of the game induced by a DFPA. Given the
    joint probability distribution $F$ for the values, the marginal distribution
    $F_i$ of bidder $i$ is has pmf
    $$
    f_i(v_i)\coloneqq \sum_{\vec{v_{-i}}\in \vec{V_{-i}}}f(v_i,\vec{v_{-i}}) \qquad\text{for all}\;\; v_i\in V_i,
    $$
    where $\vec{V}_{-i}\coloneqq \times_{j\in N\setminus\ssets{i}} V_j$. Given a
    value $v_i\in\support{F_i}\subseteq{V_i}$ that the player may observe with
    nonzero probability, her beliefs about the values of the others are captured
    by the conditional distribution $F_{i|v_i}$ with probability mass function:
    \begin{equation}\label{eq:discrete-conditional}
    f_{i|v_i}(\vec{v_{-i}}) := \frac{f(v_i,\vec{v_{-i}})}{f_i(v_i)}. 
    \end{equation}
    
    We will be interested in both pure and mixed strategies of the bidder in
    this game. A mixed \emph{strategy} of bidder $i$ is a function $\beta_i: V_i
    \rightarrow \Delta(B)$ mapping values to distributions over bids.
    \emph{Pure} strategies correspond to the special case where a mixed strategy
    $\beta_i$ always assigns full mass on single bids; that is, for all $v_i\in
    V_i$ there exists a $b_i \in B$ such that $\beta_i(v_i)(b_i)=1$ and
    $\beta_i(v_i)(b)=0$ for all $b\neq b_i$. Therefore, for simplicity, we will
    sometimes represent pure strategies directly as functions $\beta_i:V_i \to
    B$ from values to bids. Whether the bidding function $\beta_i$ refers to
    pure or mixed strategies will be clear for the context.

    \paragraph{Expected utilities.} Given a strategy profile
    $\vec{\beta}_{-i}\in \times_{j\in N\setminus\ssets{i}}\Delta(B)^{V_j}$ of
    the other bidders, the (interim) \emph{utility} of a bidder $i$ with value
    $v_i$, when bidding $b\in B$, is given by
    \begin{align}
    u_i(b,\vec{\beta}_{-i};v_i) 
        &\coloneq\expect[\vec{v_{-i}}\sim F_{i|v_i}]{\expect[\vec b_{-i}\sim \vec{\beta_{-i}}(\vec{v_{-i}})]{\tilde{u}_i(b,\vec{b_{-i}};v_i)}} \notag\\
        &=  \sum_{\vec{v_{-i}}\in \vec{V}_{-i}}f_{i|v_i}(\vec{v_{-i}})\sum_{\vec{b}_{-i}\in B^{N\setminus\ssets{i}}}\left(\prod_{j\in N\setminus\ssets{i}} \beta_j(v_j)(b_j)\right)\tilde{u}_i(b,\vec{b}_{-i};v_i).\label{eq:DFPA-utility-interim-mixed}
    \end{align}
    where $\vec{\beta}_{-i}(\vec{v}_{-i})$ is a shorthand for the product
    distribution $\times_{j\in N\setminus\ssets{i}}\beta_j(v_j)$, and
    $\beta_j(v_j)(b_j)$ denotes the probability that bidder $j \neq i$ submits
    bid $b_j$ when having value $v_j$.
    Intuitively, the bidder's utility can be calculated via the following
    sequence of steps: (i) a vector $\vec{v_{-i}}$ of the values of the other
    bidders is drawn from the marginal distribution $F_{i\mid v_i}$; (ii) a
    vector of bids $\vec{b_{-i}}$ is drawn from the distribution induced by
    applying the strategy function $\vec{\beta}_{-i}$ to the $\vec{v_{-i}}$
    obtained in (i).

    In the case of mixed strategies, where bidder $i$ randomizes over her bids; i.e., when
    $\vec{\gamma}\in\Delta(B)$ we define:
    \begin{equation}
    \label{eq:DFPA-utility-interim-mixed-randomized-bidding}
    u_i(\vec{\gamma},\vec{\beta}_{-i};v_i) 
    \coloneq \expect[b\sim \vec{\gamma}]{u_i(b,\vec{\beta}_{-i};v_i)}
    = \sum_{b\in B} \vec{\gamma}(b) u_i(b,\vec{\beta}_{-i};v_i).
    \end{equation}

\paragraph{Equilibria.} 
The appropriate notion of equilibrium for Bayesian games is the \emph{(interim)
Bayes-Nash equilibrium}.\footnote{This is the standard notion of equilibrium
used in the literature of the problem---see, e.g., \citep{krishna2009auction}
(section ``Bayesian-Nash Equilibria'' on page 296), or \citep[Section
2]{Athey2001}---where it is assumed that a bidder observes their own value and
the expectation is only over the values of the other bidders. An alternative,
weaker definition that has appeared in the literature before, but is less
natural for auction settings, is the \emph{ex ante} BNE (see, e.g.,
\citep[Eq.~F.2, p.~296]{krishna2009auction}), where the equilibrium condition
takes the expectation over the bidder's own value too. Thus, every interim BNE
is also an ex ante BNE.} Below, we provide the definition of a relaxed notion,
which allows for deviations that can not increase the utilities by more than an
additive parameter $\varepsilon$.  

    \begin{definition}[$\varepsilon$-approximate mixed Bayes-Nash equilibrium of
    the DFPA]\label{def:approx-mixed-bayes-nash-equilibrium} Let $\varepsilon
    \geq 0$. A (mixed) strategy profile $\vec{\beta}=(\beta_1, \ldots, \beta_n)$
    is an (interim) $\varepsilon$-approximate mixed Bayes-Nash equilibrium
    (MBNE) of the DFPA if for any bidder $i \in N$ and any value $v_i \in
    \support{F_i}$, 
    \begin{equation}
    \label{eq:MBNE-def-condition-full}
    u_i(\beta_i(v_i),\vec{\beta}_{-i};v_i) \geq u_i(\vec{\gamma},\vec{\beta}_{-i};v_i) - \varepsilon \qquad \text{for all}\;\; \vec{\gamma}\in \Delta(B).
    \end{equation}
    We will refer to a $0$-approximate MBNE as an \emph{exact} MBNE.
    \end{definition}

\begin{remark}
\label{note:MBNE-def-pure-deviation}
Note that in Condition~\eqref{eq:MBNE-def-condition-full} it suffices to
guarantee that the strategy $\beta_i(v_i)$ does not result in lower utility than
any \emph{pure} strategy $\gamma\in B$ (instead of any \emph{mixed} ones
$\vec{\gamma}\in\Delta(B)$). 
\end{remark}

Similarly, one can define the notion of an (approximate) pure
Bayes-Nash equilibrium:

\begin{definition}[$\varepsilon$-approximate pure Bayes-Nash equilibrium of the
DFPA]\label{def:approx-pure-bayes-nash-equilibrium} Let $\varepsilon \geq 0$. A
pure strategy profile $\hat{\vec{\beta}}=(\hat\beta_1, \ldots, \hat\beta_n)$ is
an (interim) $\varepsilon$-approximate pure Bayes-Nash equilibrium (PBNE) of the
DFPA if for any bidder $i \in N$ and any value $v_i \in \support{F_i}$,
\begin{equation}
    \label{eq:PBNE-def-condition-full}
u_i(\hat\beta_i(v_i),\hat{\vec{\beta}}_{-i};v_i) \geq u_i(b,\hat{\vec{\beta}}_{-i};v_i) - \varepsilon \qquad \text{for all}\;\; b\in B.
\end{equation}
\end{definition}

\paragraph{$\varepsilon$-best responses.} We will say that a strategy
$\beta$ is an $\varepsilon$-best response to the strategies $\vec{\beta}_{-i}$
of the other bidders, if it satisfies \eqref{eq:MBNE-def-condition-full} (for
mixed strategies) and \eqref{eq:PBNE-def-condition-full} (for pure strategies).

\paragraph{Symmetric Equilibria.} When all bidders choose the same
strategies, i.e., $V_i=V_i'$ and $\beta_i=\beta_{i'}$ for all $i,i'\in N$, we
will refer to the (approximate) equilibrium as \emph{symmetric}. More generally,
when the values are $k$-GSAPV, with groups $N_1,N_2,\ldots,N_k$, we will say
that an (approximate) equilibrium is symmetric with respect to those groups, if
$\beta_i=\beta_{i'}$ for all $i,i'\in N_{k'}$, for all $k' \in [k]$.

\paragraph{Monotone Equilibria.} A very natural class of equilibria that
has been studied extensively in the literature (e.g., see
\citep{Athey2001,maskin2000equilibrium,reny2004existence}) is that of
\emph{monotone equilibria}, in which a bidder's bidding behaviour is
non-decreasing in her value. The version that we define below is the appropriate
generalization for mixed strategies and follows
\citep{mcadams2007monotonicity,fghk24}.

\begin{definition}[Monotone mixed strategies and equilibria in a DFPA]
    \label{def:monotonicty-discrete}
A mixed strategy $\beta_i\in \Delta(B)^{V_i}$ (of bidder $i$) will be called
\emph{monotone} if
$$
\max \support{\beta_i(v)} \leq \min \support{\beta_i(v')}
\qquad
\text{for all}\;\; v,v' \in V_i\;\; \text{with}\;\; v < v'. 
$$
A strategy profile (and hence an equilibrium) will be called monotone, if the
strategies of all bidders in it are monotone. In the case of pure strategies,
the condition above reduces to the bidding function $\beta_i$ being
non-decreasing in the bidder's value. 
\end{definition}
As we will discuss in \cref{sec:representation}, monotone equilibria
are also particularly appealing for reasons related to their representation. 

 \paragraph{No overbidding.} Throughout the paper we will only be interested in
 \emph{no overbidding} equilibria, in which no bidder assigns a positive
 probability to any bid larger than her value. This is standard in the
 literature of the problem (e.g., see
 \citep{maskin2000equilibrium,maskin2003uniqueness,lebrun2006uniqueness}), and
 is motivated by the fact that overbidding strategies are weakly dominated by
 any non-overbidding strategy. Intuitively, bidders should never bid more than
 their value, as in that case their utility would always be non-positive. See
 also \citep{fghk24} for a related discussion.

\subsection{Continuous First-Price Auctions}\label{sec:CFPA-model}

The setup of the continuous (Bayesian) \emph{first-price auction (CFPA)} is very
similar to that of the DFPA introduced in \cref{sec:DFPA-model}, but now the
joint distribution $F$ is (absolutely) continuous, with density $f$, supported
over a (non-discrete) value space $V=V_1\times V_2\times\dots\times V_n
\subseteq [0,1]^n$. The definitions of APV, ($k$-G)SAPV, IPV, and iid extend
naturally from the ones in the DFPA by substituting the pmf with a pdf.

The definitions of expected utilities and equilibria also extend
straightforwardly, the only difference being that now the marginal distributions
$F_i$ are defined as follows:
$$
f_i(v_i)\coloneqq \int_{\vec{v_{-i}}\in \vec{V_{-i}}} f(v_i,\vec{v_{-i}})\,\mathrm{d}\vec{v_{-i}}.
$$

\noindent In turn, the utility of bidder $i$ is now defined as:

\begin{align}
    u_i(b,\vec{\beta}_{-i};v_i) 
        &\coloneq\expect[\vec{v_{-i}}\sim F_{i|v_i}]{\expect[\vec b_{-i}\sim \vec{\beta_{-i}}(\vec{v_{-i}})]{\tilde{u}_i(b,\vec{b_{-i}};v_i)}} \notag\\
        &=  \int_{\vec{v}_{-i}\in \vec{V}_{-i}}f_{i|v_i}(v_{-i})\sum_{\vec{b}_{-i}\in B^{N\setminus\ssets{i}}}\left(\prod_{j\in N\setminus\ssets{i}} \beta_j(v_j)(b_j)\right)\tilde{u}_i(b,\vec{b}_{-i};v_i)\,\mathrm{d}\vec{v_{-i}}.\label{eq:CFPA-utility-interim-mixed}
\end{align}

\noindent The definitions of equilibria are identical to the ones in the
discrete setting, as the continuous priors only appear within the computation of
the expected utility.

\subsection{Representation}\label{sec:representation}

As we are interested in results about the computational complexity of
equilibrium computation in first-price auctions, it is crucial to specify how
the inputs and the outputs of the corresponding computational problems will be
represented.  

\subsubsection{Representation in the DFPA}\label{sec:representation-dfpa}
For the DFPA, we consider the following representation.

\paragraph{Input:} Similarly to prior literature on the topic
\citep{fghlp2021_sicomp,chen2023complexity,fghk24}, we assume that the bidding
space $B$ is given explicitly by listing all possible bids $b \in B$, as
rational numbers $r/q$, with $r\geq 0$ and $q >0$ being two integers given by
their binary representation. The sets $V_1,V_2,\ldots, V_n$ are also provided
explicitly in an identical manner.

For the representation of the value distribution $F$, we will express it via its
pmf $f$ and its support $\support{F}$. Specifically, we will express $F$ as a
finite set of $\cards{\support{F}}$-many $n$-tuples of the form
$\vec{v}=(v_1,v_2,\ldots,v_n)\in \vec{V}$, together with their corresponding
mass $f(\vec{v})$ (which is given by a rational number), one for each point
$\vec{v} \in \support{F}$. Specifically, we will use the notation
$\left(\vec{v}, f(\vec{v})\right)
$
to denote that the value vector $\vec{v}=(v_1,v_2, \ldots, v_n)$ is assigned probability $f(\vec{v})$.

In the case where the priors are $k$-GSAPV, it is appropriate to consider a more succinct representation.
We will assume that we are only given the values of the pmf $f$ over the set $\support{F} \cap \vec{V}_{\geq}$, where 
$$\vec{V}_{\geq} \coloneq \{\vec{v} \in \vec{V} : v_1 \geq \ldots \geq v_{n_1}, \ v_{n_1+1} \geq \dots \geq v_{n_1+n_2}, \ \ldots, \ v_{n_1+\ldots+n_{k-1}+1} \geq \ldots  \geq v_{n_1+\ldots+n_k}\},$$ meaning that the values of bidders in the same group are sorted in non-increasing order.
This suffices to fully determine $f$ by the symmetry of the distribution.
More formally, the distribution is represented by a finite list $(\vec{t}^1,p_1), \dots, (\vec{t}^\ell,p_\ell)$, where the $\vec{t}^j$ are distinct tuples in $\vec{V}_{\geq}$ and the $p_j \in [0,1]$ are the probabilities of their corresponding tuple, i.e., $f(\vec{t}^j) = p_j$.
Additionally, in order for these probabilities to induce a valid symmetric distribution over $\vec{V}$, we must also have $\sum_{j \in [\ell]} m_j p_j = 1$, where $m_j := |\{(t_{\pi(1)}^j, \dots, t_{\pi(n)}^j): \pi \in \perm{n_1,n_2,\ldots,n_k}\}|$ is the number\footnote{We can write $m_j = \frac{n_1!n_2! \dots n_k!}{\prod_{v \in V_1} n_v^1!\prod_{v \in V_2} n_v^2! \dots \prod_{v \in V_k} n_v^k!}$, where $n_v^\ell$ is the number of times that value $v \in V_\ell$ appears in tuple $\vec{t}^j$ in the entries corresponding to group $\ell$.} of distinct, \emph{group-valid} permutations of the tuple $\vec{t}^j$, i.e., permutations which, given the $k$ groups in the $k$-GSAPV setting, only allow for exchanges between the entries corresponding to bidders in the same group.
Given $k$ groups, we denote the set of group-valid permutations of the integers from $1$ to $n$ by $\perm{n_1,n_2,\ldots,n_k}$.

\paragraph{Output:} A mixed strategy of player $i$ will also be explicitly
represented using rational numbers 
\[\{p_i(v,b)\}_{v \in V_i, b \in B} \in [0,1] \ \ \text{with} \ \  \sum_{b \in
B} p_i(v,b) =1 \text{ for all } v \in V_i,\] Here, $p_i(v,b)$ denotes the
probability that bidder $i$ submits bid $b$ when having value $v$. A pure
strategy will be represented in the same way, but for any value $v \in V_i$,
there will be exactly one $b \in B$ for which $p_i(v,b)=1$, and $p_i(v,b)$ will
be $0$ for all the other bids. A mixed strategy profile (and hence, an
equilibrium as well) will be output as a vector of mixed strategies, represented
as above.

\subsubsection{Representation in the CFPA} 
\label{sec:CFPA-represent}

While for the case of the DFPA the representation of inputs and outputs was
mostly straightforward, the appropriate representation for the CFPA is more
intricate, as it involves the representation of continuous objects (both the
value distribution and the bidding functions). The representation that we
present below appropriately generalizes the one used in \citep{fghk24} to the
correlated values setting. 

\paragraph{Input:} 
The bidding space $B$ is given explicitly as above. For the value priors, we
consider distributions over $[0,1]^n$ with density functions of the form
$f(\vec{v}) = \sum_{j=1}^\ell w_j \cdot \mathbbm{1}_{R^j}(\vec{v})$, where the
$R^j \subseteq [0,1]^n$ are hyperrectangles, i.e.,
$$R^j = [a^j_1,b^j_1] \times \dots \times [a^j_n,b^j_n],$$ and $\ell$ is a
positive integer. Thus, the distribution is represented by a list $(R^1,w_1),
\dots, (R^\ell,w_\ell)$, where each hyperrectangle $R^j$ is simply represented
by the numbers $a^j_1, b^j_1, \dots, a^j_n, b^j_n$ as above.

For general group-symmetric instances ($k$-GSAPV), we will consider, without loss of generality, that
the $k$ groups are ordered, such that a value profile $\vec{v}$ specifies the
value of bidders in $N_1$ in its first arguments, followed by the value of
bidders in $N_2$ etc.. We can then consider a more succinct representation of
the distribution where we are given a list $(R^1,w_1), \dots, (R^\ell,w_\ell)$
and the density function is:
\begin{equation}
    \label{eq:represent-CFPA-density-permutations-group}
    f(\vec{v}) = \sum_{j=1}^\ell w_j \sum_{\pi \in p(n, \vec{g})} \mathbbm{1}_{\pi(R^j)}(\vec{v}),
\end{equation}
where we extend our permutation notation to let $\pi(R^j)\coloneqq
[a^j_{\pi(1)},b^j_{\pi(1)}] \times \dots \times [a^j_{\pi(n)},b^j_{\pi(n)}]$
and, we also let $p(n, \vec{g})$ denote set of permutations of $n$ objects being
placed (contiguously) in $k$ groups of size $g_1,g_2,\ldots,g_k$ (in our case we
set $\vec{g}=(n_1,n_2,\ldots,n_k)$), allowing only for permutations of objects
within the \emph{same} group. In particular, for $k=1$ (SAPV),
representation~\eqref{eq:represent-CFPA-density-permutations-group} can be
simply expressed as  
\begin{equation}
    \label{eq:represent-CFPA-density-permutations}
    f(\vec{v}) = \sum_{j=1}^\ell w_j \sum_{\pi \in \perm{n}} \mathbbm{1}_{\pi(R^j)}(\vec{v}).
\end{equation}

Some of our results in \cref{{sec:densification}} will explicitly concern the IID
    setting. In this case, the values are represented even more succinctly as
    follows. The marginal distribution $F_1=F_2=\dots=F_n$ are given to us in
    the input\label{page:IID-represent-internal-lemma} by a set of possible
    values $0=a_0<a_1<a_2<\dots<a_{k-1}<a_{k}=1$ together with their
    corresponding probabilities $\ssets{p_j}_{j\in[k]}\in [0,1]$, such that
    $\sum_{j=1}^k(a_j-a_{j-1})p_j=1$. Then, the marginal's density is given by
    $f_1(x)=p_{j}$ for all $x\in(a_{j-1},a_{j})$, $j\in[k]$.

\paragraph{Output:} For the CFPA, we will only be interested in pure equilibria.
Given that the (pure) bidding strategy of each bidder $i$ needs to specify which
of the (finitely many) bids will be played for each possible value $v \in V_i$,
it is inherently a continuous function, which makes its representation
non-trivial. In fact, we are not aware of how to represent general equilibria of
the CFPA. This is part of the reason why the related literature both in
economics \citep{Athey2001,reny2004existence,maskin2000equilibrium} and in
computer science \citep{fghlp2021_sicomp,chen2023complexity,fghk24} has
restricted attention to monotone equilibria (see
\cref{def:monotonicty-discrete}). For monotone strategies, the literature (e.g.,
see \citep{Athey2001}) has proposed an efficient representation by means of
their \emph{jump points}, i.e., the values $v \in V_i$ for which the strategy of
bidder $i$ changes from a bid to the next. Formally, following
\citep{fghlp2021_sicomp} we define \begin{equation}\label{eq:jump_points}
s_i(b)\coloneqq\sup\sset{v\fwh{\beta_i(v)\leq b}}.
\end{equation}
Intuitively, $s_i(b)$ is the largest value for which bidder $i$'s bid would be
at most $b$. It will be useful to think of these jump points as an ``inverse''
of the bidding strategy, as in that case we obtain $\beta_i(v)=b_{j+1}$ for any
$v \in (s_i(b_j),s_i(b_{j+1}))$, for two consecutive bids $b_j, b_{j+1} \in B$.
A pictorial representation of the bidding strategy is shown in
\cref{fig:inversebidding}.

\begin{figure}[ht]\begin{center}
    \begin{tikzpicture}
        \draw[->] (0,0) -- (5,0);
        \draw[->] (0,0) -- (0,3);
        \node[below] at (5,0) {$v$};
        \node[below] at (3,0) {\scriptsize $s_i(b_{j+1})$};
        \node[below] at (2,0) {\scriptsize $s_i(b_{j})$};
        \node[left] at (0,3) {$\beta_i(v)$};
        \node[left] at (0,1.0) {$b_{j+1}$};
        \node[left] at (0,0.5) {$b_{j}$};
        \coordinate (b1l) at (0,0.0) {};
        \coordinate (b1r) at (1,0.0) {};
        \coordinate (b2l) at (1,0.5) {};
        \coordinate (b2r) at (2,0.5) {};
        \coordinate (b3l) at (2,1.0) {};
        \coordinate (b3r) at (3,1.0) {};
        \coordinate (b4l) at (3,1.5) {};
        \coordinate (b4r) at (4,1.5) {};
        \coordinate (b5l) at (4,2.0) {};
        \coordinate (b5r) at (5,2.0) {};
        \draw[black!50!gray,thick] (b1l) -- (b1r);
        \draw[black!50!gray,thick] (b2l) -- (b2r);
        \draw[black!50!gray,thick] (b3l) -- (b3r);
        \draw[black!50!gray,thick] (b4l) -- (b4r);
        \draw[black!50!gray,thick] (b5l) -- (b5r);
        \draw[dotted] (0,1.0) -- (b3l);
        \draw[dotted] (b3r) -- (3,0.0);
        \draw[dotted] (b2r) -- (2,0.0);
        \draw[black!50!gray,thick,fill=black!50!gray] (b1r) circle (2pt);
        \draw[black!50!gray,thick,fill=black!50!gray] (b2r) circle (2pt);
        \draw[black!50!gray,thick,fill=black!50!gray] (b3r) circle (2pt);
        \draw[black!50!gray,thick,fill=black!50!gray] (b4r) circle (2pt);
        \draw[black!50!gray,thick,fill=black!50!gray] (b5r) circle (2pt);
        \end{tikzpicture}
\end{center}
\caption{A monotone bidding strategy $\beta_i(\cdot)$,  succinctly represented
by its jump points, $s_i(b)$ for $b\in B$.\label{fig:inversebidding}}
\end{figure}
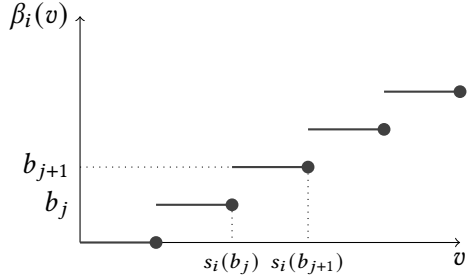 

Given the above, we will concretely represent a monotone pure strategy of bidder
$i$ as a list of jump points $\{s_i(b)\}_{b \in B}$, and a pure strategy profile
(and hence, an equilibrium as well) as a vector of those lists, one for each
bidder.

\subsection{Expected Utility Computation}\label{sec:expected-utility}
Next, we discuss the efficient computation of the bidders' utilities.

The utility of bidder $i$ with value $v_i$ when playing bid $b$ can be written as follows    
\begin{equation}\label{eq:utility}
 u_i(b,\vec{\beta_{-i}}) = (v_i-b) \cdot H_{i}(b, \vec{\beta_{-i}};v_i),
\end{equation}
where $H_{i}(b, \vec{\beta_{-i}};v_i)$ is the probability that bidder $i$ wins
the auction when her value is $v_i$, her bid is $b$, and the other bidders play
according to strategies $\vec{\beta_{-i}}$. In the DFPA setting, we can express
this as follows:

\begin{equation}\label{eq:winning-prob}
    H_{i}(b,\vec{\beta_{-i}};v_i) \coloneqq \sum_{\vec{v_{-i}} \in \vec{V}_{-i}} \frac{\mathbbm{1}\left[b \geq \max_{j \in N \backslash \{i\}}\beta_j(v_j)\right]}{\left| \ssets{j\in N \fwh{\beta_j(v_j)=b} } \right| } \cdot f_{i|v_i}(\vec{v_{-i}})
\end{equation}
Adapting this definition to the CFPA setting is straightforward; it can be
achieved by replacing the pmf with the corresponding pdf and the sum by an
integral.

To study the DFPA and the CFPA as computational problems, it is important to
show how to efficiently compute the quantities $H_{i}(b,\vec{\beta_{-i}};v_i)$
above, and as a result, the bidders' utilities. We establish that in the
following lemma, which we prove in \cref{app:utilities}.

\begin{proposition}\label{lem:efficient-computation} Bidder (interim) utilities
(see~\eqref{eq:utility}) are computable in polynomial time, in all the auction
settings we study in our paper, namely: DFPA, for both general and symmetric
correlated values, with respect to mixed bidding strategies; and CFPA for both
general and symmetric correlated values, with respect to pure bidding
strategies.
\end{proposition}

\section{(Non-) Existence of Equilibria}\label{sec:existence}

Before we dive into our computational complexity results, we first present some
results about the existence of Bayes Nash equilibria. These are useful to
specify the type of computational results we should be looking for: in cases
where equilibria do not always exist, the appropriate computational question is
to decide their existence, whereas in cases where existence is guaranteed, the
problem becomes a total search problem, and our goal is to provide algorithms
that compute them or hardness results for the appropriate complexity classes.

It is known from the related literature (see \citep{fghk24} and references
therein) that an (approximate) PBNE of the DFPA need not exist, even when the
value priors are iid:
\begin{theorem}[\cite{maskin1985auction,fghk24}]
        \label{th:tight_existence_approx}
        There are instances of the DFPA, even with two bidders and IID settings,
        for which $\varepsilon$-approximate PBNE do not exist, for any
        $\varepsilon$ smaller than a sufficiently small constant. 
\end{theorem}

Since IID is a subclass of all the priors that we consider (see
\cref{fig:bayesian_priors}), this immediately implies the same non-existence
result for all the other classes as well. Given this, the equilibrium
computation problem for PBNE is a decision problem, which we settle for general
correlated values in \cref{sec:np-hardness} below. The natural follow-up
question is to consider the MBNE of the auction. The next theorem follows from
known results about mixed equilibria of general Bayesian games (e.g., see
\citep[Sec.~7.2.3]{Jehle2001a} and the discussion in \citep{fghk24}). 

\begin{theorem}[Existence of MBNE for general correlated
values]\label{thm:MBNE-existence-DFPA-general_correlated} For the DFPA with
general correlated values, a MBNE always exists. 
\end{theorem}

\subsection{Monotone Equilibria}
We next turn our attention to monotone equilibria. As we mentioned earlier,
these are very natural and well-studied in the literature. Based on the
existence result of \citet{Athey2001} for the PBNE of the CFPA, and the
techniques developed in \citet{fghk24}, one can  derive the following existence
result.\footnote{More precisely, \citet{fghk24} proved an equivalence between
approximate PBNE of the CFPA and approximate MBNE of the DFPA.
\citeauthor{Athey2001}'s result establishes the existence of an exact PBNE for
the CFPA. This, together with the aforementioned equivalence and an appropriate
limit argument, can be used to show
\cref{thm:MBNE-existence-DFPA-general_correlated}. In fact, \citet{fghk24} did
exactly this to prove existence of the MBNE for the DFPA with IID values; the
extension to the case of IPV values is straightforward.}

\begin{theorem}[\citep{fghk24}]\label{thm:MNBE-monotone-existence-DFPA-IPV}
For the DFPA with independent private values (IPV), a monotone MBNE always exists. 
\end{theorem}

So, could we hope to extend the existence result of
\cref{thm:MNBE-monotone-existence-DFPA-IPV} to the class of general correlated
values? Below, we provide a counterexample, which establishes that monotone
equilibria in this case need not exist. The counterexample is inspired by an
instance used by \citet[Example~1, p.~100]{jackson2005existence} to show
non-existence of pure Bayes-Nash equilibria in first-price auctions in which
both the value space and the bidding space are continuous. 

\begin{proposition}\label{thm:MBNE-monotone-non-exist} There are instances of
    the DFPA with general correlated values, even with $2$ bidders and a value
    space of size $3$, for which monotone MBNE do not exist. 
\end{proposition}

\begin{proof}
    Consider an instance of the DFPA with two bidders, whose values $(v_1,v_2)$
    are uniformly distributed over the set $\{(0,1), (1/2,1/2), (1,0)\}$, and
    let the bidding space be $B=\{0,1/10,2/10,\ldots,1\}$. Assume by
    contradiction that $(\beta_1, \beta_2)$ is a monotone MBNE. We will analyse
    the equilibrium strategy of bidder $1$, depending on the value that she
    observes (the analysis for bidder $2$ is symmetric by design in our
    example).
    \begin{itemize}
        \item[-] Case 1: $v_1=0$. In this case bidder $1$ will play the pure
        strategy which bids $0$, due to the no-overbidding assumption.
        \item[-] Case 2: $v_1=1$. In this case bidder $1$ knows that $v_2=0$
        (considering the marginal distribution conditioned on $v_1=1$) with
        probability $1$ and, therefore, bidder $2$ plays the pure strategy which
        bids $0$. It is then straightforward to verify that the unique best
        response of bidder $1$ is to play the pure strategy which bids $1/10$,
        winning the item (without a tie) at the lowest price possible. 
    \end{itemize}
    Since $\beta_1$ is a monotone strategy, we know that for any $v<1$, it
    holds that $\support{\beta_1(v)} \subseteq \{0,1/10\}$; in particular, this
    implies that $\support{\beta_1(1/2)} \subseteq \{0,1/10\}$. Notice that,
    when having value $1/2$ and competing against any strategy supported in
    $\{0,1/10\}$, the pure strategy of bidding $0$ is strictly dominated by the
    pure strategy of bidding $1/10$, therefore it cannot be played with positive
    probability at any mixed equilibrium. Hence, at any monotone MBNE, the
    strategy of each player could only be supported in $\{1/10\}$, meaning it
    would have to be the pure strategy of bidding $1/10$. But this cannot
    satisfy the equilibrium condition, since the pure strategy of bidding $2/10$
    would yield strictly higher utility to the deviating bidder, contradicting
    the assumption that $(\beta_1,\beta_2)$ was a monotone MBNE.
\end{proof}

We remark that monotonicity is a key property for the above counterexample.
Indeed, in this particular instance, one can prove that not only non-monotone
MBNE exist (which is guaranteed by
\cref{thm:MBNE-existence-DFPA-general_correlated} in any case) but in fact, even
non-monotone PBNE exist. 

\paragraph{Monotonicity and correlation.} A closer inspection of the proof of
\cref{thm:MBNE-monotone-non-exist} reveals an interesting relation between
monotonicity and the correlation of the bidders' values. When the value of
bidder 1 is $1$, her best response is to bid $1/10$, as the value of bidder 2
(and hence, her bid) is $0$. However, when the value of bidder 1 is $1/2$, by
the correlation in the values, the value of bidder 2 is also $1/2$, and hence
bidding at most $1/10$ (as stipulated by a monotone strategy) is \emph{not} a
reasonable choice. This is because the values are anti-correlated in the joint
distribution; in such cases, monotonicity does not seem to be a reasonable
assumption. This is in contrast to the IPV setting in which monotonicity is a
very natural property, and in fact non-monotone strategies are weakly dominated
by monotone ones. Generalizing the IPV setting, monotonicity would only make
sense in the presence of (weakly) positive correlation between the values; the
most fundamental and well-studied such type of correlation is the setting of
affiliated private values, considered below.\footnote{We remark that
\citet{MW82} and \citet{Athey2001} also consider settings with more general
affiliated values, beyond the APV setting; these are outside the scope of our
work. We also remark that even in the regime of private values, affiliation is
not the only condition that ensures (weakly) positive correlation. However,
besides being the most popular such condition, it is also one of the few ones
for which the existence of a monotone equilibrium is guaranteed, see
\citep{castro2007affiliation} for a very interesting discussion.}

\subsection{Affiliated Private Values}

Recall the definition of the affiliation condition \eqref{eq:affiliation} from
\cref{sec:DFPA-model}. Intuitively speaking, when the values are affiliated,
then a higher value for one bidder implies that it is more likely that the other
bidders will have higher values as well. Affiliation is the form of correlation
that has predominantly been studied in the literature of the problem
\citep{Athey2001,MW82,krishna2009auction}. Most relevant to us is the following
result of \citet{Athey2001} for the PBNE of the CFPA. 

\begin{theorem}[\citep{Athey2001}]\label{thm:athey-existence-affiliated} For the
CFPA with affiliated private values (APV), a monotone PBNE always exists, when
the joint probability distribution has strictly positive density.  
\end{theorem}

The strict positivity of the density is a rather restrictive assumption, both as
a standalone condition for the CFPA, but also from a technical perspective. To
see this, consider the quest of obtaining a similar existence result for the
monotone MBNE of the DFPA. To achieve that, one can follow the blueprint laid
out by \citet{fghk24}, which connects those equilibria with the PBNE of the CFPA
in the IPV setting. The idea in \citep{fghk24} is a reduction from the discrete
to the continuous variant, via a simulation of the discrete priors $F_i$ in the
DFPA by piecewise constant continuous priors $F'_i$ in the CFPA, such that a
PBNE in the CFPA and a MBNE in the DFPA induce the same distribution over bids.
Then, one can invoke an existence theorem for PBNE of the continuous variant and
obtain the existence of MBNE in the discrete variant. This reduction is only
\emph{approximate}, i.e., when used verbatim it can only guarantee the existence
of $\varepsilon$-approximate MBNE of the DFPA, for all $\varepsilon >0$. The
final step is to use an appropriate limit argument to obtain the result for
exact equilibria (i.e., $\varepsilon =0$). 

We can indeed construct such a reduction for the APV setting, which we state in
\cref{lem:discrete-to-continuous-and-back-for-existence} below. Crucially
however, this reduction constructs a continuous distribution in the CFPA which
\emph{fundamentally} needs to have parts with zero density; indeed, one could
attempt to ``smoothen'' the distribution by artificially adding a small amount
of mass to the zero parts, but this would inherently ``break'' the affiliation
condition. Given this, \citeauthor{Athey2001}'s result cannot be used to obtain
the existence of monotone MBNE of the DFPA.

The only way around this obstacle is seemingly to prove a corresponding
existence theorem for the CFPA without the positive density assumption. This
however imposes certain challenges, the main one being that the \emph{single
crossing condition (SCC)} of \citet{milgrom1994monotone} that
\citeauthor{Athey2001}'s proof heavily relies on is not satisfied in this case.
We circumvent this obstacle, by defining a weaker property, which we refer to as
\emph{Forward-SCC}. While our setting with possibly zero densities now satisfies
the Forward-SCC, showing that \citeauthor{kakutani1941generalization}'s fixed
point theorem [\citeyear{kakutani1941generalization}] can still be applied under
this weaker condition becomes more challenging, in particular when arguing
convexity of the best-response sets. We manage to establish the desired
convexity for a certain general class of distributions with piecewise-constant
density functions. Luckily, the distribution that our reduction in
\cref{lem:discrete-to-continuous-and-back-for-existence} constructs is in this
class, and we obtain the existence of a monotone MBNE in the DFPA as a
corollary.
 
\begin{theorem}\label{thm:MNBE-monotone-existence-DFPA-APV} For the CFPA with
affiliated private values (APV), and a distribution with piecewise constant
density, a monotone PBNE always exists. As a result, in the DFPA with affiliated
private values, a monotone MBNE always exists. Furthermore, if the values are
symmetric ($k$-GSAPV, in particular SAPV), then a monotone \emph{symmetric} equilibrium is guaranteed to
exist (in both CFPA and DFPA).
\end{theorem}

In the rest of this section we prove
\cref{thm:MNBE-monotone-existence-DFPA-APV}. Namely, we show the existence of
monotone PBNE in the CFPA with APV, under piecewise constant densities. By
\cref{lem:discrete-to-continuous-and-back-for-existence}, this then yields the
existence of MBNE in the DFPA with APV. For symmetric
instances (i.e., $k$-GSAPV, in particular SAPV), we guarantee the existence of symmetric equilibria.
Before presenting the proof, we begin with a short discussion explaining why
\cref{thm:MNBE-monotone-existence-DFPA-APV} does not follow from existing
results.

\subsubsection{The Single-Crossing Condition (SCC) Fails}

The usual approach for establishing existence of monotone pure Bayes Nash
equilibria in Bayesian games with a finite number of actions is to use Athey's
framework with the single-crossing condition (SCC)~\citep{Athey2001}. Namely,
one first establishes that the SCC holds for the game at hand, and then the
existence immediately follows. In particular, the SCC holds for the CFPA with
APV, when the joint density function has full support.

Unfortunately, in our case we cannot assume that the density has full support,
because we want to obtain an existence result for the DFPA through
\cref{lem:discrete-to-continuous-and-back-for-existence}. Furthermore, adding a
very small baseline mass to a density to make it have full support breaks the
affiliation property. Thus, we would like to establish that the SCC holds in our
case as well.

\begin{definition}[Single-Crossing Condition (SCC)]
A CFPA satisfies the single-crossing condition (SCC), if for any bidder $i \in
[n]$ and for any \emph{monotone} strategy profile $\vec{\beta}_{-i}$,
$$u_i(b^H, \vec{\beta}_{-i};v^L) \geq u_i(b^L, \vec{\beta}_{-i};v^L) \implies
u_i(b^H, \vec{\beta}_{-i};v^H) \geq u_i(b^L, \vec{\beta}_{-i};v^H)$$ and
$$u_i(b^H, \vec{\beta}_{-i};v^L) > u_i(b^L, \vec{\beta}_{-i};v^L) \implies
u_i(b^H, \vec{\beta}_{-i};v^H) > u_i(b^L, \vec{\beta}_{-i};v^H)$$ for all $b^L,
b^H \in B$ and $v^L, v^H \in V_i$ with $b^L < b^H \leq v^L < v^H$.
\end{definition}

Unfortunately, the SCC is not guaranteed to hold for the CFPA with APV when we
do not have full support, as the following example shows.

\begin{example}
Consider a CFPA with two bidders, with bidding space $B = \{0,1/4,1/2\}$ and the
following joint density function over $[0,1]^2$, $f(v_1,v_2) = 32 \cdot
\mathbbm{1}_{[1/4,3/8]^2}(v_1,v_2) + 32 \cdot \mathbbm{1}_{[3/8,1/2] \times
[7/8,1]}(v_1,v_2)$. Note that this density satisfies the affiliation condition.
Assume that bidder 2 uses the following monotone (and non-overbidding) strategy:
she bids $0$ when $v_2 \in (0,1/4)$, $1/4$ when $v_2 \in (1/4,7/8)$, and $1/2$
when $v_2 \in (7/8,1)$.

Let us examine the set of best-response bids for bidder 1 at two particular
values. When bidder 1 has value $v_1^L = 5/16$, the other bidder has value $v_2
\in [1/4,3/8]$ and thus bids $1/4$ with probability 1. As a result, the only
best-response for bidder 1 is to bid $1/4$ as well (she cannot bid $1/2$, since
that would be above her value).

On the other hand, when bidder 1 has value $v_1^H = 7/16$, the other bidder has
value $v_2 \in [7/8,1]$ and thus bids $1/2$ with probability 1. As a result,
bidder 1 is indifferent between bidding $0$ or $1/4$, since both options give
her zero utility, and she cannot bid $1/2$, because that would be above her
value.

Now, we can see that the second part of the SCC fails. Indeed, bidding $b^H =
1/4$ is strictly better than bidding $b^L = 0$ at value $v_1^L$ for bidder 1,
but at the higher value $v_1^H$, the bidder is indifferent between the two
options. Note however that this does not contradict the first part of the SCC.
Indeed, as we will show below, the first part will always hold in our setting.
\end{example}

Due to the failure of the SCC, in the next section we investigate whether
existence can be shown by using only the first part of the SCC (which we call
Forward-SCC below). We show that this is possible for the CFPA with APV, with
the additional assumption that the joint density function is piecewise constant.

\subsubsection{The Proof of Existence}

We use $V_i \subseteq [0,1]$ to denote the support of the marginal distribution
of bidder $i$'s value, i.e, the support of the distribution with density
$f_i(v_i) = \int_{\vec{v}_{-i}} f(\vec{v}) d\vec{v}_{-i}$. As a result, the
conditional distribution $f(\vec{v}_{-i}|v_i)$, and thus the utility function
$u_i$, are well-defined for all $v_i \in V_i$. We recall that the utility
function of bidder $i$ can be written as $u_i(b, \vec{\beta}_{-i};v_i) = (v_i -
b) \cdot H_i(b, \vec{\beta}_{-i};v_i)$. Here $H_i(b, \vec{\beta}_{-i};v_i)$
denotes the probability of bidder $i$ winning the item, given that she bids $b$
and that the other bidders bid according to the strategy profile
$\vec{\beta}_{-i}$, conditioned on the fact that bidder $i$ has value $v_i$ for
the item. Note that $H_i(b, \vec{\beta}_{-i};v_i) \geq H_i(b',
\vec{\beta}_{-i};v_i)$, whenever $b \geq b'$.

\begin{definition}
Let $X$ be a lattice. A function $h: X \to {\R}_{\geq 0}$ is
\emph{log-supermodular} if for all $x, x' \in X$
$$h(x \wedge x') \cdot h(x \vee x') \geq h(x) \cdot h(x').$$
\end{definition}

Note that the affiliation condition for the joint distribution of values is
equivalent to saying that the joint density is log-supermodular.

\begin{lemma}[\citep{Athey2001}]\label{lem:H-log-super} In the CFPA with APV,
for any bidder $i \in [n]$ and for any \emph{monotone} strategy profile
$\vec{\beta}_{-i}$, the function $(b,v_i) \mapsto H_i(b,\vec{\beta}_{-i};v_i)$
is log-supermodular.
\end{lemma}

\begin{proof}
We can write
\begin{equation*}
\begin{split}
H_i(b,\vec{\beta}_{-i};v_i) &= \int_{\vec{v}_{-i} \in \vec{V}_{-i}} \phi_i(b, \vec{\beta}_{-i}(\vec{v}_{-i})) \cdot f(\vec{v}_{-i}|v_i) d\vec{v}_{-i}\\
&= \int_{\vec{v}_{-i} \in \vec{V}_{-i}} \phi_i(b, \vec{\beta}_{-i}(\vec{v}_{-i})) \cdot \frac{f(v_i, \vec{v}_{-i})}{f_i(v_i)} d\vec{v}_{-i}
\end{split}
\end{equation*}
where
\begin{equation*}
\phi_i(\vec{b}) \coloneq 
\begin{cases}
\frac{1}{\cards{W(\vec{b})}}, & \text{if}\;\; i\in W(\vec{b}), \\
0, & \text{otherwise}, 
\end{cases}
\qquad\text{where}\;\; W(\vec{b})=\argmax_{j\in N} b_j
\end{equation*}
Now, it can be checked that $\phi_i$ is log-supermodular (see, e.g.,
\citep[p.~886]{Athey2001}). Since the strategy profile $\vec{\beta}_{-i}$ is
monotone, it follows that $(b, \vec{v}_{-i}) \mapsto \phi_i(b,
\vec{\beta}_{-i}(\vec{v}_{-i}))$ is also log-supermodular. By assumption, $f$
satisfies the affiliation condition, which means that it is log-supermodular.
Since the one-dimensional function $1/f_i$ is trivially log-supermodular, and
products of log-supermodular functions remain log-supermodular, it follows that
the function inside the integral is log-supermodular in $(b, \vec{v}_{-i},v_i)$.
By the (somewhat surprising) fact that log-supermodularity is preserved by
integration (see, e.g., \citep[pp.~192-193]{athey02-statics}), it follows that
$H_i(b,\vec{\beta}_{-i};v_i)$ is log-supermodular in $(b, v_i)$.
\end{proof}

Unfortunately, the SCC might not hold in our setting. Nevertheless, we show that
the following weaker condition is satisfied. It is one of the two conditions
that SCC requires. We call it forward-SCC, because it guarantees that a bid that
is optimal at the current value, cannot be beaten by a lower bid at a higher
value. The full SCC also includes a similar guarantee in the backwards
direction, i.e., about smaller values.

\begin{lemma}[Forward-SCC]\label{lem:forward-SCC} In the CFPA with APV, for any
bidder $i \in [n]$ and for any \emph{monotone} strategy profile
$\vec{\beta}_{-i}$, if for some $b^L, b^H \in B$ and $v^L, v^H \in V_i$ with
$b^L < b^H \leq v^L < v^H$ we have
$$u_i(b^H, \vec{\beta}_{-i};v^L) \geq u_i(b^L, \vec{\beta}_{-i};v^L)$$
then this implies
$$u_i(b^H, \vec{\beta}_{-i};v^H) \geq u_i(b^L, \vec{\beta}_{-i};v^H).$$
\end{lemma}

\begin{proof}
We omit the term $\vec{\beta}_{-i}$ from the notation, since it remains fixed
throughout the proof. We prove the contrapositive. Let $b^L < b^H \leq v^L <
v^H$ be such that $u_i(b^H;v^H) < u_i(b^L;v^H)$. Our goal is to show that
$u_i(b^H;v^L) < u_i(b^L;v^L)$. The assumption in particular yields that
$u_i(b^L;v^H) > 0$, which implies $H_i(b^L;v^H) > 0$. Since $H_i(\max B;v^L) >
0$, by log-supermodularity of $H_i$ (\cref{lem:H-log-super}), it follows that
$$H_i(\max B;v^H) \cdot H_i(b^L;v^L) \geq H_i(\max B; v^L) \cdot H_i(b^L;v^H) >
0$$ which implies $H_i(b^L; v^L) > 0$, and thus $u_i(b^L;v^L) > 0$. Now, it can
be checked that $(b,v) \mapsto (v-b)$ is also log-supermodular over the lattice
$\{b^L,b^H\} \times \{v^L,v^H\}$, and since the product of two log-supermodular
functions remains log-supermodular, we obtain that $(b,v) \mapsto (v-b) \cdot
H_i(b;v) = u_i(b;v)$ is also log-supermodular over the lattice $\{b^L,b^H\}
\times \{v^L,v^H\}$. As a result, given that $u_i(b^L;v^H) > 0$ and
$u_i(b^L;v^L) > 0$, log-supermodularity yields
$$\frac{u_i(b^H;v^L)}{u_i(b^L;v^L)} \leq \frac{u_i(b^H;v^H)}{u_i(b^L;v^H)} < 1$$
by assumption. This in turn yields $u_i(b^H;v^L) < u_i(b^L;v^L)$, as desired.
\end{proof}

The forward-SCC implies that, in response to a monotone strategy profile
$\vec{\beta}_{-i}$, bidder $i$ can always best-respond with a monotone strategy.
Indeed, the forward-SCC tells us that whenever some bid $b$ becomes an optimal
choice at some value $v_i$, then we will never be forced to play a bid $b' < b$
at any value larger than $v$, since $b$ (weakly) dominates all lower bids for
higher values. Furthermore, it is easy to see that bidder $i$ can always
best-respond with a monotone strategy that is also non-overbidding. Monotone
strategies can be represented by their jump points, and thus the set of all
monotone (and non-overbidding) strategies of a bidder $i$ can be defined as
$$D = \left\{\vec{x} \in [0,1]^{|B|+1}: 0 = x^1 \leq x^2 \leq \dots \leq
x^{|B|+1} = 1, x^j \geq b_j \, \forall j \in [|B|]\right\}$$ where $B = \{b_1,
\dots, b_{|B|}\} \subset [0,1]$ and $b_1 = 0$. A point $x \in D$ represents the
strategy that bids $b_j$ when $v_i \in (x^j,x^{j+1})$ for all $j \in [|B|]$.
Note that we do not care about what happens at the jump points, since this is a
set of measure zero.

Now, given a strategy profile $\vec{x}_{-i} \in D^{n-1}$, let
$\Gamma_i(\vec{x}_{-i}) \subseteq D$ denote the set of all monotone
non-overbidding strategies of bidder $i$ that are best-responses to
$\vec{x}_{-i}$, almost everywhere in the support $V_i$. Define the
correspondence $\Gamma: D^n \to D^n, (\vec{x}_1, \dots, \vec{x}_n) \mapsto
\Gamma_1(\vec{x}_{-1}) \times \dots \times \Gamma_n(\vec{x}_{-n})$. Clearly, any fixed
point of $\Gamma$ yields a PBNE of the auction, and so our goal will be to use
Kakutani's fixed point theorem to prove that such a fixed point must exist.

This correspondence is the same as the one used by \citet{Athey2001}, except for
the constraints that we have introduced in $D$ to disallow overbidding. Note
that $D^n$ is compact and convex. In order to apply Kakutani's fixed point
theorem, we have to show that $\Gamma$ has a closed graph, and that
$\Gamma(\vec{x}_1, \dots, \vec{x}_n)$ is non-empty and convex. We have already
argued about the non-emptiness. We omit the arguments showing the closed graph
property since they are identical to \citep{Athey2001}.

Finally, we argue that $\Gamma(\vec{x}_1, \dots, \vec{x}_n)$ is convex. This is
where our argument differs from \citep{Athey2001}, since we do not have the SCC.
We make the assumption that the joint density function is \emph{piecewise
constant}, meaning that it can be written as the (weighted) sum of a finite
number of hyperrectangle-indicator functions, i.e., $f(\vec{v}) = \sum_{j \in
[m]} w_j \cdot \mathbbm{1}_{R_j}$, where $R_j = [a^j_1,b^j_1] \times \dots
\times [a^j_n,b^j_n]$. In that case, the support $V_i$ for each bidder $i$ is a
finite union of disjoint intervals.

\begin{lemma}
For any bidder $i \in [n]$ and any profile $\vec{x}_{-i} \in D^{n-1}$, the set
$\Gamma_i(\vec{x}_{-i})$ is convex.
\end{lemma}

\begin{proof}
Consider two strategies $\vec{w}, \vec{y} \in D$ that are both best-responses to
$\vec{x}_{-i}$, and let $\vec{z} \in D$ be any convex combination of $\vec{w}$
and $\vec{y}$. Our goal is to show that $\vec{z}$ is also a best-response to
$\vec{x}_{-i}$. For this, it suffices to show that for any $b_m \in B$, bidding
$b_m$ is a best-response at any $v_i \in (z^m,z^{m+1}) \cap V_i$, i.e.,
$u_i(b_m;v_i) \geq \max_{b \in B} u_i(b;v_i)$, where we suppress $\vec{x}_{-i}$
in the notation. If $w^m = w^{m+1}$ and $y^m = y^{m+1}$, then it follows that
$z^m = z^{m+1}$, and the claim trivially holds. Next, consider the case where
$w^m < w^{m+1}$ and $y^m < y^{m+1}$. If the intervals $[w^m,w^{m+1}]$ and
$[y^m,y^{m+1}]$ overlap, then $b_m$ is a best-response at any value in
$(\min(w^m,y^m),\max(w^{m+1},y^{m+1})) \cap V_i$, and thus at any value in
$(z^m,z^{m+1}) \cap V_i$.

If the intervals do not overlap, then assume without loss of generality that
$w^m < w^{m+1} < y^m < y^{m+1}$. We want to show that $b_m$ is a best-response
almost everywhere in $(w^{m+1}, y^m) \cap V_i$. Since the joint density function
is piecewise constant (as defined above), we can partition the interval
$[w^{m+1}, y^m] \cap V_i = \bigcup_{j \in [s]} \overline{A}_j$, where
$\overline{\cdot}$ denotes the closure of a set, such that for all $j \in [s]$
\begin{enumerate}
\item $A_j$ is a non-empty open interval,
\item $A_j \subseteq V_i$,
\item the $A_j$ are pairwise disjoint,
\item for all $v_i, v_i' \in A_j$, $H_i(b;v_i) = H_i(b;v_i')$, so we just write $H_i(b)$
\item the strategies represented by $\vec{w}$ and $\vec{y}$ are constant over $A_j$
\end{enumerate}
Note that this last point is possible because the strategies are monotone
step-functions, and thus they only change value a finite number of times. Now
consider any such interval $A_j$. Since $w^{m+1} \leq \inf A_j$, the strategy
represented by $\vec{w}$ uses a bid $b_\ell$ over all of $A_j$, where $\ell >
m$. Similarly, since $y^m \geq \sup A_j$, the strategy represented by $\vec{y}$
uses a bid $b_k$ over all of $A_j$, where $k < m$. Thus, both $b_k$ and $b_\ell$
are best-responses for any $v_i \in A_j$. This means that for all $v_i \in A_j$
$$(v_i-b_k) \cdot H_i(b_k) = (v_i - b_\ell) \cdot H_i(b_\ell)$$ where we used
point 4 above. Since $b_k < b_\ell$ and $A_j$ is an interval of non-zero length,
it follows that $H_i(b_k) = H_i(b_\ell) = 0$. But this means that $\max_{b \in
B} u_i(b;v_i) = 0$ for all $v_i \in A_j$. As a result, $b_m$ is also a
best-response over all of $A_j$. Since this holds for all $A_j$, we obtain that
$b_m$ is a best-response almost everywhere in $(w^{m+1}, y^m) \cap V_i$, as
desired.
\end{proof}

\paragraph{Symmetric instances.}
For the CFPA with SAPV, we instead use the correspondence $\Gamma': D \to D,
\vec{x} \mapsto \Gamma_1(\vec{x}, \dots, \vec{x})$. By the symmetry of the
instance, any fixed point $\vec{x}$ of this correspondence will yield a
symmetric PBNE $(\vec{x}, \dots, \vec{x})$ of the auction.
More generally, for the $k$-GSAPV setting with groups $N_1, \dots, N_k$, we can use the correspondence $\Gamma': D^k \to D^k, (\vec{x}_1, \dots, \vec{x}_k) \mapsto \Gamma_1(\psi_{-1}(\vec{x}_1, \dots, \vec{x}_k)) \times \dots \times \allowbreak \Gamma_k(\psi_{-k}(\vec{x}_1, \dots, \vec{x}_k))$, where we assume, without loss of generality, that $i \in N_i$ for all $i \in [k]$, and where $\psi: D^k \to D^n$ is defined as, for all $i \in [n]$,
$$\psi_i(\vec{x}_1, \dots, \vec{x}_k) = \vec{x}_\ell$$
where $\ell$ is such that $i \in N_\ell$. By the symmetry of the instance, any fixed point $(\vec{x}_1, \dots, \vec{x}_k)$ of this correspondence yields a symmetric PBNE $\psi(\vec{x}_1, \dots, \vec{x}_k)$ of the auction. All the arguments work as above.

\subsection{A Reduction From the DFPA to the CFPA}

We conclude the section by presenting our reduction from the problem of
computing monotone (approximate) MBNE of the DFPA to the problem of computing
monotone approximate PBNE of the CFPA. A similar reduction was presented in
\citep{fghk24} for the IPV setting, which we generalize here.

\begin{lemma}\label{lem:discrete-to-continuous-and-back-for-existence} Given
$\delta \in (0,1)$ and an instance of the DFPA, we can construct in polynomial
time an instance of the CFPA such that for any $\varepsilon \geq 0$, we can
transform in polynomial time any monotone $\varepsilon$-approximate PBNE of the
CFPA into a monotone $(\varepsilon + \delta)$-approximate MBNE of the DFPA.
Furthermore, this reduction maps (instances of) auctions with APV (resp. SAPV, resp. $k$-GSAPV)
to auctions with APV (resp. SAPV, resp. $k$-GSAPV), and symmetric equilibria to symmetric
equilibria.
\end{lemma}

\begin{proof}
Let $\delta \in (0,1)$ and a DFPA with bidding space $B$, value spaces $V_1,
\dots, V_n$, and a joint distribution with density $f^d$ be given. Without loss
of generality, we can assume that $\delta$ satisfies the following two
conditions
\begin{equation}\label{eq:delta-choice1}
\delta < \min \left\{\left|p-q\right|: p, q \in B \cup \bigcup_i V_i, p \neq q\right\},
\end{equation}
\begin{equation}\label{eq:delta-choice2}
\max \bigcup_i V_i < 1 - \delta.
\end{equation}
It is easy to see that the first condition is without loss of generality: if
$\delta$ does not satisfy it, we can just replace $\delta$ by a smaller positive
number that does (and which can be computed efficiently). The same idea also
works for the second condition, except in the case where $1 \in \bigcup_i V_i$.
In that case, we can consider a modified instance where we use value spaces
$V_i' = (1 - \gamma) \cdot V_i$ for some sufficiently small $\gamma > 0$ and
pick $\delta$ sufficiently small so that both conditions are satisfied. It is
easy to check that an $\varepsilon$-approximate PBNE for the modified instance
yields an $(\varepsilon + \gamma)$-approximate PBNE for the original instance.
Since we can make $\gamma$ and $\delta$ as small as needed, both conditions are
without loss of generality.

We construct a CFPA with bidding space $B$ and with joint distribution given by
density function $f^c$. This density function over $[0,1]^n$ is defined as 
$$f^c(\vec{x}) = \sum_{\vec{v} \in \vec{V}} w_{\vec{v}} \cdot \mathbbm{1}_{R_{\vec{v}}}(\vec{x}),$$ 
where 
\begin{equation*}
    R_{\vec{v}} = [v_1,v_1+\delta] \times \dots \times [v_n, v_n + \delta]
    \qquad 
    \text{and}
    \qquad
    w_{\vec{v}} = f^d(\vec{v})/\delta^n.
\end{equation*}
Observe that:
\begin{itemize}[leftmargin=*]
    \item[-] The density $f^c$ is well-defined. In particular, we have
    $R_{\vec{v}} \subseteq [0,1]^n$ for all $\vec{v} \in \vec{V}$ by Condition
    \eqref{eq:delta-choice2}.
    \item[-] Disjoint hyperrectangles: By Condition~\eqref{eq:delta-choice1} the
    hyperrectangles $R_{\vec{v}}$ and $R_{\vec{v'}}$ do not overlap for any
    distinct $\vec{v}, \vec{v'} \in \vec{V}$.
    \item[-] If $f^d$ satisfies the affiliation condition
    \eqref{eq:affiliation}, then so does $f^c$. Indeed, consider any $\vec{x},
    \vec{x'} \in [0,1]^n$. If $f^c(\vec{x}) = 0$ or $f^c(\vec{x'}) = 0$, then
    $\vec{x}$ and $\vec{x'}$ trivially satisfy the affiliation condition. If the
    density is not zero at any of those two points, then it must be that
    $\vec{x} \in R_{\vec{v}}$ and $\vec{x'} \in R_{\vec{v'}}$ for some $\vec{v},
    \vec{v'} \in \vec{V}$. It follows that $\vec{x} \vee \vec{x'} \in R_{\vec{v}
    \vee \vec{v'}}$ and $\vec{x} \wedge \vec{x'} \in R_{\vec{v} \wedge
    \vec{v'}}$. Now, it is easy to see that $\vec{x}$ and $\vec{x'}$ must
    satisfy the affiliation condition for $f^c$, because $\vec{v}$ and
    $\vec{v'}$ satisfy the condition for $f^d$, and because distinct
    hyperrectangles cannot overlap.
    \item[-] If $f^d$ is (group) symmetric, then so is $f^c$. This immediately follows from the construction of $f^c$.
    \item[-] The density $f^c$ can be represented efficiently. In the case where
    $f^d$ is not symmetric, it is represented by a list of the elements in its
    support along with the corresponding probabilities. Then, $f^c$ will be
    represented by a list of hyperrectangles and corresponding weights.
    Importantly, we only need to list the hyperrectangles with non-zero weight,
    i.e., as many as the size of the support of $f^d$. In the case where $f^d$
    is (group) symmetric, only the elements of the support that lie in $\vec{V}_{\geq}$ are listed, along with
    their probabilities (see \cref{sec:representation-dfpa}). Then, $f^c$ will also be represented with similar
    succinctness. Namely, according to the succinct representation for (group) symmetric
    instances, it suffices to list the hyperrectangles $R_{\vec{v}}$ with
    $\vec{v} \in \vec{V}_\geq$ that have non-zero weight.
\end{itemize}
Now let $\varepsilon \geq 0$ and consider any monotone $\varepsilon$-approximate
PBNE $\vec{\beta}^c$ of the CFPA. We construct a corresponding \emph{mixed}
strategy profile $\vec{\beta}^d$ in the DFPA as follows. For any bidder $i \in
N$ and any value $v_i \in V_i$, let $\beta^d_i(v_i)$ be the distribution of
$\beta^c_i(x_i)$ where $x_i$ is drawn uniformly at random from $[v_i, v_i +
\delta]$. Note that $\beta^d_i$ is non-overbidding, since $\beta^c_i$ is
non-overbidding and Condition~\eqref{eq:delta-choice1} ensures that $\min \{b
\in B: b > v_i\} > v_i + \delta$ for all $v_i \in V_i$. Furthermore, the
monotonicity of $\beta^c_i$, together with $\min \{|v_i - v_i'|: v_i, v_i' \in
V_i, v_i \neq v_i'\} > \delta$ (by Condition~\eqref{eq:delta-choice1}) implies
that $\beta^d_i$ is also monotone. Finally, if the instance is (group) symmetric and
$\vec{\beta}^c$ is symmetric, then so is $\vec{\beta}^d$.

Fix some bidder $i \in N$ and value $v_i \in \support{F_i}$. Then, the
construction of $\vec{\beta}^d$ from $\vec{\beta}^c$ ensures that the following
two distributions over $B^{n-1}$ are the same:
\begin{itemize}
    \item[-] Draw $\vec{v}_{-i} \in \vec{V}_{-i}$ according to the conditional
    distribution $f^d_{i|v_i}$, and then, for each $j \in N \setminus \{i\}$,
    (independently) draw $b_j \in B$ according to the distribution
    $\beta^d_j(v_j)$.
    \item[-] Draw $\vec{x}_{-i} \in [0,1]^{n-1}$ according to the conditional
    distribution $f^c_{i|v_i}$, and then, for each $j \in N \setminus \{i\}$,
    output $b_j := \beta^c_j(x_j)$.
\end{itemize}
Furthermore, this remains true if in the second distribution, we replace $v_i$
by any $x_i \in [v_i,v_i+\delta]$, since the corresponding conditional
distributions are identical, i.e., $f^c_{i|v_i} = f^c_{i|x_i}$. From this, we
deduce that, for any bidder $i \in N$ and any value $v_i \in \support{F_i}$,
\begin{equation}\label{eq:equivalent-distributions}
H^d_i(b,\vec{\beta}^d_{-i};v_i) = H^c_i(b,\vec{\beta}^c_{-i};x_i) \quad \forall b \in B, \forall x_i \in [v_i, v_i + \delta]
\end{equation}
where, recall, that the function $H^c_i$ (resp.\ $H^d_i$) denotes the
probability of winning the item in the CFPA (resp.\ DFPA), given the bid, the
strategy profile of the other bidders, and the value.

It remains to prove that $\vec{\beta}^d$ is an $(\varepsilon + \delta)$-MBNE of
the DFPA. For any bidder $i \in N$, value $v_i \in \support{F_i}$, and bid $b
\in B$, we can write, using~\cref{eq:equivalent-distributions},
\begin{equation*}
\begin{split}
u^d_i(b, \vec{\beta}^d_{-i};v_i) = (v_i - b) \cdot H^d_i(b, \vec{\beta}^d_{-i};v_i) &= (v_i - b) \cdot H^c_i(b, \vec{\beta}^c_{-i};x_i)\\
&= u^c_i(b, \vec{\beta}^c_{-i};x_i) + (v_i - x_i) \cdot H^c_i(b, \vec{\beta}^c_{-i};x_i)
\end{split}
\end{equation*}
for all $x_i \in [v_i, v_i + \delta]$. Consider any $b^* \in B$ with
$\beta^d_i(v_i)(b^*) > 0$, and note that by construction of $\beta^d_i$ there
must exist $x_i^* \in [v_i,v_i+\delta]$ such that $\beta^c_i(x_i^*) = b^*$. Now,
we can write, for any alternative $b \in B$,
\begin{equation*}
\begin{split}
&\quad u^d_i(b,\vec{\beta}^d_{-i};v_i) - u^d_i(b^*,\vec{\beta}^d_{-i};v_i)\\
&= u^c_i(b, \vec{\beta}^c_{-i};x_i^*) - u^c_i(b^*, \vec{\beta}^c_{-i};x_i^*) + (v_i - x_i^*) \cdot (H^c_i(b, \vec{\beta}^c_{-i};x_i^*) - H^c_i(b^*, \vec{\beta}^c_{-i};x_i^*))\\
&\leq u^c_i(b, \vec{\beta}^c_{-i};x_i^*) - u^c_i(b^*, \vec{\beta}^c_{-i};x_i^*) + |v_i - x_i^*|\\
&\leq u^c_i(b, \vec{\beta}^c_{-i};x_i^*) - u^c_i(b^*, \vec{\beta}^c_{-i};x_i^*) + \delta\\
&\leq \varepsilon + \delta
\end{split}
\end{equation*}
where we used the fact that $b^* = \beta^c_i(x_i^*)$ is an
$\varepsilon$-best-response to $\vec{\beta}^c_{-i}$ at value $x_i^*$ in the
CFPA. As a result, we have shown that any $b^* \in B$ with $\beta^d_i(v_i)(b^*)
> 0$ is an $(\varepsilon + \delta)$-best-response to $\vec{\beta}^d_{-i}$ at
value $v_i$ in the DFPA. It follows that $\vec{\beta}^d$ is an $(\varepsilon +
\delta)$-MBNE of the DFPA.
\end{proof}

\section{NP-hardness of Computing PBNE for Correlated Priors} \label{sec:np-hardness}
    
Motivated by the non-existence result of \cref{th:tight_existence_approx}, we
continue by studying the computational problem of deciding the existence of a
PBNE in a DFPA with discrete, correlated priors. In fact, in this section we
show that the problem of deciding the existence of an exact PBNE is NP-hard. At
the same time, the problem of deciding the existence of an
$\varepsilon$-approximate PBNE remains NP-hard for $\varepsilon$
inverse-polynomial in the size of the input. Our proof is via a reduction from a
version of the satisfiability problem. We will first provide a high-level
overview of our techniques and then a complete proof of the theorem at the end
of the section.

Previous work on the topic has established an NP-hardness result for the case of
subjective prior distributions \citep{fghk24}, where each bidder can have their
own, independent beliefs about the values of the others. Intuitively, the way
the reduction from satisfiability works in that case is that new bidders are
introduced to the auction for each operator of the boolean formula, the
strategies of which are only affected by the bidders corresponding to the inputs
of that operator; this is because in the subjective priors setting, the bidders'
beliefs are independent and each bidder can have their own beliefs about each
other bidder. This is achieved by introducing bidders with a prior belief of
value $0$ for the item for any bidder that we would like to not affect their
best response, forcing them to always bid $0$ at any equilibrium due to the
no-overbidding assumption. In our case, since the distribution of the values is
joint, this cannot be achieved.

In our setting, we can still construct a joint distribution that only contains
points of positive mass in which the only bidders appearing with a positive
value are involved as inputs or outputs of the same operator. However, a bidder
can appear with positive value while being the output of one operator and the
input of another, which makes it difficult to reason about her best response
with respect to both. To overcome this obstacle, it helps to think about the
evaluation of the SAT instance in levels, where first any negations to the
variables are applied, and then these are followed by at most $2$ OR operations
per clause. This allows us to introduce the idea of discounting factors
$\delta$, one for each level of the construction. The point of these discounting
factors is that they make points of the distribution that were added due to some
operator lower in the evaluation tree of the boolean formula appear with smaller
probability, such that when a bidder that appears as the output of one operator
and input of another computes her conditional distribution to find her best
response, she is primarily affected by the former operator. This allows us to
simulate the evaluation of the formula and then embed a counterexample for the
non-existence of a PBNE in the output of all these clauses to reduce the problem
of deciding if there is a satisfying assignment to the boolean formula to the
problem of deciding the existence of a PBNE in the DFPA with correlated priors,
yielding the following theorem:

\begin{theorem}\label{thm:NP-completeness} There exists an $\varepsilon$ of size
    inverse-polynomial to the problem description such that, for all
    $\varepsilon'<\varepsilon$, the problem of deciding the existence of an
    $\varepsilon'$-PBNE of a DFPA with correlated priors is (strongly) NP-hard. 
\end{theorem}

\cref{fig:construction} shows a high-level view of the construction
for an example of a boolean formula with three variables. The first layer
contains the input bidders $i_{x_1},i_{x_2}$ and $i_{x_3}$, each of which is
copied three times, as they might appear in up to three clauses (represented by
bidders denoted as $j$, $k$, and $\ell$, and indexed by the variable name). The
SAT operators are simulated in layers, first the NOT operators (simulated by a
combination of NOT bidders and Projection bidders) and then two layers of OR
operators. The output bidders encode the example that shows the non-existence of
equilibrium, ensuring that an equilibrium exists if and only if the formula is
satisfiable.   \medskip 

We use the remaining of this section to provide the complete proof
of~\cref{thm:NP-completeness}. We begin by specifying the version of
satisfiability that we will reduce from.

\paragraph{The \sat problem}
For the purpose of our reduction, we will be using a variant of the classic
satisfiability problem, which will make our analysis simpler. The \sat problem
is a restriction of the classic \textsc{3-SAT} problem to instances in which
each clause can have either $2$ or $3$ variables, and each variable occurs at
most $3$ times in the formula.

\begin{definition}[\sat]\label{def:2/3,3-sat}

An instance of \sat consists of a set of variables $X = \{ x_1, x_2, \ldots, x_n
\}$ which can take values in $\{ 0,1 \}$ and a set of clauses $C = \{ C_1, C_2,
\ldots, C_m \}$. Each clause $C_i$ can be either of the form $\{y_1^i, y_2^i\}$
or $\{y_1^i,y_2^i,y_3^i\}$ where $y_1^i,y_2^i,y_3^i \in \{ x_j, \bar{x_j} \}_{j
\in [n]}$. Let $\phi:\{0,1\}^n\rightarrow\{0,1\}$ denote the function evaluating
an instance of \sat, given an assignment to its variables $X$. A yes instance of
the decision problem is one in which there is a $\vec{z} \in \{0,1\}^n$ for
which $\phi(\vec{z})=1$.

\end{definition}

\begin{theorem}[\cite{tovey_1984}]
   The \sat problem is NP-complete.
\end{theorem}

\subsection{Construction of the Auction}

To prove the NP-hardness result, we will provide a reduction from the \sat
problem to the problem of deciding the existence of a PBNE in a DFPA with
correlated priors. Consider an instance $(X,C)$ of \sat. We will describe how to
create a DFPA instance that encodes the \sat instance.

Consider a DFPA with bidding space $B = \{0, b_1 = 1/7, b_2 = 2/7, b_3 = 3/7\}$
and a common value space for all bidders $V = \{ 0,23/64,1 \}$. For ease of
notation, we represent each bidder's strategy by a tuple
$(\beta(0),\beta(23/64),\beta(1))$ representing the bid that the bidder plays
for each value they can have. Due to the no-overbidding assumption, we also know
that $\beta(0)=0$ at any equilibrium.

The logical values of \emph{false} and \emph{true} will be encoded by a bidder's
strategy as follows:
\begin{itemize}
    \item[-] \emph{false} is encoded by the bidding strategy $s_0:=(0,b_1,b_2)$;
    \item[-] \emph{true} is encoded by the bidding strategy $s_1:=(0,b_2,b_3)$.
\end{itemize}

We will now describe the bidders that will participate in the auction, along
with the joint distribution induced. We will be gradually adding bidders to the
auction, along with mass for specific points of the joint distribution. Since
the number of bidders depends on the \sat instance, the points of mass added to
the distribution at each step will have value $0$ for any bidders that are not
directly mentioned in that step. As our construction is computable in polynomial
time, we can think of this as specifying efficiently what the joint distribution
will look like, and then constructing it all at once at the end, when we know
exactly how many bidders there are. This idea is also useful in defining a valid
distribution -- as we are adding more mass depending on the size of the
instance, we need to make sure that at the end the sum of the probabilities of
all points in the joint distribution is $1$. To do so, we will normalize
everything by multiplying the mass of each new point added by
$\frac{1}{\Delta}$, where $\Delta$ is chosen at the end of the reduction to be
the sum of all the total mass added in the previous steps of the construction.
We will now demonstrate how to introduce bidders and points of mass to the joint
distribution depending on the \sat instance.

\paragraph{Input bidders.}

The purpose of the input bidders is to make sure that all (at most three)
appearances of the same variable represent the same boolean value. To achieve
this, for each variable $x \in X$ of the \sat instance, we will introduce $4$
bidders to the auction, $i,j,k,\ell$, along with the following points of mass to
the joint distribution:

\begin{itemize}
    \item[-] $(0,\ldots, v_j=23/64, \ldots, 0)$ with probability $\frac{33}{128\Delta}$
    \item[-] $(0,\ldots, v_k=23/64, \ldots, 0)$ with probability $\frac{33}{128\Delta}$
    \item[-] $(0,\ldots, v_{\ell}=23/64, \ldots, 0)$ with probability $\frac{33}{128\Delta}$
    \item[-] $(0,\ldots, v_i=23/64, v_j = 23/64, \ldots, 0)$ with probability $\frac{2}{\Delta}$
    \item[-] $(0,\ldots, v_i=23/64, v_k = 23/64, \ldots, 0)$ with probability $\frac{2}{\Delta}$
    \item[-] $(0,\ldots, v_i=23/64, v_{\ell} = 23/64, \ldots, 0)$ with probability $\frac{2}{\Delta}$
    \item[-] $(0,\ldots, v_i=23/64, v_j = 1 \ldots, 0)$ with probability $\frac{1}{\Delta}$
    \item[-] $(0,\ldots, v_i=23/64, v_k = 1 \ldots, 0)$ with probability $\frac{1}{\Delta}$
    \item[-] $(0,\ldots, v_i=23/64, v_{\ell} = 1 \ldots, 0)$ with probability $\frac{1}{\Delta}$
\end{itemize}

where the notation $(0,\ldots,v_i=v,\ldots,0)$ means that at this point all
bidders other than $i$ have value $0$.

\paragraph{\NOT bidders.}
For each negated literal in a clause, we add a bidder $j$ to the auction, along
with the following points of mass to the joint distribution:

\begin{itemize}
    \item[-] $(0,\ldots, v_j=23/64, \ldots, 0)$ with probability $\frac{33\delta_{\NOT}}{256\Delta}$
    \item[-] $(0,\ldots, v_j=1, \ldots, 0)$ with probability $\frac{\delta_{\NOT}}{\Delta}$
    \item[-] $(0,\ldots, v_i=1, v_j=23/64, \ldots, 0)$ with probability $\frac{\delta_{\NOT}}{\Delta}$
    \item[-] $(0,\ldots, v_i=1, v_j=1, \ldots, 0)$ with probability $\frac{\delta_{\NOT}}{\Delta}$
\end{itemize}

where $i$ is the bidder corresponding to the variable that was negated, and
$\delta_{\NOT}$ is a constant that we will pick at the end of the reduction.

\paragraph{Projection bidders.}
For each ``NOT bidder'' $j$, we add another bidder $k$, along with the following
points of mass to the joint distribution:
\begin{itemize}
    \item[-] $(0,\ldots, v_k=23/64, \ldots, 0)$ with probability $\frac{33\delta_{\PROJ}}{256\Delta}$
    \item[-] $(0,\ldots, v_j=23/64, v_k=23/64, \ldots, 0)$ with probability $\frac{\delta_{\PROJ}}{\Delta}$
    \item[-] $(0,\ldots, v_j=23/64, v_k=1, \ldots, 0)$ with probability $\frac{\delta_{\PROJ}}{\Delta}$ 
\end{itemize}
where $\delta_{\PROJ}$ is a constant that we will pick at the end of the reduction.

\paragraph{$\OR_1$ bidders.}
For each clause that has $3$ literals with corresponding bidders $i,j,k$ (if a
literal is negated, the corresponding bidder is the projection bidder of the
appropriate NOT bidder, otherwise it is just one of the input bidders for the
specific variable), we introduce a bidder $\ell$, along with the following
points of mass to the joint distribution:

\begin{itemize}
    \item[-] $(0,\ldots, v_{\ell}=23/64, \ldots, 0)$ with probability $\frac{\delta_{\OR_1}}{128\Delta}$
    \item[-] $(0,\ldots, v_i=23/64, v_{\ell} = 23/64, \ldots, 0)$ with probability $\frac{\delta_{\OR_1}}{\Delta}$
    \item[-] $(0,\ldots, v_i=23/64, v_{\ell} = 1 \ldots, 0)$ with probability $\frac{\delta_{\OR_1}}{\Delta}$
    \item[-] $(0,\ldots, v_j=23/64, v_{\ell} = 23/64 \ldots, 0)$ with probability $\frac{\delta_{\OR_1}}{\Delta}$
    \item[-] $(0,\ldots, v_j=23/64, v_{\ell}=1, \ldots, 0)$ with probability $\frac{\delta_{\OR_1}}{\Delta}$
\end{itemize}

where $\delta_{\OR_1}$ is a constant that we will pick at the end of the reduction.

\paragraph{$\OR_2$ bidders.}
For each clause with $2$ literals with corresponding bidders $i,j$ (taking into
account negation and projection as above) and for each clause with $3$ literals
$k_1,k_2,j$ for which the corresponding $\OR_1$ bidder is $i$, we introduce a
bidder $\ell$, along with the following points of mass to the joint
distribution:
\begin{itemize}
    \item[-] $(0,\ldots, v_{\ell}=23/64, \ldots, 0)$ with probability $\frac{\delta_{\OR_2}}{128\Delta}$
    \item[-] $(0,\ldots, v_i=23/64, v_{\ell} = 23/64, \ldots, 0)$ with probability $\frac{\delta_{\OR_2}}{\Delta}$
    \item[-] $(0,\ldots, v_i=23/64, v_{\ell} = 1 \ldots, 0)$ with probability $\frac{\delta_{\OR_2}}{\Delta}$
    \item[-] $(0,\ldots, v_j=23/64, v_{\ell} = 23/64 \ldots, 0)$ with probability $\frac{\delta_{\OR_2}}{\Delta}$
    \item[-] $(0,\ldots, v_j=23/64, v_{\ell}=1, \ldots, 0)$ with probability $\frac{\delta_{\OR_2}}{\Delta}$
\end{itemize}
where $\delta_{\OR_2}$ is a constant that we will pick at the end of the
reduction.

\paragraph{Output bidders.}
For each $\delta_{\OR_2}$ bidder $i$, we introduce bidders $k,\ell$, along with
the following points of mass to the joint distribution:
\begin{itemize}
    \item[-] $(0,\ldots, v_k=1, \ldots, 0)$ with probability $\frac{\delta_{\OUT}}{\Delta}$
    \item[-] $(0,\ldots, v_{\ell}=1, \ldots, 0)$ with probability $\frac{\delta_{\OUT}}{\Delta}$
    \item[-] $(0,\ldots, v_i=23/64, v_k = 1, v_{\ell}=1 \ldots, 0)$ with probability $\frac{3\delta_{\OUT}}{4\Delta}$
\end{itemize} 
where $\delta_{\OUT}$ is a constant that we will pick at the end of the reduction.

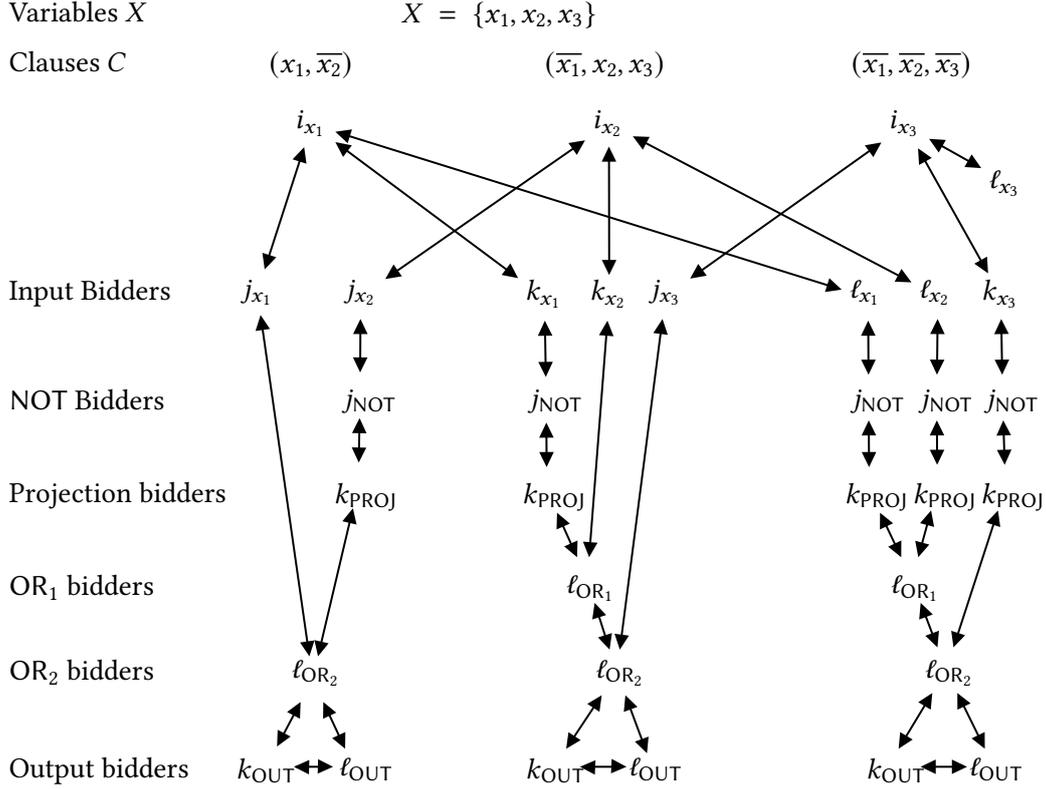
\begin{figure}[ht!]
    \begin{center}

\tikzset{every picture/.style={line width=0.75pt}} 

\begin{tikzpicture}[x=0.75pt,y=0.75pt,yscale=-0.85,xscale=0.85]

\draw (6,5.4) node [anchor=north west][inner sep=0.75pt]   [align=left] {Variables $X$};
\draw (237,5.4) node [anchor=north west][inner sep=0.75pt]    {$X\ =\ \{x_{1} ,x_{2} ,x_{3}\}$};
\draw (6,170) node [anchor=north west][inner sep=0.75pt]   [align=left] {Input Bidders};
\draw (174.33,70) node [anchor=north west][inner sep=0.75pt]    {$i_{x_1}$};
\draw (142,170) node [anchor=north west][inner sep=0.75pt]    {$j_{x_1}$};
\draw (309.67,170) node [anchor=north west][inner sep=0.75pt]    {$k_{x_1}$};
\draw (499.67,170) node [anchor=north west][inner sep=0.75pt]    {$\ell _{x_1}$};
\draw (349,70) node [anchor=north west][inner sep=0.75pt]    {$i_{x_2}$};
\draw (202.67,170) node [anchor=north west][inner sep=0.75pt]    {$j_{x_2}$};
\draw (347.67,170) node [anchor=north west][inner sep=0.75pt]    {$k_{x_2}$};
\draw (540.67,170) node [anchor=north west][inner sep=0.75pt]    {$\ell _{x_2}$};
\draw (523,70) node [anchor=north west][inner sep=0.75pt]    {$i_{x_3}$};
\draw (381.67,170) node [anchor=north west][inner sep=0.75pt]    {$j_{x_3}$};
\draw (577.67,170) node [anchor=north west][inner sep=0.75pt]    {$k_{x_3}$};
\draw (581.67,105.4) node [anchor=north west][inner sep=0.75pt]    {$\ell _{x_3}$};
\draw (6,35) node [anchor=north west][inner sep=0.75pt]   [align=left] {Clauses $C$};
\draw (158,35) node [anchor=north west][inner sep=0.75pt]    {$\left( x_{1} ,\overline{x_{2}}\right)$};
\draw (320,35) node [anchor=north west][inner sep=0.75pt]    {$\left(\overline{x_{1}} ,x_{2} ,x_{3}\right)$};
\draw (500,35) node [anchor=north west][inner sep=0.75pt]    {$\left(\overline{x_{1}} ,\overline{x_{2}} ,\overline{x_{3}}\right)$};
\draw (6,235) node [anchor=north west][inner sep=0.75pt]   [align=left] {\NOT Bidders};
\draw (201.33,235) node [anchor=north west][inner sep=0.75pt]    {$j_{\NOT}$};
\draw (310.67,235) node [anchor=north west][inner sep=0.75pt]    {$j_{\NOT}$};
\draw (500,235) node [anchor=north west][inner sep=0.75pt]    {$j_{\NOT}$};
\draw (540,235) node [anchor=north west][inner sep=0.75pt]    {$j_{\NOT}$};
\draw (578.67,235) node [anchor=north west][inner sep=0.75pt]    {$j_{\NOT}$};
\draw (6,290) node [anchor=north west][inner sep=0.75pt]   [align=left] {Projection bidders};
\draw (197.33,290) node [anchor=north west][inner sep=0.75pt]    {$k_{\PROJ}$};
\draw (308,290) node [anchor=north west][inner sep=0.75pt]    {$k_{\PROJ}$};
\draw (497.33,290) node [anchor=north west][inner sep=0.75pt]    {$k_{\PROJ}$};
\draw (537.33,290) node [anchor=north west][inner sep=0.75pt]    {$k_{\PROJ}$};
\draw (577.33,290) node [anchor=north west][inner sep=0.75pt]    {$k_{\PROJ}$};
\draw (6,345) node [anchor=north west][inner sep=0.75pt]   [align=left] {$\OR_1$ bidders};
\draw (172,395) node [anchor=north west][inner sep=0.75pt]    {$\ell _{\OR_2}$};
\draw (333.33,345) node [anchor=north west][inner sep=0.75pt]    {$\ell _{\OR_1}$};
\draw (524,345) node [anchor=north west][inner sep=0.75pt]    {$\ell _{\OR_1}$};
\draw (6,395) node [anchor=north west][inner sep=0.75pt]   [align=left] {$\OR_2$ bidders};
\draw (350.67,395) node [anchor=north west][inner sep=0.75pt]    {$\ell _{\OR_2}$};
\draw (544,395) node [anchor=north west][inner sep=0.75pt]    {$\ell _{\OR_2}$};
\draw (6,453) node [anchor=north west][inner sep=0.75pt]   [align=left] {Output bidders};
\draw (140,453) node [anchor=north west][inner sep=0.75pt]    {$k_{\OUT}$};
\draw (200,453) node [anchor=north west][inner sep=0.75pt]    {$\ell _{\OUT}$};
\draw (310,453) node [anchor=north west][inner sep=0.75pt]    {$k_{\OUT}$};
\draw (370,453) node [anchor=north west][inner sep=0.75pt]    {$\ell _{\OUT}$};
\draw (510,453) node [anchor=north west][inner sep=0.75pt]    {$k_{\OUT}$};
\draw (570,453) node [anchor=north west][inner sep=0.75pt]    {$\ell _{\OUT}$};
\draw    (158.74,162.81) -- (180.09,96.19) ;
\draw [shift={(181.01,93.33)}, rotate = 107.77] [fill={rgb, 255:red, 0; green, 0; blue, 0 }  ][line width=0.08]  [draw opacity=0] (8.93,-4.29) -- (0,0) -- (8.93,4.29) -- cycle    ;
\draw [shift={(157.83,165.67)}, rotate = 287.77] [fill={rgb, 255:red, 0; green, 0; blue, 0 }  ][line width=0.08]  [draw opacity=0] (8.93,-4.29) -- (0,0) -- (8.93,4.29) -- cycle    ;
\draw    (201.73,92.14) -- (304.27,169.07) ;
\draw [shift={(306.67,170.87)}, rotate = 216.88] [fill={rgb, 255:red, 0; green, 0; blue, 0 }  ][line width=0.08]  [draw opacity=0] (8.93,-4.29) -- (0,0) -- (8.93,4.29) -- cycle    ;
\draw [shift={(199.33,90.34)}, rotate = 36.88] [fill={rgb, 255:red, 0; green, 0; blue, 0 }  ][line width=0.08]  [draw opacity=0] (8.93,-4.29) -- (0,0) -- (8.93,4.29) -- cycle    ;
\draw    (493.8,176.73) -- (202.2,85.13) ;
\draw [shift={(199.33,84.23)}, rotate = 17.44] [fill={rgb, 255:red, 0; green, 0; blue, 0 }  ][line width=0.08]  [draw opacity=0] (8.93,-4.29) -- (0,0) -- (8.93,4.29) -- cycle    ;
\draw [shift={(496.67,177.63)}, rotate = 197.44] [fill={rgb, 255:red, 0; green, 0; blue, 0 }  ][line width=0.08]  [draw opacity=0] (8.93,-4.29) -- (0,0) -- (8.93,4.29) -- cycle    ;
\draw    (343.54,91.01) -- (231.13,169.64) ;
\draw [shift={(228.67,171.36)}, rotate = 325.03] [fill={rgb, 255:red, 0; green, 0; blue, 0 }  ][line width=0.08]  [draw opacity=0] (8.93,-4.29) -- (0,0) -- (8.93,4.29) -- cycle    ;
\draw [shift={(346,89.29)}, rotate = 145.03] [fill={rgb, 255:red, 0; green, 0; blue, 0 }  ][line width=0.08]  [draw opacity=0] (8.93,-4.29) -- (0,0) -- (8.93,4.29) -- cycle    ;
\draw    (360.03,96) -- (360.14,166) ;
\draw [shift={(360.14,169)}, rotate = 269.91] [fill={rgb, 255:red, 0; green, 0; blue, 0 }  ][line width=0.08]  [draw opacity=0] (8.93,-4.29) -- (0,0) -- (8.93,4.29) -- cycle    ;
\draw [shift={(360.02,93)}, rotate = 89.91] [fill={rgb, 255:red, 0; green, 0; blue, 0 }  ][line width=0.08]  [draw opacity=0] (8.93,-4.29) -- (0,0) -- (8.93,4.29) -- cycle    ;
\draw    (376.65,88.29) -- (535.01,171.91) ;
\draw [shift={(537.67,173.32)}, rotate = 207.84] [fill={rgb, 255:red, 0; green, 0; blue, 0 }  ][line width=0.08]  [draw opacity=0] (8.93,-4.29) -- (0,0) -- (8.93,4.29) -- cycle    ;
\draw [shift={(374,86.89)}, rotate = 27.84] [fill={rgb, 255:red, 0; green, 0; blue, 0 }  ][line width=0.08]  [draw opacity=0] (8.93,-4.29) -- (0,0) -- (8.93,4.29) -- cycle    ;
\draw    (517.56,92.29) -- (410.1,169.35) ;
\draw [shift={(407.67,171.1)}, rotate = 324.35] [fill={rgb, 255:red, 0; green, 0; blue, 0 }  ][line width=0.08]  [draw opacity=0] (8.93,-4.29) -- (0,0) -- (8.93,4.29) -- cycle    ;
\draw [shift={(520,90.54)}, rotate = 144.35] [fill={rgb, 255:red, 0; green, 0; blue, 0 }  ][line width=0.08]  [draw opacity=0] (8.93,-4.29) -- (0,0) -- (8.93,4.29) -- cycle    ;
\draw    (542.97,96.62) -- (581.2,165.38) ;
\draw [shift={(582.66,168)}, rotate = 240.92] [fill={rgb, 255:red, 0; green, 0; blue, 0 }  ][line width=0.08]  [draw opacity=0] (8.93,-4.29) -- (0,0) -- (8.93,4.29) -- cycle    ;
\draw [shift={(541.51,94)}, rotate = 60.92] [fill={rgb, 255:red, 0; green, 0; blue, 0 }  ][line width=0.08]  [draw opacity=0] (8.93,-4.29) -- (0,0) -- (8.93,4.29) -- cycle    ;
\draw    (550.61,89.89) -- (576.05,104.27) ;
\draw [shift={(578.67,105.74)}, rotate = 209.47] [fill={rgb, 255:red, 0; green, 0; blue, 0 }  ][line width=0.08]  [draw opacity=0] (8.93,-4.29) -- (0,0) -- (8.93,4.29) -- cycle    ;
\draw [shift={(548,88.41)}, rotate = 29.47] [fill={rgb, 255:red, 0; green, 0; blue, 0 }  ][line width=0.08]  [draw opacity=0] (8.93,-4.29) -- (0,0) -- (8.93,4.29) -- cycle    ;
\draw    (214.07,198) -- (213.93,223.67) ;
\draw [shift={(213.91,226.67)}, rotate = 270.33] [fill={rgb, 255:red, 0; green, 0; blue, 0 }  ][line width=0.08]  [draw opacity=0] (8.93,-4.29) -- (0,0) -- (8.93,4.29) -- cycle    ;
\draw [shift={(214.09,195)}, rotate = 90.33] [fill={rgb, 255:red, 0; green, 0; blue, 0 }  ][line width=0.08]  [draw opacity=0] (8.93,-4.29) -- (0,0) -- (8.93,4.29) -- cycle    ;
\draw    (322.44,199) -- (322.89,226.33) ;
\draw [shift={(322.94,229.33)}, rotate = 269.05] [fill={rgb, 255:red, 0; green, 0; blue, 0 }  ][line width=0.08]  [draw opacity=0] (8.93,-4.29) -- (0,0) -- (8.93,4.29) -- cycle    ;
\draw [shift={(322.39,196)}, rotate = 89.05] [fill={rgb, 255:red, 0; green, 0; blue, 0 }  ][line width=0.08]  [draw opacity=0] (8.93,-4.29) -- (0,0) -- (8.93,4.29) -- cycle    ;
\draw    (512.26,199) -- (512.41,225) ;
\draw [shift={(512.42,228)}, rotate = 269.68] [fill={rgb, 255:red, 0; green, 0; blue, 0 }  ][line width=0.08]  [draw opacity=0] (8.93,-4.29) -- (0,0) -- (8.93,4.29) -- cycle    ;
\draw [shift={(512.24,196)}, rotate = 89.68] [fill={rgb, 255:red, 0; green, 0; blue, 0 }  ][line width=0.08]  [draw opacity=0] (8.93,-4.29) -- (0,0) -- (8.93,4.29) -- cycle    ;
\draw    (552.98,198) -- (552.68,225) ;
\draw [shift={(552.65,228)}, rotate = 270.64] [fill={rgb, 255:red, 0; green, 0; blue, 0 }  ][line width=0.08]  [draw opacity=0] (8.93,-4.29) -- (0,0) -- (8.93,4.29) -- cycle    ;
\draw [shift={(553.02,195)}, rotate = 90.64] [fill={rgb, 255:red, 0; green, 0; blue, 0 }  ][line width=0.08]  [draw opacity=0] (8.93,-4.29) -- (0,0) -- (8.93,4.29) -- cycle    ;
\draw    (590.44,198) -- (590.89,225) ;
\draw [shift={(590.94,228)}, rotate = 269.05] [fill={rgb, 255:red, 0; green, 0; blue, 0 }  ][line width=0.08]  [draw opacity=0] (8.93,-4.29) -- (0,0) -- (8.93,4.29) -- cycle    ;
\draw [shift={(590.39,195)}, rotate = 89.05] [fill={rgb, 255:red, 0; green, 0; blue, 0 }  ][line width=0.08]  [draw opacity=0] (8.93,-4.29) -- (0,0) -- (8.93,4.29) -- cycle    ;
\draw    (213.14,277) -- (213.52,256.67) ;
\draw [shift={(213.58,253.67)}, rotate = 91.07] [fill={rgb, 255:red, 0; green, 0; blue, 0 }  ][line width=0.08]  [draw opacity=0] (8.93,-4.29) -- (0,0) -- (8.93,4.29) -- cycle    ;
\draw [shift={(213.09,280)}, rotate = 271.07] [fill={rgb, 255:red, 0; green, 0; blue, 0 }  ][line width=0.08]  [draw opacity=0] (8.93,-4.29) -- (0,0) -- (8.93,4.29) -- cycle    ;
\draw    (323.4,279.67) -- (323.27,259.33) ;
\draw [shift={(323.25,256.33)}, rotate = 89.64] [fill={rgb, 255:red, 0; green, 0; blue, 0 }  ][line width=0.08]  [draw opacity=0] (8.93,-4.29) -- (0,0) -- (8.93,4.29) -- cycle    ;
\draw [shift={(323.42,282.67)}, rotate = 269.64] [fill={rgb, 255:red, 0; green, 0; blue, 0 }  ][line width=0.08]  [draw opacity=0] (8.93,-4.29) -- (0,0) -- (8.93,4.29) -- cycle    ;
\draw    (512.73,279.67) -- (512.6,258) ;
\draw [shift={(512.58,255)}, rotate = 89.65] [fill={rgb, 255:red, 0; green, 0; blue, 0 }  ][line width=0.08]  [draw opacity=0] (8.93,-4.29) -- (0,0) -- (8.93,4.29) -- cycle    ;
\draw [shift={(512.75,282.67)}, rotate = 269.65] [fill={rgb, 255:red, 0; green, 0; blue, 0 }  ][line width=0.08]  [draw opacity=0] (8.93,-4.29) -- (0,0) -- (8.93,4.29) -- cycle    ;
\draw    (552.73,278.33) -- (552.6,258) ;
\draw [shift={(552.58,255)}, rotate = 89.64] [fill={rgb, 255:red, 0; green, 0; blue, 0 }  ][line width=0.08]  [draw opacity=0] (8.93,-4.29) -- (0,0) -- (8.93,4.29) -- cycle    ;
\draw [shift={(552.75,281.33)}, rotate = 269.64] [fill={rgb, 255:red, 0; green, 0; blue, 0 }  ][line width=0.08]  [draw opacity=0] (8.93,-4.29) -- (0,0) -- (8.93,4.29) -- cycle    ;
\draw    (592.32,278.33) -- (591.68,258) ;
\draw [shift={(591.59,255)}, rotate = 88.21] [fill={rgb, 255:red, 0; green, 0; blue, 0 }  ][line width=0.08]  [draw opacity=0] (8.93,-4.29) -- (0,0) -- (8.93,4.29) -- cycle    ;
\draw [shift={(592.41,281.33)}, rotate = 268.21] [fill={rgb, 255:red, 0; green, 0; blue, 0 }  ][line width=0.08]  [draw opacity=0] (8.93,-4.29) -- (0,0) -- (8.93,4.29) -- cycle    ;
\draw    (183.66,391.7) -- (155.84,195.64) ;
\draw [shift={(155.42,192.67)}, rotate = 81.92] [fill={rgb, 255:red, 0; green, 0; blue, 0 }  ][line width=0.08]  [draw opacity=0] (8.93,-4.29) -- (0,0) -- (8.93,4.29) -- cycle    ;
\draw [shift={(184.08,394.67)}, rotate = 261.92] [fill={rgb, 255:red, 0; green, 0; blue, 0 }  ][line width=0.08]  [draw opacity=0] (8.93,-4.29) -- (0,0) -- (8.93,4.29) -- cycle    ;
\draw    (189.84,391.75) -- (208.99,309.92) ;
\draw [shift={(209.67,307)}, rotate = 103.17] [fill={rgb, 255:red, 0; green, 0; blue, 0 }  ][line width=0.08]  [draw opacity=0] (8.93,-4.29) -- (0,0) -- (8.93,4.29) -- cycle    ;
\draw [shift={(189.16,394.67)}, rotate = 283.17] [fill={rgb, 255:red, 0; green, 0; blue, 0 }  ][line width=0.08]  [draw opacity=0] (8.93,-4.29) -- (0,0) -- (8.93,4.29) -- cycle    ;
\draw    (340.08,333.26) -- (330.76,312.41) ;
\draw [shift={(329.53,309.67)}, rotate = 65.92] [fill={rgb, 255:red, 0; green, 0; blue, 0 }  ][line width=0.08]  [draw opacity=0] (8.93,-4.29) -- (0,0) -- (8.93,4.29) -- cycle    ;
\draw [shift={(341.3,336)}, rotate = 245.92] [fill={rgb, 255:red, 0; green, 0; blue, 0 }  ][line width=0.08]  [draw opacity=0] (8.93,-4.29) -- (0,0) -- (8.93,4.29) -- cycle    ;
\draw    (348.6,333.01) -- (358.9,198.99) ;
\draw [shift={(359.13,196)}, rotate = 94.39] [fill={rgb, 255:red, 0; green, 0; blue, 0 }  ][line width=0.08]  [draw opacity=0] (8.93,-4.29) -- (0,0) -- (8.93,4.29) -- cycle    ;
\draw [shift={(348.37,336)}, rotate = 274.39] [fill={rgb, 255:red, 0; green, 0; blue, 0 }  ][line width=0.08]  [draw opacity=0] (8.93,-4.29) -- (0,0) -- (8.93,4.29) -- cycle    ;
\draw    (530.35,333.29) -- (520.48,312.38) ;
\draw [shift={(519.2,309.67)}, rotate = 64.74] [fill={rgb, 255:red, 0; green, 0; blue, 0 }  ][line width=0.08]  [draw opacity=0] (8.93,-4.29) -- (0,0) -- (8.93,4.29) -- cycle    ;
\draw [shift={(531.63,336)}, rotate = 244.74] [fill={rgb, 255:red, 0; green, 0; blue, 0 }  ][line width=0.08]  [draw opacity=0] (8.93,-4.29) -- (0,0) -- (8.93,4.29) -- cycle    ;
\draw    (542.45,333.1) -- (548.38,311.23) ;
\draw [shift={(549.17,308.33)}, rotate = 105.18] [fill={rgb, 255:red, 0; green, 0; blue, 0 }  ][line width=0.08]  [draw opacity=0] (8.93,-4.29) -- (0,0) -- (8.93,4.29) -- cycle    ;
\draw [shift={(541.66,336)}, rotate = 285.18] [fill={rgb, 255:red, 0; green, 0; blue, 0 }  ][line width=0.08]  [draw opacity=0] (8.93,-4.29) -- (0,0) -- (8.93,4.29) -- cycle    ;
\draw    (352.52,365.86) -- (359.48,387.81) ;
\draw [shift={(360.39,390.67)}, rotate = 252.41] [fill={rgb, 255:red, 0; green, 0; blue, 0 }  ][line width=0.08]  [draw opacity=0] (8.93,-4.29) -- (0,0) -- (8.93,4.29) -- cycle    ;
\draw [shift={(351.61,363)}, rotate = 72.41] [fill={rgb, 255:red, 0; green, 0; blue, 0 }  ][line width=0.08]  [draw opacity=0] (8.93,-4.29) -- (0,0) -- (8.93,4.29) -- cycle    ;
\draw    (391.06,197.98) -- (366.78,387.69) ;
\draw [shift={(366.39,390.67)}, rotate = 277.29] [fill={rgb, 255:red, 0; green, 0; blue, 0 }  ][line width=0.08]  [draw opacity=0] (8.93,-4.29) -- (0,0) -- (8.93,4.29) -- cycle    ;
\draw [shift={(391.44,195)}, rotate = 97.29] [fill={rgb, 255:red, 0; green, 0; blue, 0 }  ][line width=0.08]  [draw opacity=0] (8.93,-4.29) -- (0,0) -- (8.93,4.29) -- cycle    ;
\draw    (551.88,386.52) -- (544.12,365.81) ;
\draw [shift={(543.06,363)}, rotate = 69.44] [fill={rgb, 255:red, 0; green, 0; blue, 0 }  ][line width=0.08]  [draw opacity=0] (8.93,-4.29) -- (0,0) -- (8.93,4.29) -- cycle    ;
\draw [shift={(552.94,389.33)}, rotate = 249.44] [fill={rgb, 255:red, 0; green, 0; blue, 0 }  ][line width=0.08]  [draw opacity=0] (8.93,-4.29) -- (0,0) -- (8.93,4.29) -- cycle    ;
\draw    (563.28,386.48) -- (587.56,311.19) ;
\draw [shift={(588.48,308.33)}, rotate = 107.88] [fill={rgb, 255:red, 0; green, 0; blue, 0 }  ][line width=0.08]  [draw opacity=0] (8.93,-4.29) -- (0,0) -- (8.93,4.29) -- cycle    ;
\draw [shift={(562.35,389.33)}, rotate = 287.88] [fill={rgb, 255:red, 0; green, 0; blue, 0 }  ][line width=0.08]  [draw opacity=0] (8.93,-4.29) -- (0,0) -- (8.93,4.29) -- cycle    ;
\draw    (166.48,445.34) -- (177.52,424.32) ;
\draw [shift={(178.91,421.67)}, rotate = 117.7] [fill={rgb, 255:red, 0; green, 0; blue, 0 }  ][line width=0.08]  [draw opacity=0] (8.93,-4.29) -- (0,0) -- (8.93,4.29) -- cycle    ;
\draw [shift={(165.09,448)}, rotate = 297.7] [fill={rgb, 255:red, 0; green, 0; blue, 0 }  ][line width=0.08]  [draw opacity=0] (8.93,-4.29) -- (0,0) -- (8.93,4.29) -- cycle    ;
\draw    (176,460) -- (192,460) ;
\draw [shift={(198,460)}, rotate = 181.43] [fill={rgb, 255:red, 0; green, 0; blue, 0 }  ][line width=0.08]  [draw opacity=0] (8.93,-4.29) -- (0,0) -- (8.93,4.29) -- cycle    ;
\draw [shift={(175,460)}, rotate = 1.43] [fill={rgb, 255:red, 0; green, 0; blue, 0 }  ][line width=0.08]  [draw opacity=0] (8.93,-4.29) -- (0,0) -- (8.93,4.29) -- cycle    ;
\draw    (193.52,424.39) -- (203.82,446.61) ;
\draw [shift={(205.08,449.33)}, rotate = 245.14] [fill={rgb, 255:red, 0; green, 0; blue, 0 }  ][line width=0.08]  [draw opacity=0] (8.93,-4.29) -- (0,0) -- (8.93,4.29) -- cycle    ;
\draw [shift={(192.26,421.67)}, rotate = 65.14] [fill={rgb, 255:red, 0; green, 0; blue, 0 }  ][line width=0.08]  [draw opacity=0] (8.93,-4.29) -- (0,0) -- (8.93,4.29) -- cycle    ;
\draw    (336.55,446.83) -- (354.12,420.17) ;
\draw [shift={(355.77,417.67)}, rotate = 123.39] [fill={rgb, 255:red, 0; green, 0; blue, 0 }  ][line width=0.08]  [draw opacity=0] (8.93,-4.29) -- (0,0) -- (8.93,4.29) -- cycle    ;
\draw [shift={(334.9,449.33)}, rotate = 303.39] [fill={rgb, 255:red, 0; green, 0; blue, 0 }  ][line width=0.08]  [draw opacity=0] (8.93,-4.29) -- (0,0) -- (8.93,4.29) -- cycle    ;
\draw    (346,460) -- (366,460) ;
\draw [shift={(370.33,460)}, rotate = 182.49] [fill={rgb, 255:red, 0; green, 0; blue, 0 }  ][line width=0.08]  [draw opacity=0] (8.93,-4.29) -- (0,0) -- (8.93,4.29) -- cycle    ;
\draw [shift={(343,460)}, rotate = 2.49] [fill={rgb, 255:red, 0; green, 0; blue, 0 }  ][line width=0.08]  [draw opacity=0] (8.93,-4.29) -- (0,0) -- (8.93,4.29) -- cycle    ;
\draw    (381.3,449.19) -- (370.7,420.48) ;
\draw [shift={(369.66,417.67)}, rotate = 69.72] [fill={rgb, 255:red, 0; green, 0; blue, 0 }  ][line width=0.08]  [draw opacity=0] (8.93,-4.29) -- (0,0) -- (8.93,4.29) -- cycle    ;
\draw [shift={(382.34,452)}, rotate = 249.72] [fill={rgb, 255:red, 0; green, 0; blue, 0 }  ][line width=0.08]  [draw opacity=0] (8.93,-4.29) -- (0,0) -- (8.93,4.29) -- cycle    ;
\draw    (531.64,446.76) -- (548.36,418.91) ;
\draw [shift={(549.9,416.33)}, rotate = 120.96] [fill={rgb, 255:red, 0; green, 0; blue, 0 }  ][line width=0.08]  [draw opacity=0] (8.93,-4.29) -- (0,0) -- (8.93,4.29) -- cycle    ;
\draw [shift={(530.1,449.33)}, rotate = 300.96] [fill={rgb, 255:red, 0; green, 0; blue, 0 }  ][line width=0.08]  [draw opacity=0] (8.93,-4.29) -- (0,0) -- (8.93,4.29) -- cycle    ;
\draw    (579.42,446.64) -- (565.92,419.03) ;
\draw [shift={(564.6,416.33)}, rotate = 63.95] [fill={rgb, 255:red, 0; green, 0; blue, 0 }  ][line width=0.08]  [draw opacity=0] (8.93,-4.29) -- (0,0) -- (8.93,4.29) -- cycle    ;
\draw [shift={(580.73,449.33)}, rotate = 243.95] [fill={rgb, 255:red, 0; green, 0; blue, 0 }  ][line width=0.08]  [draw opacity=0] (8.93,-4.29) -- (0,0) -- (8.93,4.29) -- cycle    ;
\draw    (567,460) -- (551,460) ;
\draw [shift={(543,460)}, rotate = 360] [fill={rgb, 255:red, 0; green, 0; blue, 0 }  ][line width=0.08]  [draw opacity=0] (8.93,-4.29) -- (0,0) -- (8.93,4.29) -- cycle    ;
\draw [shift={(570.33,460)}, rotate = 180] [fill={rgb, 255:red, 0; green, 0; blue, 0 }  ][line width=0.08]  [draw opacity=0] (8.93,-4.29) -- (0,0) -- (8.93,4.29) -- cycle    ;

\end{tikzpicture}

    \end{center}
    \caption{Outline of the construction of the DFPA from the \sat instance. The
    double arrows indicate the existence of a point in the support of the
    distribution in which both of the bidders in the arrow's endpoints have
    strictly positive value.}
    \label{fig:construction}
    \end{figure}

\subsection{Analysis}
We will show that there exists a satisfying assignment of the \sat instance if
and only if the corresponding DFPA has an $\varepsilon$-PBNE (we will compute
the value of $\varepsilon$ for which this holds at the end of the reduction).

\paragraph{Extraction.}
Let $\vec{\beta}$ be a strategy profile of the bidders of the auction. We define
an assignment $\mathbf{A}:X \rightarrow \{0,1\}$ to the variables of the \sat
instance, such that $\forall x \in X$:
\begin{itemize}
    \item [-] $\mathbf{A}[x] = 0$, if $\beta_x = (0,1/7,2/7) = s_0$
    \item [-] $\mathbf{A}[x] = 1$, if $\beta_x = (0,2/7,3/7) = s_1$
\end{itemize}
where, with slight abuse of notation, we have expressed the bidding function
$\beta_x$ using a vector of dimension $|V|$ explicitly stating the bid that $x$
plays for every value in $V$, i.e., $\beta_x =
(\beta_x(0),\beta_x(23/64),\beta_x(1))$. We also define the mapping
$\mathbf{\chi}$ from bidding strategies to $\{0,1\}$ to satisfy
$\mathbf{\chi}[s_0] = 0$ and $\mathbf{\chi}[s_1] = 1$.

In the rest of this section, we will prove that the \sat instance has a
satisfying assignment if and only if we can construct a strategy profile
$\vec{\beta}$ that is a PBNE.

\paragraph{Input bidders.}
We begin by arguing that the input bidders all have the required behaviour,
which is summarized by the following lemma:

\begin{lemma}\label{lem:input-bidder} Fix any $\delta_{\NOT}<\frac{33}{1792}$
    and $\varepsilon<\frac{1}{\Delta}(\frac{33}{896}-2\delta_{\NOT})$. Then, at
    any $\varepsilon$-PBNE $\vec{\beta}$, for any bidders $i,j,k,\ell$
    introduced by some variable $x \in X$ of the \sat instance, we have
    $\beta_j=\beta_k=\beta_{\ell} \in \{s_0,s_1\}$.
\end{lemma}
\begin{proof}
    We start with the analysis of the best-responses of $j,k,\ell$. Firstly, due
    to no-overbidding, we know that, when having value $0$, all three of them
    will bid $0$. When arguing about their best responses given some positive
    value, we need to take into account their conditional distributions. Notice
    that $j,k,\ell$ have been defined in a symmetric way and there are no points
    in the joint distribution where more than one of them has a positive value,
    so we can reason about the intended behaviour of one of them and the
    analysis follows for the other two. Thus, we will provide the analysis only
    for bidder $j$'s best response. Notice also that, for any point in the
    support of the distribution where $j$ has positive value, $i$ can only have
    value $0$ or $23/64$, so we can ignore the value of $\beta_i(1)$.
    Additionally, $\beta_i(0)=0$ due to the no-overbidding assumption. There can
    only be at most one other bidder affecting $j$'s best-response (meaning that
    there is positive probability that she has positive value at the same time
    as $j$), and this could be either a \NOT, an $\OR_1$, or an $\OR_2$ bidder.
    Let this bidder be $x$, with strategy $\beta_x$, and let the number of total
    bidders be $n$. One of the key ideas of our construction is the choice of
    appropriate ``discounting factors''
    $\delta_{\NOT},\delta_{\OR_1},\delta_{\OR_2}$ in the description of \NOT,
    $\OR_1$, and $\OR_2$ bidders, such that $j$'s best response is in fact only
    depending on $i$. We will pick the values of the discounting factors at the
    last step of the reduction, to satisfy
    $\delta_{\NOT}>\delta_{\OR_1}>\delta_{\OR_2}$ along with other properties we
    will discuss there. To find $j$'s best response, we will compute the
    expected utility $u_j(b,\vec{\beta'};v_j)$ for each $v_j \in V_j$ and $b \in
    B$ when the other bidders bid according to $\vec{\beta'}$, and output the
    strategy that for each value chooses the bid that maximizes the utility.
    Notice that here, as well as in the remaining of the analysis for the other
    bidders, when we compute the utility of $j$ when having value $v_j$ via
    means of its conditional distribution $f_{j\mid v_j}$ we use the formula
    described in \cref{eq:DFPA-utility-interim-mixed}, which sums over all
    points in the support of $f_{j\mid v_j}$. Our construction has been
    carefully designed so that at any of these points there is at most one other
    bidder with strictly positive value, who can therefore bid anything higher
    than $0$ (with the exception of output bidders; see the proof of
    \cref{lem:output} for that case). This is also demonstrated in
    \cref{fig:construction}.

    In our analysis, we will consider different cases according to $i$'s
    strategy, assuming that the remaining bidders (all but $i$ and $j$) play
    according to strategies $\vec{\beta'}$. We will split the computation of the
    expected utility of bidder $j$ when bidding $b$ with value $v_j$ and the
    others following strategy $\vec{\beta'}$ as follows:
    \begin{equation}\label{eq:split-utility}
        u_j(b,\vec{\beta'};v_j) =  u_j^{(i)}(b,\vec{\beta'};v_j) + u_j^{(x)}(b,\vec{\beta'};v_j)
    \end{equation}
    where $u_j^{(i)}$ is the expected utility function that takes into account
    only the points in the support where $i$ has positive value, while
    $u_j^{(x)}$ considers points in the support where $x$ has positive value. It
    is obvious from our construction that \cref{eq:split-utility} holds at any
    equilibrium, as there are no other points in the support where $j$ has
    positive value, and at no equilibrium will $j$ receive positive utility by
    bidding $0$ (since other bidders with positive value would have an incentive
    to place a positive bid). The key idea in the reasoning behind $j$'s best
    responses is that each point in the computation of $u_j^{(x)}$ has
    probability scaled by some discounting factor which we choose to be small
    enough such that $j$'s utility is affected primarily by $u_j^{(i)}$. Since
    $\delta_{\NOT}>\delta_{\OR_1}>\delta_{\OR_2}$, we can upper bound this
    discounting factor by $\delta_{\NOT}$. Again, given the nature of our
    construction, there are at most $2$ such points, and the maximum ex-post
    utility that a bidder can observe at a point in the support is trivially
    bounded by $1$, so we can safely bound  $u_j^{(x)}(b,\vec{\beta'};v_j) \leq
    \frac{2}{\Delta}\delta_{\NOT}$. In the remaining of this proof, we will
    analyse the value of $u_j^{i}$ for the different strategies of $i$ and $j$.
    Instead of computing $u_j^{i}$, it is more practical at this step to compute
    the quantity $\Delta \cdot u_j^{(i)}$, as $\Delta$ will be chosen at the end
    of our reduction. Notice that this does not affect our computation of
    best-responses, as $\frac{1}{\Delta}$ is a common factor to the mass of each
    point of the distribution.

    \begin{enumerate}
        \item If $\beta_i(23/64)=0$:

        \begingroup
        \renewcommand{\arraystretch}{1.2}
        \begin{center}\begin{tabular}{c|cccc}
            $b$ & $0$ & $b_1$ & $b_2$ & $b_3$\\
            \hline
            $\Delta \cdot u_j^{(i)}(b;23/64)$ & $\frac{6647}{8192n}$ & $\frac{28033}{57344}$ & $\frac{9537}{57344}$ & $-$\\
            $\Delta \cdot u_j^{(i)}(b;1)$ & $\frac{1}{n}$ & $\frac{6}{7}$ & $\frac{5}{7}$ & $\frac{4}{7}$
        \end{tabular}\end{center}
        \endgroup 

        We can see that, for $n \geq 2$, $j$ receives the highest utility by
        playing $(0,b_1,b_1)$. Additionally, for $n \geq 3$, no other strategy
        achieves utility $u_j^{(i)}$ within $\frac{1}{7}$ of the optimal one in
        the table.
        
        \item $\beta_i(23/64)=1/7$:
        
        \begingroup
        \renewcommand{\arraystretch}{1.2}
        \begin{center}\begin{tabular}{c|cccc}
            $b$ & $0$ & $b_1$ & $b_2$ & $b_3$\\
            \hline
            $\Delta \cdot u_j^{(i)}(b;23/64)$ & $\frac{759}{8192n}$ & $\frac{2231}{8192}$ & $\frac{9537}{57344}$ & $-$\\
            $\Delta \cdot u_j^{(i)}(b;1)$ & $0$ & $\frac{3}{7}$ & $\frac{5}{7}$ & $\frac{4}{7}$
        \end{tabular}\end{center}
        \endgroup 

        We can see that, for $n \geq 2$, $j$ receives the highest utility by
        playing $(0,b_1,b_2)$. Additionally, no other strategy achieves utility
        $u_j^{(i)}$ within $\frac{95}{896}$ of the optimal one in the table.
        
        \item $\beta_i(23/64)=2/7$:
        
        \begingroup
        \renewcommand{\arraystretch}{1.2}
        \begin{center}\begin{tabular}{c|cccc}
            $b$ & $0$ & $b_1$ & $b_2$ & $b_3$\\
            \hline
            $\Delta \cdot u_j^{(i)}(b;23/64)$ & $\frac{759}{8192n}$ & $\frac{3201}{57344}$ & $\frac{759}{8192}$ & $-$\\
            $\Delta \cdot u_j^{(i)}(b;1)$ & $0$ & $0$ & $\frac{5}{14}$ & $\frac{4}{7}$
        \end{tabular}\end{center}
        \endgroup 

        We can see that, for $n \geq 2$, $j$ receives the highest utility by
        playing $(0,b_2,b_3)$. Additionally, no other strategy achieves utility
        $u_j^{(i)}$ within $\frac{33}{896}$ of the optimal one in the table.
        
    \end{enumerate}

    Now let us consider $i$'s best-responses, since we are interested in the
    equilibria of the game. Once again, the construction has no points in the
    support where $v_i=1$, so we only analyse what $i$ bids when observing value
    $23/64$.

    Assume for a contradiction that $i$ plays a strategy such that
    $\beta_i(23/64)=0$ at an $\varepsilon$-PBNE. The above analysis then shows
    that bidders $j,k,\ell$ all play $(0,b_1,b_1)$. It is straightforward to see
    that $i$ gets utility $0$ in this case, since it is always the case that at
    least one other bidder plays a non-zero bid. At the same time, if $i$
    switched to a strategy such that $\beta_i(23/64)=2/7$, they would get the
    whole item and an expected utility of $\frac{45}{7\Delta}$. Therefore, $i$
    has a utility-improving unilateral deviation, contradicting the assumption
    that we started from an $\varepsilon$-PBNE where $i$ plays
    $\beta_i(23/64)=0$.

    Additionally, it is crucial to show that, when $j,k,\ell$ best-respond to
    $i$, the latter also does not have an incentive to deviate, therefore we are
    indeed at an equilibrium. For our analysis, it suffices to only check the
    cases where all three bidders $j,k,\ell$ play the same strategy, as we have
    just proved that at an equilibrium their strategies are the unique best
    responses to $i$'s strategy. Furthermore, we only examine the three
    strategies that we computed as best responses in the above computation. We
    consider the following cases:

    \begin{enumerate}
        \item $j,k,\ell$ play strategy $(0,b_1,b_1)$:
        \begingroup
        \renewcommand{\arraystretch}{1.2}
        \begin{center}\begin{tabular}{c|cccc}
            $b$ & $0$ & $b_1$ & $b_2$ & $b_3$\\
            \hline
            $\Delta \cdot u_i(b;23/64)$ & $0$ & $\frac{873}{896}$ & $\frac{297}{448}$ & $-$\\
        \end{tabular}\end{center}
        \endgroup

        We can see that, $i$'s best response is $\beta_i(23/64)=b_1$.
        Additionally, no other strategy achieves utility $u_i$ within
        $\frac{279}{896}$ of the optimal one in the table. Since $j,k,\text{and
        }\ell$'s best response when $\beta_i(23/64)=b_1=1/7$ was $(0,b_1,b_2)$,
        there cannot be an equilibrium where $j,k,\ell$ play strategy
        $(0,b_1,b_1)$ (as the 4 bidders are not simultaneously best-responding).

        \item $j,k,\ell$ play strategy $(0,b_1,b_2)$:
        
        \begingroup
        \renewcommand{\arraystretch}{1.2}
        \begin{center}\begin{tabular}{c|cccc}
            $b$ & $0$ & $b_1$ & $b_2$ & $b_3$\\
            \hline
            $\Delta \cdot u_i(b;23/64)$ & $0$ & $\frac{291}{448}$ & $\frac{495}{896}$ & $-$\\
        \end{tabular}\end{center}
        \endgroup 

        We can see that $i$'s best response is $\beta_i(23/64)=b_1$.
        Additionally, no other strategy achieves a value of $\Delta \cdot u_i$
        within $\frac{87}{896}$ of the optimal one in the table.

        \item $j,k,\ell$ play strategy $(0,b_2,b_3)$:
        
        \begingroup
        \renewcommand{\arraystretch}{1.2}
        \begin{center}\begin{tabular}{c|cccc}
            $b$ & $0$ & $b_1$ & $b_2$ & $b_3$\\
            \hline
            $\Delta \cdot u_i(b;23/64)$ & $0$ & $0$ & $\frac{99}{448}$ & $-$\\
        \end{tabular}\end{center}
        \endgroup 

        We can see that $i$'s best response is $\beta_i(23/64)=b_2$.
        Additionally, no other strategy achieves a value of $\Delta \cdot u_i$
        within $\frac{99}{448}$ of the optimal one in the table.

    \end{enumerate}

    Notice that the minimum margin in the tables above by which the best
    response of a bidder was unique was $\frac{33}{896}$. To translate this back
    to the setting of our problem, this would be a difference in utility of
    $\frac{33}{896\Delta}$. Therefore, if the extra utility gained in the (at
    most two other) points in the support where one input bidder and one
    \NOT/$\OR_1$/$\OR_2$ bidder have positive value is less than
    $\frac{33}{896\Delta}$, there can be two types of equilibria, one where all
    $j,k,\ell$ play $s_0$ and one where they play $s_1$, which allows us to
    express all possible inputs to the SAT instance. Thus, we need the following
    relationship to hold:
    \begin{equation*}
        \frac{1}{\Delta} \cdot 2 \delta_{\NOT} < \frac{33}{896} \quad\implies\quad \delta_{\NOT}<\frac{33}{1792}
    \end{equation*}
    Furthermore, given our choice of $\delta_{\NOT}$ and $\Delta$ at the end of
    the reduction, the analysis we provided also holds for
    $\varepsilon$-approximate equilibria, for any
    $\varepsilon<\frac{1}{\Delta}(\frac{33}{896}-2\delta_{\NOT}$).

\end{proof}

For the simulation of negation, we will proceed in two steps: we will first
demonstrate the behaviour of the \NOT bidders and the Projection bidders, and
then we will show how the two together implement a negation.

\paragraph{\NOT bidders.}
We proceed to the analysis of the behaviour of the \NOT bidders, for which we
prove the following lemma:

\begin{lemma}\label{lem:not-bidder} Let $j$ be the \NOT bidder added to the
    auction because of some negated literal $x$ (with corresponding input bidder
    $i$). Fix any $\delta_{\PROJ} < \frac{33}{3584}\delta_{\NOT}$ and
    $\varepsilon<\frac{1}{\Delta}(\frac{33}{1792}\delta_{\NOT}-2\delta_{\PROJ})$.
    Then, at any $\varepsilon$-PBNE $\vec{\beta}$, it should be the case that 
    $$\beta_i(23/64)=1/7 \quad\implies\quad \beta_j(23/64)=2/7$$ and 
    $$\beta_i(23/64)=2/7 \quad\implies\quad \beta_j(23/64)=1/7.$$
\end{lemma}
\begin{proof}
    In our construction, each \NOT bidder $j$ aims to encapsulate the idea of
    negating the strategy of the input bidder $i$ when best-responding. Once
    again, it is immediate from the design of the DFPA instance (also visible in
    \cref{fig:construction}) that at any point in the support of the
    distribution where $j$ has positive value, only one other bidder can have
    positive value, either the input bidder $i$ or a projection bidder, let $x$.
    Thus, similarly to the earlier analysis, we can express the utility of $j$
    at any equilibrium as $u_j = u_j^{(i)} + u_j^{(x)}$. Note that this is about
    utility at an equilibrium to avoid the cases where $j$ could get some
    positive utility from a point in the support where she has value $0$ but
    someone else chooses to bid $0$ albeit having positive value (strictly
    dominated strategy). Given the description of the projection bidders, the
    total mass of points used for the computation of $u_j^{(x)}$ is bounded by
    $\frac{2}{\Delta}\delta_{\PROJ}$. Therefore, as we will pick
    $\delta_{\PROJ}$ to be small enough, it suffices to analyse $u_j^{(i)}$ to
    compute $j$'s best response depending on $i$'s strategy. By
    \cref{lem:input-bidder}, it suffices to only check the cases where $i$ plays
    $s_0$ or $s_1$:

    \begin{enumerate}
        \item If $i$ plays strategy $s_0=(0,b_1,b_2)$:
        
        \begingroup
        \renewcommand{\arraystretch}{1.2}
        \begin{center}\begin{tabular}{c|cccc}
            $b$ & $0$ & $b_1$ & $b_2$ & $b_3$\\
            \hline
            $\frac{\Delta}{\delta_{\NOT}} \cdot u_j^{(i)}(b;23/64)$ & $\frac{759}{16384n}$ & $\frac{3201}{114688}$ & $\frac{759}{16384}$ & $-$\\
            $\frac{\Delta}{\delta_{\NOT}} \cdot u_j^{(i)}(b;1)$ & $\frac{1}{n}$ & $\frac{6}{7}$ & $\frac{15}{14}$ & $\frac{8}{7}$
        \end{tabular}\end{center}
        \endgroup

        We can see that, for $n\geq 2$, $j$'s best response is $s_1=(0,b_2,b_3)$.
        Additionally, no other strategy achieves a value of $\frac{\Delta}{\delta_{\NOT}} \cdot u_j^{(i)}$ within $\frac{33}{1792}$ of the optimal one in the table.

        \item If $i$ plays strategy $s_1=(0,b_2,b_3)$:
        
        \begingroup
        \renewcommand{\arraystretch}{1.2}
        \begin{center}\begin{tabular}{c|cccc}
            $b$ & $0$ & $b_1$ & $b_2$ & $b_3$\\
            \hline
            $\frac{\Delta}{\delta_{\NOT}} \cdot u_j^{(i)}(b;23/64)$ & $\frac{759}{16384n}$ & $\frac{3201}{114688}$ & $\frac{1089}{114688}$ & $-$\\
            $\frac{\Delta}{\delta_{\NOT}} \cdot u_j^{(i)}(b;1)$ & $\frac{1}{n}$ & $\frac{6}{7}$ & $\frac{5}{7}$ & $\frac{6}{7}$
        \end{tabular}\end{center}
        \endgroup

        We can see that, for $n \geq 2$, $j$'s best response is $(0,b_1,b_1)$ or
        $(0,b_1,b_3)$. Additionally, no other strategy achieves a value of
        $\frac{\Delta}{\delta_{\NOT}} \cdot u_j^{(i)}$ within $\frac{33}{1792}$
        of the optimal one in the table.
    \end{enumerate}

    As always, in the above tables we have computed the value of
    $\frac{\Delta}{\delta_{\NOT}} \cdot u_j^{(i)}$, as $\Delta$ and
    $\delta_{\NOT}$ will be picked at the end of the reduction; however, as
    mentioned earlier, this does not affect the computation of best-responses,
    as $\frac{\delta_{\NOT}}{\Delta}$ is a common factor to the mass of each
    point of the distribution introduced for a \NOT bidder. To make sure that
    $j$'s best response only depends on $i$'s strategy, we need to establish
    that $u_j^{(x)}$ does not provide $j$ with enough utility to incentivize her
    to change her strategy, namely:
    \begin{equation*}
        \frac{1}{\Delta} \cdot 2 \delta_{\PROJ} < \frac{1}{\Delta}\cdot\delta_{\NOT}\frac{33}{1792} \quad\implies\quad \delta_{\PROJ} < \frac{33}{3584}\delta_{\NOT} 
    \end{equation*}

    Once again, given our choice of $\delta_{\PROJ}$ and $\Delta$ at the end of
    the reduction, the analysis we provided also holds for
    $\varepsilon$-approximate equilibria, for any
    $\varepsilon<\frac{1}{\Delta}(\frac{33}{1792}\delta_{\NOT}-2\delta_{\PROJ}$).

\end{proof}
Note that in the second case, we did not get a strategy in $\{s_0, s_1\}$ as a
best response; this is where the projection bidders come into play.

\paragraph{Projection bidders.}
The projection bidders we introduced satisfy the following lemma, which
essentially guarantees that together with the \NOT bidders they simulate a \NOT
gate:

\begin{lemma}\label{lem:proj-bidder} Let $j$ be a projection bidder and $i$ be
    the corresponding \NOT bidder. Fix any $\delta_{\OR_1} <
    \frac{33}{3584}\delta_{\PROJ}$ and
    $\varepsilon<\frac{1}{\Delta}(\frac{33}{1792}\delta_{\PROJ}-2\delta_{\OR_1})$.
    Then, at any $\varepsilon$-PBNE $\vec{\beta}$, it is the case that 
    $$\beta_i(23/64)=1/7 \quad\implies\quad \beta_j = s_0$$ and 
    $$\beta_i(23/64)=2/7 \quad\implies\quad \beta_j = s_1.$$
\end{lemma}
\begin{proof}
    Let a projection bidder $j$, introduced after a \NOT bidder $i$. By
    construction (see \cref{fig:construction}), at any point in the support of
    the distribution in which $j$ has positive value, there is at most $1$ other
    bidder with positive value, which is either bidder $i$, or a bidder $x$
    which can be either an $\OR_1$ bidder or an $\OR_2$ bidder. Similarly to the
    previous steps in the proof, we will express the total expected utility for
    $j$ at an equilibrium as $u_j = u_j^{(i)} + u_j^{(x)}$, where $u_j^{(i)}$
    comes from the points in the support where $i$ has positive value and
    $u_j^{(x)}$ comes from the ones where $x$ has positive value.

    Given the description of the $\OR_1$ and $\OR_2$ bidders, the total mass of
    points used for the computation of $u_j^{(x)}$ is bounded by
    $\frac{2}{\Delta}\delta_{\OR_1}$ (note that we pick the discounting factors
    such that $\delta_{\OR_1}>\delta_{\OR_2}$). Therefore, as we will pick
    $\delta_{\OR_1}$ to be small enough, it suffices to analyse $u_j^{(i)}$ to
    compute $j$'s best response depending on $i$'s strategy. We will now
    calculate $j$'s best response to each of $i$'s strategies. In our
    description of the DFPA instance, $i$ has value $23/64$ in all points of the
    distribution where both $i$ and $j$ have positive value. Therefore, it
    suffices to describe how $j$ responds according to $\beta_i(23/64)$: 

    \begin{enumerate}
        \item If $\beta_i(23/64)=b_1$:
        
        \begingroup
        \renewcommand{\arraystretch}{1.2}
        \begin{center}\begin{tabular}{c|cccc}
            $b$ & $0$ & $b_1$ & $b_2$ & $b_3$\\
            \hline
            $\frac{\Delta}{\delta_{\PROJ}} \cdot u_j^{(i)}(b;23/64)$ & $\frac{759}{16384n}$ & $\frac{2231}{16384}$ & $\frac{9537}{114688}$ & $-$\\
            $\frac{\Delta}{\delta_{\PROJ}} \cdot u_j^{(i)}(b;1)$ & $0$ & $\frac{3}{7}$ & $\frac{5}{7}$ & $\frac{4}{7}$
        \end{tabular}\end{center}
        \endgroup

        We can see that, for $n \geq 2$, $j$'s best response is
        $s_0=(0,b_1,b_2)$. Additionally, no other strategy achieves a value of
        $\frac{\Delta}{\delta_{\PROJ}} \cdot u_j^{(i)}$ within $\frac{95}{1792}$
        of the optimal one in the table.

        \item If $\beta_i(23/64)=b_2$:
        
        \begingroup
        \renewcommand{\arraystretch}{1.2}
        \begin{center}\begin{tabular}{c|cccc}
            $b$ & $0$ & $b_1$ & $b_2$ & $b_3$\\
            \hline
            $\frac{\Delta}{\delta_{\PROJ}} \cdot u_j^{(i)}(b;23/64)$ & $\frac{759}{16384n}$ & $\frac{3201}{114688}$ & $\frac{759}{16384}$ & $-$\\
            $\frac{\Delta}{\delta_{\PROJ}} \cdot u_j^{(i)}(b;1)$ & $0$ & $0$ & $\frac{5}{14}$ & $\frac{4}{7}$
        \end{tabular}\end{center}
        \endgroup

        We can see that, for $n \geq 2$, $j$'s best response is
        $s_1=(0,b_2,b_3)$. Additionally, no other strategy achieves a value of
        $\frac{\Delta}{\delta_{\PROJ}} \cdot u_j^{(i)}$ within $\frac{33}{1792}$
        of the optimal one in the table.

    \end{enumerate}

    To make sure that $j$'s best response only depends on $i$'s strategy, we
    need to establish that $u_j^{(x)}$ does not provide $j$ with enough utility
    to incentivize her to change her strategy, namely:
    \begin{equation*}
        \frac{1}{\Delta} \cdot 2 \delta_{\OR_1} < \frac{1}{\Delta}\cdot\delta_{\PROJ}\cdot \frac{33}{1792} \quad\implies\quad \delta_{\OR_1} < \frac{33}{3584}\delta_{\PROJ} 
    \end{equation*}

    Once again, given our choice of $\delta_{\OR_1}$ and $\Delta$ at the end of
    the reduction, the analysis we provided also holds for
    $\varepsilon$-approximate equilibria, for any
    $\varepsilon<\frac{1}{\Delta}(\frac{33}{1792}\delta_{\PROJ}-2\delta_{\OR_1})$.
\end{proof}

Using \cref{lem:not-bidder,lem:proj-bidder}, we can derive the following:

\begin{lemma}\label{lem:not-full} Fix an
    $\varepsilon<\frac{1}{\Delta}(\frac{33}{1792}\delta_{\PROJ}-2\delta_{\OR_1})$.
    For any negated literal $x$, with corresponding input bidder $i$, NOT bidder
    $j$, and projection bidder $k$, it is the case that, at any
    $\varepsilon$-PBNE $\vec{\beta}$, $\beta_k \in \{s_0,s_1\}$ and
    $\mathbf{\chi}[\beta_i]= \overline{\mathbf{\chi}[\beta_k]}$.
\end{lemma}
\begin{proof}
    Follows directly from \cref{lem:not-bidder,lem:proj-bidder}, keeping the
    smallest value $\varepsilon'$ (using the fact that
    $\delta_{\NOT}>\delta_{\PROJ}$) so that we get the result for any
    $\varepsilon$-PBNE where $\varepsilon<\varepsilon'$. 
\end{proof}

\paragraph{$\OR_1$ bidders.}
We proceed to the analysis of the first layer of bidders simulating the
behaviour of an \OR, for which we establish the following lemma:

\begin{lemma}\label{lem:or1} Fix any $\delta_{\OR_2} <
    \frac{1}{1792}\delta_{\OR_1}$ and
    $\varepsilon<\frac{1}{\Delta}(\frac{1}{896}\delta_{\OR_1}-2\delta_{\OR_2})$.
    For any $\OR_1$ bidder $\ell$ introduced by literals with corresponding
    bidders $i,j,k$, at any $\varepsilon$-PBNE $\vec{\beta}$, it must be the
    case that $\mathbf{\chi}[\beta_{\ell}] = \mathbf{\chi}[\beta_i] \vee
    \mathbf{\chi}[\beta_j]$.
\end{lemma}
\begin{proof}
    From the description of the $\OR_{1}$ bidders, any point in the support of
    the joint distribution where $\ell$ has positive value can only have at most
    one other bidder with positive value; this could be either bidder $i$ or $j$
    that belong in either the set of Input bidders or the set of Projection
    bidders (see \cref{fig:construction}), or it could be some $\OR_2$ bidder
    $x$. Similarly to the analysis for the previous bidders, here too we will
    express the total expected utility of $\ell$ as
    $u_{\ell}=u_{\ell}^{(i,j)}+u_{\ell}^{(x)}$, where $u_{\ell}^{(i,j)}$ is the
    utility that $\ell$ receives from points in the support where \emph{either}
    $i$ or $j$ have positive value. From the construction of $\OR_2$ bidders, we
    can once again derive the bound $u_{\ell}^{(x)} \leq
    \frac{2}{\Delta}\delta_{\OR_2}$.

    Below is the analysis of $\ell$'s best-response to each possible pair of
    strategies of $i,j$, with respect to $u_{\ell}^{(i,j)}$:

    \begin{enumerate}
        \item Both $i$ and $j$ play strategy $s_0=(0,b_1,b_2)$:

        \begingroup
        \renewcommand{\arraystretch}{1.2}
        \begin{center}\begin{tabular}{c|cccc}
            $b$ & $0$ & $b_1$ & $b_2$ & $b_3$\\
            \hline
            $\frac{\Delta}{\delta_{\OR_1}} \cdot u_{\ell}^{(i,j)}(b;23/64)$ & $\frac{23}{8192n}$ & $\frac{12513}{57344}$ & $\frac{8481}{57344}$ & $-$\\
            $\frac{\Delta}{\delta_{\OR_1}} \cdot u_{\ell}^{(i,j)}(b;1)$ & $0$ & $\frac{6}{7}$ & $\frac{10}{7}$ & $\frac{8}{7}$
        \end{tabular}\end{center}
        \endgroup

        We can see that, for $n \geq 2$, $\ell$'s best response is
        $(0,b_1,b_2)$. Additionally, no other strategy achieves a value of
        $\frac{\Delta}{\delta_{\OR_1}} \cdot u_{\ell}^{(i,j)}$ within
        $\frac{9}{128}$ of the optimal one in the table.

        \item $i$ plays $s_0=(0,b_1,b_2)$ and $j$ plays $s_1=(0,b_2,b_3)$ (or
        the opposite; notice that the analysis of the two cases is symmetric due
        to the symmetry of the construction on $i$ and $j$):

        \begingroup
        \renewcommand{\arraystretch}{1.2}
        \begin{center}\begin{tabular}{c|cccc}
            $b$ & $0$ & $b_1$ & $b_2$ & $b_3$\\
            \hline
            $\frac{\Delta}{\delta_{\OR_1}} \cdot u_{\ell}^{(i,j)}(b;23/64)$ & $\frac{23}{8192n}$ & $\frac{6305}{57344}$ & $\frac{6369}{57344}$ & $-$\\
            $\frac{\Delta}{\delta_{\OR_1}} \cdot u_{\ell}^{(i,j)}(b;1)$ & $\frac{1}{128n}$ & $\frac{195}{448}$ & $\frac{965}{896}$ & $\frac{257}{224}$
        \end{tabular}\end{center}
        \endgroup
        
        We can see that, for $n \geq 2$, $\ell$'s best response is
        $(0,b_2,b_3)$. Additionally, no other strategy achieves a value of
        $\frac{\Delta}{\delta_{\OR_1}} \cdot u_{\ell}^{(i,j)}$ within
        $\frac{1}{896}$ of the optimal one in the table.

        \item Both $i$ and $j$ play $s_1=(0,b_2,b_3)$:

        \begingroup
        \renewcommand{\arraystretch}{1.2}
        \begin{center}\begin{tabular}{c|cccc}
            $b$ & $0$ & $b_1$ & $b_2$ & $b_3$\\
            \hline
            $\frac{\Delta}{\delta_{\OR_1}} \cdot u_{\ell}^{(i,j)}(b;23/64)$ & $\frac{23}{8192n}$ & $\frac{97}{57344}$ & $\frac{4257}{57344}$ & $-$\\
            $\frac{\Delta}{\delta_{\OR_1}} \cdot u_{\ell}^{(i,j)}(b;1)$ & $\frac{1}{128n}$ & $\frac{3}{448}$ & $\frac{645}{896}$ & $\frac{257}{224}$
        \end{tabular}\end{center}
        \endgroup
        
        We can see that, for $n \geq 2$, $\ell$'s best response is
        $(0,b_2,b_3)$. Additionally, no other strategy achieves a value of
        $\frac{\Delta}{\delta_{\OR_1}} \cdot u_{\ell}^{(i,j)}$ within
        $\frac{65}{896}$ of the optimal one in the table.
    \end{enumerate}
    To make sure that $\ell$'s best response only depends on the strategies of
    $i$ and $j$, we need to establish that $u_{\ell}^{(x)}$ does not provide
    $\ell$ with enough utility to incentivize her to change her strategy,
    namely:
    \begin{equation*}
        \frac{1}{\Delta} \cdot 2 \delta_{\OR_2} < \frac{1}{\Delta}\cdot\delta_{\OR_1}\cdot \frac{1}{896} \quad\implies\quad \delta_{\OR_2} < \frac{1}{1792}\delta_{\OR_1} 
    \end{equation*}

    Once again, given our choice of $\delta_{\OR_2}$ and $\Delta$ at the end of
    the reduction, the analysis we provided also holds for
    $\varepsilon$-approximate equilibria, for any
    $\varepsilon<\frac{1}{\Delta}(\frac{1}{896}\delta_{\OR_1}-2\delta_{\OR_2})$.
\end{proof}

\paragraph{$\OR_2$ bidders.}
Next, we will reason for the behaviour of the $\OR_2$ bidders, getting the following result:

\begin{lemma}\label{lem:or2} Fix any $\delta_{\OUT} <
    \frac{1}{672}\delta_{\OR_2}$ and
    $\varepsilon<\frac{1}{\Delta}(\frac{1}{896}\delta_{\OR_2}-\frac{3}{4}\delta_{\OUT})$.
    For any $\OR_2$ bidder $\ell$ introduced by literals with corresponding
    bidders $i,j$, at any $\varepsilon$-PBNE $\vec{\beta}$, it must be the case
    that $\mathbf{\chi}[\beta_{\ell}] = \mathbf{\chi}[\beta_i] \vee
    \mathbf{\chi}[\beta_j]$.
\end{lemma}
\begin{proof}
    From the description of the $\OR_{2}$ bidders, there are two types of points
    in the support of the joint distribution where $\ell$ has positive value.
    Firstly, it could be the case that exactly one other bidder has positive
    value; this would be either bidder $i$ or $j$ and would belong in any of the
    sets of Input/Projection/$\OR_1$ bidders (see \cref{fig:construction}).
    Secondly, there is one point in the distribution where exactly two other
    bidders $x,y$, which are Output bidders, have positive value. Similarly to
    the analysis for the previous bidders, here too we will express the total
    expected utility of $\ell$ as $u_{\ell}=u_{\ell}^{(i,j)}+u_{\ell}^{(x,y)}$,
    where $u_{\ell}^{(x,y)}$ is the utility that $\ell$ receives from the point
    in the support where $x$ and $y$ have positive value. From the construction
    of output bidders, we obtain the bound $u_{\ell}^{(x,y)} \leq
    \frac{3}{4\Delta}\delta_{\OUT}$.

    Below is the analysis of $\ell$'s best-response to each possible pair of
    strategies of $i,j$, with respect to $u_{\ell}^{(i,j)}$:

    \begin{enumerate}
        \item Both $i$ and $j$ play strategy $s_0=(0,b_1,b_2)$:

        \begingroup
        \renewcommand{\arraystretch}{1.2}
        \begin{center}\begin{tabular}{c|cccc}
            $b$ & $0$ & $b_1$ & $b_2$ & $b_3$\\
            \hline
            $\frac{\Delta}{\delta_{\OR_2}} \cdot u_{\ell}^{(i,j)}(b;23/64)$ &
            $\frac{23}{8192n}$ & $\frac{12513}{57344}$ & $\frac{8481}{57344}$ &
            $-$\\
            $\frac{\Delta}{\delta_{\OR_2}} \cdot u_{\ell}^{(i,j)}(b;1)$ & $0$ &
        $\frac{6}{7}$ & $\frac{10}{7}$ & $\frac{8}{7}$ \end{tabular}\end{center}
        \endgroup

        We can see that, for $n \geq 2$, $\ell$'s best response is
        $(0,b_1,b_2)$. Additionally, no other strategy achieves a value of
        $\frac{\Delta}{\delta_{\OR_2}} \cdot u_{\ell}^{(i,j)}$ within
        $\frac{9}{128}$ of the optimal one in the table.

        \item $i$ plays $s_0=(0,b_1,b_2)$ and $j$ plays $s_1=(0,b_2,b_3)$ (or
        the opposite; notice that the analysis of the two cases is symmetric due
        to the symmetry of the construction on $i$ and $j$):

        \begingroup
        \renewcommand{\arraystretch}{1.2}
        \begin{center}\begin{tabular}{c|cccc}
            $b$ & $0$ & $b_1$ & $b_2$ & $b_3$\\
            \hline
            $\frac{\Delta}{\delta_{\OR_2}} \cdot u_{\ell}^{(i,j)}(b;23/64)$ &
            $\frac{23}{8192n}$ & $\frac{6305}{57344}$ & $\frac{6369}{57344}$ &
            $-$\\
            $\frac{\Delta}{\delta_{\OR_2}} \cdot u_{\ell}^{(i,j)}(b;1)$ &
        $\frac{1}{128n}$ & $\frac{195}{448}$ & $\frac{965}{896}$ &
        $\frac{257}{224}$ \end{tabular}\end{center}
        \endgroup
        
        We can see that, for $n \geq 2$, $\ell$'s best response is
        $(0,b_2,b_3)$. Additionally, no other strategy achieves a value of
        $\frac{\Delta}{\delta_{\OR_2}} \cdot u_{\ell}^{(i,j)}$ within
        $\frac{1}{896}$ of the optimal one in the table.

        \item Both $i$ and $j$ play $s_1=(0,b_2,b_3)$:

        \begingroup
        \renewcommand{\arraystretch}{1.2}
        \begin{center}\begin{tabular}{c|cccc} $b$ & $0$ & $b_1$ & $b_2$ &
            $b_3$\\
            \hline
            $\frac{\Delta}{\delta_{\OR_2}} \cdot u_{\ell}^{(i,j)}(b;23/64)$ &
            $\frac{23}{8192n}$ & $\frac{97}{57344}$ & $\frac{4257}{57344}$ &
            $-$\\
            $\frac{\Delta}{\delta_{\OR_2}} \cdot u_{\ell}^{(i,j)}(b;1)$ &
        $\frac{1}{128n}$ & $\frac{3}{448}$ & $\frac{645}{896}$ &
        $\frac{257}{224}$ \end{tabular}\end{center}
        \endgroup
        
        We can see that, for $n \geq 2$, $\ell$'s best response is
        $(0,b_2,b_3)$. Additionally, no other strategy achieves a value of
        $\frac{\Delta}{\delta_{\OR_2}} \cdot u_{\ell}^{(i,j)}$ within
        $\frac{65}{896}$ of the optimal one in the table.
    \end{enumerate}

    To make sure that $\ell$'s best response only depends on the strategies of
    $i$ and $j$, we need to establish that $u_{\ell}^{(x,y)}$ does not provide
    $\ell$ with enough utility to incentivize her to change her strategy,
    namely:

    \begin{equation*}
        \frac{3}{4\Delta} \delta_{\OUT} < \frac{1}{\Delta}\cdot\delta_{\OR_2}\cdot \frac{1}{896} \quad\implies\quad \delta_{\OUT} < \frac{1}{672}\delta_{\OR_2} 
    \end{equation*}

    Once again, given our choice of $\delta_{\OUT}$ and $\Delta$ at the end of
    the reduction, the analysis we provided also holds for
    $\varepsilon$-approximate equilibria, for any
    $\varepsilon<\frac{1}{\Delta}(\frac{1}{896}\delta_{\OR_2}-\frac{3}{4}\delta_{\OUT})$.
\end{proof}

\paragraph{Output bidders.}
We now present the analysis for the last part of our construction, that of the
output bidders. These are designed so that they can simultaneously best-respond
if and only if the corresponding $\OR_2$ bidder plays $s_1$. Indeed, we show the
following:

\begin{lemma}\label{lem:output} Fix an
    $\varepsilon<\frac{\delta_{OUT}}{56\Delta}$. For any output bidders $k,\ell$
    corresponding to an $\OR_2$ bidder $i$, it is the case that, at any
    $\varepsilon$-PBNE $\vec{\beta}$, $\mathbf{\chi}[\beta_i] = 1$.
\end{lemma}
\begin{proof}
    For the first part of our proof, we need to show that whenever $i$ plays
    $s_0$ there is no equilibrium. We proceed by computing $k$'s best-response
    according to $\ell$'s strategy, when $i$ plays $s_0 = (0,b_1,b_2)$. Our
    construction ensures that $k$ and $\ell$ have value $23/64$ with probability
    $0$, so it suffices to check their strategies when having value $1$ (again,
    no-overbidding means that they will always bid $0$ when having value $0$):
    
    \begin{enumerate}
        \item If $\beta_{\ell}(1)=0$:

        \begingroup
        \renewcommand{\arraystretch}{1.2}
        \begin{center}\begin{tabular}{c|cccc}
            $b$ & $0$ & $b_1$ & $b_2$ & $b_3$\\
            \hline
            $\frac{\Delta}{\delta_{\OUT}} \cdot u_k(b;1)$ & $\frac{1}{n}$ &
            $\frac{33}{28}$ & $\frac{5}{4}$ & $1$\\
        \end{tabular}\end{center}
        \endgroup 

        so, for $n\geq 2$, $k$'s best response is $\beta_k(1)=b_2$.
        Additionally, no other strategy achieves a value of
        $\frac{\Delta}{\delta_{\OUT}} \cdot u_k$ within $\frac{1}{14}$ of the
        optimal one in the table. 

        \item If $\beta_{\ell}(1)=b_1$:

        \begingroup
        \renewcommand{\arraystretch}{1.2}
        \begin{center}\begin{tabular}{c|cccc}
            $b$ & $0$ & $b_1$ & $b_2$ & $b_3$\\
            \hline
            $\frac{\Delta}{\delta_{\OUT}} \cdot u_k(b;1)$ & $\frac{1}{n}$ &
            $\frac{15}{14}$ & $\frac{5}{4}$ & $1$\\
        \end{tabular}\end{center}
        \endgroup 
        so, for $n \geq 2$, $k$'s best response is $\beta_k(1)=b_2$.
        Additionally, no other strategy achieves a value of
        $\frac{\Delta}{\delta_{\OUT}} \cdot u_k$ within $\frac{5}{28}$ of the
        optimal one in the table. 

        \item If $\beta_{\ell}(1)=b_2$:

        \begingroup
        \renewcommand{\arraystretch}{1.2}
        \begin{center}\begin{tabular}{c|cccc}
            $b$ & $0$ & $b_1$ & $b_2$ & $b_3$\\
            \hline
            $\frac{\Delta}{\delta_{\OUT}} \cdot u_k(b;1)$ & $\frac{1}{n}$ &
            $\frac{6}{7}$ & $\frac{55}{56}$ & $1$\\
        \end{tabular}\end{center}
        \endgroup 
        so, for $n \geq 2$, $k$'s best response is $\beta_k(d)=b_3$.
        Additionally, no other strategy achieves a value of
        $\frac{\Delta}{\delta_{\OUT}} \cdot u_k$ within $\frac{1}{56}$ of the
        optimal one in the table. 

        \item If $\beta_{\ell}(1)=b_3$:

        \begingroup
        \renewcommand{\arraystretch}{1.2}
        \begin{center}\begin{tabular}{c|cccc}
            $b$ & $0$ & $b_1$ & $b_2$ & $b_3$\\
            \hline
            $\frac{\Delta}{\delta_{\OUT}} \cdot u_k(b;1)$ & $\frac{1}{n}$ &
            $\frac{6}{7}$ & $\frac{5}{7}$ & $\frac{11}{14}$\\
        \end{tabular}\end{center}
        \endgroup 

        so, for $n \geq 2$, $k$'s best response is $\beta_k(1)=b_1$.
        Additionally, no other strategy achieves a value of
        $\frac{\Delta}{\delta_{\OUT}} \cdot u_k$ within $\frac{1}{14}$ of the
        optimal one in the table.
    \end{enumerate}

    We summarize $k$'s best-responses, which are unique within
    $\frac{\delta_{\OUT}}{56\Delta}$, in the following table:
    \begin{center}\begin{tabular}{c|cccc}
        $\beta_{\ell}(1)$ & $0$ & $b_1$ & $b_2$ & $b_3$\\
        \hline
        $\beta_k(1)$ & $b_2$ & $b_2$ & $b_3$ & $b_1$ 
    \end{tabular}
    \captionof{table}{Bidder $k$'s best-responses to bidder $\ell$'s strategies}
    \end{center}

    As bidders $k$ and $\ell$ are symmetrically defined, the analysis for
    $\ell$'s best-responses is identical. Therefore, we can see from the best
    response table that it is impossible for $k$ and $\ell$ to simultaneously
    pick $\varepsilon$- best-responses for
    $\varepsilon<\frac{\delta_{\OUT}}{56\Delta}$ (if that were the case, there
    would have to exist a column in the table where both played the same
    strategy or two columns where they swap strategies).

    We now analyse the remaining case, where the output bidder $i$ plays
    strategy $s_1=(0,b_2,b_3)$. We will demonstrate that the pair of strategies
    where, at value $1$, one of $k,\ell$ plays $b_1$ and the other plays $b_3$
    leads to an equilibrium:

    \begin{enumerate}
        \item If $\beta_{\ell}(1)=b_1$:

        \begingroup
        \renewcommand{\arraystretch}{1.2}
        \begin{center}\begin{tabular}{c|cccc}
            $b$ & $0$ & $b_1$ & $b_2$ & $b_3$\\
            \hline
            $u_k(b;1)$ & $\frac{1}{n}$ & $\frac{6}{7}$ & $\frac{55}{56}$ & $1$\\
        \end{tabular}\end{center}
        \endgroup 

        so, for $n \geq 2$, $k$'s best response is $\beta_k(1)=b_3$.
        Additionally, no other strategy achieves utility $u_k$ within
        $\frac{1}{56}$ of the optimal one in the table.

        \item If $\beta_{\ell}(1)=b_3$:

        \begingroup
        \renewcommand{\arraystretch}{1.2}
        \begin{center}\begin{tabular}{c|cccc}
            $b$ & $0$ & $b_1$ & $b_2$ & $b_3$\\
            \hline
            $u_k(b;1)$ & $\frac{1}{n}$ & $\frac{6}{7}$ & $\frac{5}{7}$ & $\frac{11}{14}$\\
        \end{tabular}\end{center}
        \endgroup 

        so, for $n \geq 2$, $k$'s best response is $\beta_k(1)=b_1$.
        Additionally, no other strategy achieves utility $u_k$ within
        $\frac{1}{14}$ of the optimal one in the table.

    \end{enumerate} 

    From the computation of the best responses for $k$ and $\ell$, we can see
    that in this case there are in fact two equilibria -- these are defined by
    the pairs of strategies where one $k,\ell$ plays $b_1$ at value $1$ and the
    other plays $b_3$ at value $1$. Hence, we have demonstrated that, if
    $\beta_i = s_1$, there is a PBNE of the DFPA.
\end{proof}

\paragraph{Choice of parameters.}
We conclude our proof of \cref{thm:NP-completeness} by showing how to pick the
values of the parameters $\delta$ and $\Delta$ of the DFPA. We begin by choosing
the values of all the $\delta$ factors to satisfy the above inequalities of the
premises of
\cref{lem:input-bidder,lem:not-bidder,lem:proj-bidder,lem:not-full,lem:or1,lem:or2,lem:output}.
Notice that we can solve these inequalities in the order we introduced them, as
every new $\delta$ was picked to be smaller than some scaled version of the
previous one (for example, $\delta_{\PROJ}$ should be less than a multiple of
$\delta_{\NOT}$, $\delta_{\OR_1}$ should be less than a multiple of
$\delta_{\PROJ}$ etc.). There is a tradeoff between the values we pick, as these
will affect the value of $\varepsilon$ we get for the hardness result of
computation of approximate equilibria; we want to try to make these $\delta$
factors as large as possible, while maintaining the aforementioned inequalities.
Notice also that the solution to these inequalities does not depend on the size
of the problem.

We now proceed to the choice of $\Delta$. Notice that every point of mass $x$ we
added to the joint distribution $F$ was defined to appear with probability equal
to something of the form $c_x \cdot \frac{1}{\Delta}$, where $c_x$ is some
constant (after fixing the $\delta$ factors). We will now define $\Delta$ as
follows:

\begin{equation*}
    \Delta = \sum_{x \in \support{F}} c_x
\end{equation*}

Crucially, this depends polynomially on the size of the SAT instance we are
reducing from. This means that the final $\varepsilon$ that we implicitly
compute here, such that for all $\varepsilon'<\varepsilon$ the problem of
deciding the existence of a $\varepsilon'$-PBNE is NP-hard, is of size
inverse-polynomial to the input. 

Moreover, notice that the numerical parameters of our instance (that is, the
values of the pmf) are products of two constants and $\frac{1}{\Delta}$, the
latter being the inverse of a sum of polynomially many constants. Given this,
our proof of~\cref{thm:NP-completeness} in this section actually implies a
\emph{strong} NP-hardness result, ruling out, thus, the existence of a
pseudopolynomial algorithm (unless P=NP); see, e.g., \cite[Sec.~4.2]{Garey1979a}

To conclude our proof, assume that the \sat instance has a satisfying assignment
$\vec{\alpha}: X \rightarrow \{0,1\}^n$, where $n$ is the number of variables.
We will show that there is an equilibrium of the DFPA that our reduction
constructs. Consider the profile in which each bidder $i$ corresponding to a
variable $X_i$ plays strategy $s_{\alpha(i)}$. By \cref{lem:input-bidder}, we
know that all the input bidders introduced for this variable will then
simultaneously best-respond to each other by playing $s_{\alpha(i)}$. Using
\cref{lem:not-full,lem:or1,lem:or2}, we get that the bidders corresponding to
each clause being evaluated will satisfy the properties of the boolean operators
as required. Moreover, these bidders will necessarily play strategy $s_1$, since
$\vec{\alpha}$ is a satisfying assignment and that the operators have been
correctly simulated. Finally, using \cref{lem:output}, we can ensure that the
output bidders will simultaneously play best-responses, since their
corresponding $\OR_2$ bidder plays $s_1$.

If, on the other hand, there is no satisfying assignment to the \sat instance,
we prove that there can be no equilibrium in the corresponding DFPA. To see
this, notice that there is no choice of strategies in $s_0,s_1$ for the input
bidders such that all $\OR_2$ bidders best respond with $s_1$ (as that would
imply a satisfying assignment), therefore there should exist at least $2$ output
bidders $k,\ell$ (corresponding to some $\OR_2$ bidder playing $s_0$) that
cannot simultaneously best respond to each other. Therefore, there is no PBNE in
the DFPA.
\qed

\section{Polynomial-Time Algorithms via Bid Sparsification}
\label{sec:sparsification}

In this section we present our first set of positive results that use a
technique which we call \emph{bid sparsification}. Our main results of the
section are polynomial-time algorithms for computing monotone
$\varepsilon$-approximate MBNE, for appropriate choices of the error parameter
$\varepsilon$.
The bid sparsification technique was essentially developed in previous work
\citep{chen2023complexity,fghk24} for the IPV and IID settings; here we work out
the details required for its expansion to the APV and SAPV settings. We develop
the proofs for the case of the monotone PBNE of the CFPA; by
\cref{lem:discrete-to-continuous-and-back-for-existence}, this immediately
implies the same type of results for the MBNE of the DFPA as well.

The term ``bid sparsification'' comes from the following lemma, which allows us
to work with a (much) smaller subset of the bidding space $B$, at the expense of
some error in the equilibrium approximation. The first version of such a lemma
was developed by \citet{chen2023complexity}. The version that we use here is a
straightforward adaptation of the version presented in \citep{fghk24}.
\begin{lemma}[Bidding Space Shrinkage Lemma]
    \label{lem:shrinkage}
    Consider a CFPA with APV and bidding space $B$ and let $M$ be a positive
    integer. We can construct a bidding space $B'\subseteq B$ with cardinality
    $\card{B'}\leq M$, in time polynomial in $M$ and the size of the input such
    that any $\varepsilon$-approximate PBNE of the auction restricted to the
    bidding space $B'$ is a $\left(\varepsilon+\frac{1}{M}\right)$-approximate
    PBNE in the original auction.
\end{lemma}
\begin{proof}[Proof sketch]
    The result follows from the proof of \cite[Lemma 5.1]{fghk24}, by noticing that the priors are internal to the computation of the utilities, which are bounded by the differences in the $H$ functions (the winning probabilities) that we bound trivially by $1$ in our setting too.
    Additionally, it is safe to replace the MBNE condition in the original proof by the PBNE condition in the CFPA setting.
    Finally, notice that, starting from a symmetric equilibrium in the auction with the smaller bidding space, the approximate equilibrium that we retrieve in the original auction is also symmetric by construction.
\end{proof}

On the smaller bidding space $B'$, we can then formulate the equilibrium
computation problem as a system of polynomial inequalities, which can be solved
via an algorithm of \citet{grigor1988solving} within $\delta$ precision in time
polynomial in the number of polynomials and their degrees, polynomial in
$\log(1/\delta)$, and exponential in the number of variables; in our
formulation, both $n$ and $|B'|$ appear in the exponent. This is where the
Shrinkage lemma above is employed, as we can choose $B'$ to be exponentially
smaller than $B$. For the dependence on $n$, we can either fix the number of
players $n$ (as in \citep{fghlp2021_sicomp}, or we can use the symmetry
condition (as in \citep{fghk24}) to efficiently enumerate over the different
possible supports of the distribution in an appropriate representation. The main
results of the section are captured by the following theorem:

\begin{theorem}\label{thm:bid-sparsification-CFPA-main-text} For any fixed
$\varepsilon>0$, a symmetric monotone $\varepsilon$-approximate PBNE of the CFPA can be
computed in time polynomial in the description of the auction, when the values are $k$-GSAPV for a fixed~$k$. 
In particular, there is a PTAS for computing
\begin{enumerate}[label=(\alph*)]
\item under APV and a fixed number of bidders: a monotone approximate PBNE, and 
\item under SAPV (i.e., $1$-GSAPV): a symmetric monotone approximate PBNE. 
\end{enumerate}
Similarly, these results also hold for computing approximate MBNE in the DFPA.
\end{theorem}

Before presenting the complete proof of
\cref{thm:bid-sparsification-CFPA-main-text}, we will provide a high-level
outline of the technique.
We employ the concept of $k$-GSAPV to prove a general version of the result, which elegantly unifies the proofs for (a) and (b) in the statement of the theorem.
We begin by establishing an efficient algorithm in the case where the number of
bidders is constant and then we proceed by demonstrating an efficient algorithm
in the SAPV setting, taking advantage of the symmetry of the bidders. We then
use our result from \cref{lem:discrete-to-continuous-and-back-for-existence} to
transfer our results to the DFPA setting under the same restrictions. In both
cases, our technique is inspired by the efficient algorithm in \cite[Section
6]{fghlp2021_sicomp}, combined with an adapted version of the Shrinkage Lemma
from \cite[Lemma 5.1]{fghk24}, which we state in \cref{lem:shrinkage}.

The high level idea of our proof follows the structure below:

\begin{enumerate}
    \item Use \cref{lem:shrinkage} to show that it suffices to search for an
    approximate equilibrium in a corresponding auction with a smaller bidding
    space.
    \item Guess (for each bidder) the jump points of her equilibrium strategy.
    \item Formulate the problem of finding the exact positions of the jump
    points as a system of polynomial inequalities of polynomially-large degree,
    to which we can compute a $\delta$-approximate solution using standard
    methods in time polynomial in $\log 1/\delta$ and the size of the input using a result from \cite{grigor1988solving}.
    \item Project the approximate solution computed in the previous step back to the space of feasible strategies, bounding the approximation that we get on the equilibrium condition.
\end{enumerate}

\subsection{Proof of~\texorpdfstring{\cref{thm:bid-sparsification-CFPA-main-text}}{Theorem~\ref*{thm:bid-sparsification-CFPA-main-text}}}
\begin{proof}
    We will follow closely the proof of \cite[Section 6]{fghlp2021_sicomp}, outlining the parts that need to be carefully adapted for the proof to work in our setting.
    In the subjective prior setting, the distributions were assumed to be piecewise polynomial over defined sub-intervals.
    Let $\mathcal{R}=\{(\vec{R^1},w_1),\ldots,(\vec{R^\ell},w_\ell)\}$ be the representation of the joint prior $F$.
    Given this representation, we can efficiently compute the support of the marginal distribution and find the intervals in which the marginal is constant, yielding a succinctly representable marginal distribution for each bidder, consisting of $\poly(\ell)$ intervals.

    Following \cite{fghlp2021_sicomp}, we carry out the same procedure of guessing, for each bidder, the assignment of jump points to the sub-intervals.
    This requires enumerating over the possible ways of assigning the $k \cdot (\card{B}-1)$ jump points (representing the strategy corresponding to each group) to the $\poly(\ell)$ intervals, which can be done in time $O(\poly (\ell)^{k\card{B}})$.
    Notice that this is exponential in $\card{B}$; however, utilizing the \cref{lem:shrinkage}, we will only run the enumeration step in the auction with the reduced bidding space and then transfer the approximate equilibrium to the original auction. 
    The reasoning for handling potential collisions of sequential jump points also follows directly from \cite{fghlp2021_sicomp}.
    
    We can then proceed to writing the system of polynomial inequalities that express the equilibrium conditions.
    The only difference here is in the expression of the utility functions, and consequently of the winning probabilities $H$.
    In this case, we will use the definition of $H$ from \cref{eq:cfpa-utility-h-via-repr} to show how to efficiently express the inequalities we will add to the system.
    We can write $T_n(b,n-1,r,\vec{v_{-n}})$ as follows:
    \begin{equation}
        T_n(b,n-1,r,\vec{v_{-n}}) = \sum_{\substack{(r_1,r_2,\ldots,r_k)\in \{0,1,\ldots,r\}^k, \\ r_1+r_2+\ldots+r_k=r}} \prod_{j \in [k]}\binom{n_j'}{r_j} g_{\rho(j),b}^{r_j} \cdot G_{\rho(j),b}^{n_j'-r_j}
    \end{equation}
    where $n_j'=n_j$ for all groups other than the one $n$ belongs in, and $n_j'=n_j-1$ for $n$'s group (recall that $n_k$ indicates the number of bidders in group $k$), $\rho(j)$ is a representative bidder from group $j$ (which can be computed to be the bidder with index $\rho(j)=1+\sum_{k'=1}^{j-1}n_{k'}$) and $g_{\rho(j),b}$ and $G_{\rho(j),b}$ come from the definitions in~\eqref{eq:g-utility-cfpa} and~\eqref{eq:G-utility-cfpa} respectively.
    We can now see that the number of summands is only exponential in $k$, but polynomial in the size of the rest of the input.
    Therefore, we can express the equilibrium as a system of polynomial inequalities of degree at most $kn$.
    For constant $k$, the degree of the polynomials is at most polynomial in $n$, which means that we can invoke the theory from \cite{grigor1988solving} in order to achieve, for any $\delta \in (0,1]$ of our choice, a $\delta$-approximate solution in time polynomial in $\log 1/\delta$ and the size of the input.
    Additionally, when the number of bidders is constant, we can write an efficient system of polynomial inequalities even when no bidders are in the same group (i.e., $n=k$). 
    For the final step of the proof, we need to round the approximate solution to the system of inequalities back to a feasible equilibrium strategy.
    The rounding process is the same, and correctness follows from the proof of \cite[Theorem 6.1]{fghlp2021_sicomp}.
    This concludes the proof for the computation of PBNE in the CFPA.
    By direct application of~\cref{lem:discrete-to-continuous-and-back-for-existence}, we get the corresponding results about MBNE in the DFPA from the statement of the theorem.
\end{proof}

\section{Polynomial-Time Algorithms via Bid Densification}
\label{sec:densification}

In this section, we present our positive results that use a technique which we
call \emph{bid densification}. While bid sparsification, the technique used in
the previous section, was based on previous work, the bid densification approach
is introduced in our work for the first time. The idea here is the opposite:
starting from a CFPA with discrete bidding space $B$, we consider a variant with
the same joint value distribution and a continuous bidding space (without loss
of generality the interval $[0,1]$); we refer to this variant as the
\emph{continuous} CFPA (CCFPA).\label{page:CCFPA_result_overview_5} We then
invoke a closed form expression that has been developed for the CCFPA in the
economics literature; concretely for the case of SAPV, we use the equilibrium
strategy $\beta$ as described by \citet{MW82} (see \eqref{eq:MW-CCFPA-PBNE-beta}
in \cref{sec:densification}). The idea is to invert this bidding function on the
set of bids in $B$, to obtain the jump-points of a monotone strategy
$\hat{\beta}$ for the CFPA. Note that we can only do that approximately, as the
description of $\beta$ contains integrals and algebraic expressions. The most
crucial step is to show that $\hat{\beta}$ is an approximate equilibrium of the
CFPA, for an appropriate approximation parameter. This is only possible under
necessary assumptions on the \emph{density of the bidding space} (i.e., the
maximum distance between any two consecutive bids), and bounds on the density of
the joint distribution. We first state two necessary definitions.

\begin{definition}[Bounded Priors]
    \label{def:bounded-distributions}
    Let $\priorlowerbound,\priorupperbound$ be positive reals. An $n$-bidder
    CFPA with SAPV will be called
    \emph{$(\priorlowerbound,\priorupperbound)$-bounded}, if the density $f$ of
    its joint prior distribution satisfies $\priorlowerbound\leq  f(\vec{x})
    \leq \priorupperbound$ for all $\vec{x}\in[0,1]^n$. For the special IID
    case, where $F_1$ denotes the marginal prior distribution (of each player),
    we will call the auction \emph{$\priorupperbound$-bounded}, if $f_1(x)\leq
    \priorupperbound$ for all $x\in\support{F_1}$, where $f_1$ is the density of
    $F_1$.
\end{definition} 
Two immediate observations are in order,
regarding~\cref{def:bounded-distributions}. First, note that
$(\priorlowerbound,\priorupperbound)$-boundedness implies strictly positive
density, i.e., the joint prior distribution $F$ having \emph{full support}:
$\support{F}=[0,1]^n$. Secondly, although IID, as an auction model, is a special
case of SAPV, our definition of boundedness for the IID case is weaker, in the
sense that it does not demand \emph{any} lower bound on the prior's density; in
particular, $\priorupperbound$-bounded IID priors may not have full support.
Furthermore, the upper bounds (both denoted by parameter $\priorupperbound$ in
our statement of~\cref{def:bounded-distributions}) of the two boundedness
notions do not readily translate between each other, since one applies to the
\emph{joint}, $n$-dimensional density (SAPV) and the other to the
single-dimensional bidder marginals (IID).

\begin{definition}
    \label{def:denseness-bidding-space}
    Let $\delta>0$. The bidding space $B$ of a DFPA/CFPA will be called
    \emph{$\delta$-dense} if, for all $x\in[0,1]$ there exists some $b\in B$
    such that $\card{b-x} \leq \delta$. 
\end{definition}

We now formally state the main theorem that we obtain with this technique. 

\begin{theorem}
    \label{th:approximate-PBNE-CFPA-to-CCFPA}
    Consider an $n$-bidder CFPA with
    $(\priorlowerbound,\priorupperbound)$-bounded SAPV and a $\delta$-dense
    bidding space. For any $\varepsilon>0$, a
    $2\gamma(\delta+2\varepsilon)$-approximate, monotone and symmetric, PBNE of
    the auction can be found in time polynomial in its description and
    $\log(1/\varepsilon)$, where
    $\gamma=2(n-1)(\priorupperbound/\priorlowerbound)^2$. For IID settings, the
    approximation parameter can be improved to $\gamma=n\priorupperbound$
    (without assuming any lower bound on the density, or full support).
\end{theorem}

The remainder of this section is devoted to proving
\cref{th:approximate-PBNE-CFPA-to-CCFPA}. Let us begin with providing some
useful interpretation of the parameters of the statement. Firstly, we consider
the dependence on the boundedness (see~\cref{def:bounded-distributions})
``magnitude'' $\phi=\priorupperbound/\priorlowerbound$ to be rather benign; in
particular, when $\priorlowerbound$ and $\priorupperbound$ are constants, then
$\phi$ is a constant that can easily be absorbed in the $\delta$ parameter. For
example, if the distribution is uniform, then $\phi=1$. The presence of $(n-1)$
in the approximation error at first seems problematic. Observe however that in
almost all conceivable applications of first-price auctions, the number of
allowable bids would be \emph{much} larger than the number of bidders. For
example, one can envision of a bidding space that contains all multiples of 5
cents; in this case, $\delta$ will be much smaller than $n$. Even if $\delta =
1/n^2$, the bound in \cref{th:approximate-PBNE-CFPA-to-CCFPA} results in error
which is $1/\poly(n)$. Therefore in some cases, the algorithm that we obtain via
bid densification is superior to the one than the one that uses bid
sparsification, noting that the latter requires exponential time to achieve
error which is inversely polynomial in $n$.

\subsection{CCFPA: Continuous Bidding Space}
\label{sec:CCFPA-appendix-technical}
Therefore, the key object of study in this \cref{sec:densification} will be
$n$-bidder first-price auctions with continuous value priors \emph{and}
continuous bidding space, namely \emph{continuous} CFPA (CCFPA) --- see also our
previous discussion in pp.~\pageref{page:CCFPA_intro_1}
and~\pageref{page:CCFPA_result_overview_5}. This is a straightforward extension
of our standard CFPA model (\cref{sec:CFPA-model}), where we allow players to
bid over the entire unit interval; that is, bidding strategies are
functions\footnote{For equilibrium analysis purposes, which is our focus in this
section, these functions are allowed to be partial, since they only need to be
defined within the support of the bidders' marginals;
see~\cref{def:approx-mixed-bayes-nash-equilibrium}.} $\beta:[0,1]\map [0,1]$. It
is useful to also view this from the oppositive perspective: strategies in CFPA
with a discrete bidding space $B\subseteq [0,1]$ are still ``legitimate'' CCFPA
strategies, that simply happen to have a discrete/restricted range
$\beta([0,1])$. 
Furthermore, since~\cref{th:approximate-PBNE-CFPA-to-CCFPA} considers symmetric
priors, in this section we work under the SAPV assumption
(see p.~\pageref{page:SAPV-def}.) 

Following our standard notation (see~\cref{sec:CFPA-model}), let $F$ and $f$
denote the cdf and pdf, respectively, of the (absolutely) continuous joint
distribution of bidder values. Throughout this section, we will use
$(X_1,X_2,\dots,X_n)$ to denote a random vector of values from this
distribution. For $i\in[n]$, we use $F_i$ to denote the (cdf of the)
\emph{marginal} distribution of $X_i$; its support will be denoted by
$V_i\coloneqq \support{F_i}$, and its density (pdf) by $f_1$. Note that, due to
symmetry, all marginals are identical and, therefore, for simplicity we will be
usually making our arguments from the perspective of player $i=1$. We will also
use $\underline{v}:=\inf V_1$ to denote the leftmost point of the marginals'
support.

For a value $v\in V_1$, we
use $G_v$ to denote the distribution of the \emph{maximum order statistic} of
all other bidders' values, conditioned on $X_1=v$. That is, if we define the
random variable 
\begin{equation*}
Y_1 \coloneqq \max\nolimits_{i=2,3,\dots,n} X_i,
\end{equation*}
we have that
\begin{equation*}
    G_v(y) \coloneqq \prob{Y_1\leq y\fwh{X_1=v}}
\end{equation*}
for all $y\in[0,1]$.
We also let $g_v$ denote the density function of $G_v$.
Notice here that for IID values, the prior is a product distribution, and
therefore in that case we have that $G_v(y)=G(y)\coloneqq F_1^{n-1}(y)$ for all
$v\in V_1$ and $y\in[0,1]$, where $G$ denotes the (cdf of the) maximum order
statistic of $(n-1)$-many iid draws from the marginal distribution $F_1$. Thus,
the corresponding density function can also be elegantly expressed as
$g_v(y)=g(y):= (n-1) F_1^{n-2}(y) f_1(y)$.\label{page:iid-max-order-statistics}

\subsubsection{SAPV vs IID}
\label{sec:SAPVvsIID} 
As it is standard in the literature of the CCFPA setting,\footnote{See, e.g.,
\cite[p.~59]{Menezes2005}, \cite[Sec.~6.4]{krishna2009auction} and
\cite[Sec.~5.4.3]{Milgrom_2004}.} in order to avoid pathological behaviour (see,
e.g., \citep[Footnote~21]{MW82}), throughout this entire
\cref{sec:densification} we will also assume that our SAPV priors have full
support, i.e., $f(\vec{x})>0$ for all $\vec{x}\in (0,1)^n$, without explicitly
mentioning it every time. Note that this is without loss of generality, since
our main end result (\cref{th:approximate-PBNE-CFPA-to-CCFPA}) is stated under a
$(\priorlowerbound,\priorupperbound)$-boundedness assumption, which is stronger
(see also the discussion following~\cref{def:bounded-distributions}).
However, we will \emph{not} make such full-support assumptions whenever studying
IID priors, as it is not required at a technical level; this is to achieve
maximum applicability of our results and full compatibility with existing work.
In that sense, our IID model is not merely a restriction of the SAPV setting,
since it allows for a wider class of bidder marginals $F_1$ (although,
obviously, the resulting joint density $F$ needs to be a \emph{product}
obviously, the resulting joint density $F$ needs to be a \emph{product}
distribution, under IID). 

Another, perhaps even more critical difference, is complexity-theoretic. Recall
(see~\cref{sec:representation}) that the two models naturally induce different
input representations: in the SAPV case, we explicitly describe the (piecewise
constant) joint density $f$, while in the IID model the (piecewise constant)
marginal density $f_1$ needs to be provided instead. Notice that, although
mathematically one representation can be fully derived by the other, this would
induce, in general, an exponential (on the number of bidders $n$) blow-up in the
description size when translating the IID setting to the SAPV formalism.

The above points highlight why we cannot simply handle the IID case an immediate
special case of the SAPV one, directly instantiating the results of the latter
to derive results for the former. Therefore, throughout this
\cref{sec:densification} we will take care to treat the two models separately,
when needed. This necessitates, for most of our results, slightly different
result statements for the two models, as well as, many times, notably different
proof approaches, at a theoretical level.

\subsection{The Canonical Equilibrium of Symmetric CCFPA}
\label{sec:canonical-equilibrium-MW}

From the theory developed by~\cite{MW82}\footnote{For a more accessible,
textbook-style presentation of the notions in this section, we point the reader
to~\cite[Sec.~5.3]{Menezes2005} or~\cite[Sec.~6.4]{krishna2009auction}: their
presentation is further simplified by hard-wiring the full-support assumption in
their exposition. Although, as discussed above (see~\cref{sec:SAPVvsIID}), we
will also eventually apply such an assumption for our main results, we have
still decided to keep our exposition in this paper as general as possible,
staying closer to the spirit of original work of~\cite{MW82}, and even taking
additional care with handling and clarifying various technical subtleties, as
this allows us to handle collectively, up to some extent, together the SAPV and
the IID case (for which we will \emph{not} apply such a full-support assumption
in the end) in a more elegant way. Furthermore, we believe that this provides
maximum transparency for the reader and for follow-up work. The reader
interested directly in the full-support SAPV case only, and perhaps being
overwhelmed by the generality of the presentation here, can safely take
$V_1=\support{F_1}=[0,1]$ and $\underline{v}=0$ in the following.} we know that
the following bidding function $\beta$ (when adopted by all players) constitutes
a symmetric, nondecreasing (and no-overbidding) PBNE of our CCFPA setting:
\begin{align}
    \beta(v) &\coloneqq v-\int_{\underline{v}}^v L_v(y) \,\mathrm{d} y &&\text{for all}\;\; v \geq \underline{v} \label{eq:MW-CCFPA-PBNE-beta}\\
\intertext{with}
    L_v(y) &\coloneqq \exp\left(-\int_{y}^v\frac{g_t(t)}{G_t(t)}\, \mathrm{d}t\right) &&\text{for all}\;\; y\in[\underline{v},v]\label{eq:MW-CCFPA-PBNE-L}.
\end{align}
For the above quantities to be well-defined, we follow the standard convention
(see~\cite[Footnotes~22 and~23]{MW82}) of $g_t(t)/G_t(t) \coloneqq 0$ for all
$t\notin V_1$. 
For IID settings, one can show\footnote{See \citep[Sec.~6.43,
p.~96]{krishna2009auction}.} that $L_v(y)=\frac{G(y)}{G(v)}$ for all
$\underline{v}\leq y \leq v$, and thus, the equilibrium bidding strategy $\beta$
from above, can be more succinctly expressed as
\begin{equation}
    \label{eq:MW-CCFPA-PBNE-beta-IID}
    \beta(v) \coloneqq v-\int_{\underline{v}}^v \frac{G(y)}{G(v)} \,\mathrm{d} y \qquad\qquad\text{for all}\;\; v \geq \underline{v}.
\end{equation}

We will refer to the bidding strategy $\beta$ defined above, as the
\emph{canonical equilibrium strategy}. It is not hard to see\footnote{This is a
direct consequence of the fact that function $G_t(t)$ is absolutely continuous,
with respect to $t$, due to the fact that the underlying value prior
distribution $F$ is absolutely continuous, and that function $g_t(t)$ is
(Lebesgue) integrable. For more background in such concepts, the interested
reader is referred to any classical textbook in Real Analysis. For example, for
absolute continuity, see~\cite[Sec.~6.4]{RoydenFitzpatrick2010}.} that, for both
the (fully-supported) SAPV and IID settings, these canonical equilibrium
strategies are \emph{absolutely}
continuous\label{page:beta-absolutely-continuous} functions over
$[\underline{v},1]$. Furthermore, we know\footnote{See~\cite[Eq.~(7)]{MW82}
and~\cite[Eq.~(5.7)]{Menezes2005} or~\cite[Eq.~(6.6)]{krishna2009auction}} that
$\beta$ is almost everywhere differentiable (within the marginals' support) and,
in particular, it needs to satisfy the following differential equation:
\begin{align}
    \label{eq:MW-beta-differential-equation-def}
    \beta'(v) = [v-\beta(v)]\frac{g_v(v)}{G_v(v)}
    &&\text{for a.e.}\;\; v\geq \underline{v}.
\end{align}
Finally, it can be shown\footnote{\label{footnote:beta-strictly-increasing}See,
e.g., \cite[p.~66]{Menezes2005} for the (full-support) SAPV case and
\eqref{eq:MW-CCFPA-PBNE-beta-IID} for the IID one. Alternatively, the (strict)
monotonicity of $\beta$ can be derived directly by
analysing~\eqref{eq:MW-CCFPA-PBNE-beta} through the perspective that $L_v$ (as
defined in~\eqref{eq:MW-CCFPA-PBNE-L}) is a valid cumulative distribution
function, supported on $[\underline{v},v]\inters V_1$ (this is, essentially,
what Property~\ref{item:affiliation-property-measure-L} of our
following~\cref{prop:affiliation-properties} establishes). Then, the canonical
equilibrium strategy can be written as a Lebesgue-Stieltjes integral
$\beta(v)=\int_{\underline{v}}^v y \,\mathrm{d} L_v(y)$; see, e.g., \cite[Eq.~8,
p.~1107]{MW82}.} that the canonical equilibrium strategy is not only
nondecreasing in $[\underline{n},1]$, but actually it is \emph{strictly}
increasing within the support of the marginals $V_1$ and constant in
$[\underline{v},1]\setminus V_1$.

The following~\cref{prop:affiliation-properties} captures some important
properties of the canonical equilibrium strategy $\beta$ defined above. We
collect them below, for ease of reference in our technical exposition later in
this section.

\begin{lemma}
    \label{prop:affiliation-properties} Because the bidder values are
    affiliated, the following properties hold:
    \begin{enumerate}
        \item\label{item:affiliation-property-monotonicty-ratio-max-order-pdf-cdf} For any $y\in[\underline{v},1]$, the ratio $\frac{g_v(y)}{G_v(y)}$ is nondecreasing in $v\in V_1$.
        \item\label{item:affiliation-property-measure-L} For any $v\in V_1$, $L_v(y)$ is absolutely continuous and nondecreasing in $y\in[\underline{v},v]$, with $L_v(\underline{v})=0$ and $L_v(v)=1$.
        \item For any $v,v'\in V_1$ with $v'<v$: $L_v(v')<L_v(v)$.
        \item \label{item:affiliation-property-L-stochastically-increasing} For any $y\in[\underline{v},1]$, $L_v(y)$ is nonincreasing in $v\in[y,1]$.
        \item \label{item:lower-bound-MW-L} For all $\underline{v}\leq t \leq y$ with $y\in V_1$:
    $$ L_y(t) \geq \frac{G_y(t)}{G_y(y)}.$$
    \end{enumerate}
\end{lemma}
\begin{proof}
Properties
\ref{item:affiliation-property-monotonicty-ratio-max-order-pdf-cdf}--\ref{item:affiliation-property-L-stochastically-increasing}
are (explicitly, or implicitly) discussed in~\cite{MW82}.

We next prove Property~\ref{item:lower-bound-MW-L} (which we will be using
critically in the proof of~\cref{lemma:bounds-inverse-bidding} later). Using
Property~\ref{item:affiliation-property-monotonicty-ratio-max-order-pdf-cdf} of
our lemma, we deduce that 
    \begin{align*}
    \int_{t}^y\frac{g_s(s)}{G_s(s)}\, \mathrm{d} s
    \leq \int_{t}^y\frac{g_y(s)}{G_y(s)}\, \mathrm{d} s
    =  \int_{t}^y \left(\frac{\mathrm{d}}{\mathrm{d}s} \ln G_y(s)\right) \, \mathrm{d} s
    =\ln G_y(y) - \ln G_y(t)
    \end{align*}
    and therefore, using its definition from~\eqref{eq:MW-CCFPA-PBNE-L}, we can lower-bound $L_y(t)$ by
    \begin{align*}
      L_y(t) = \exp\left(-\int_{t}^y\frac{g_s(s)}{G_s(s)}\, \mathrm{d}s\right)
      \geq \exp\left(-\ln G_y(y) + \ln G_y(t)\right)
      =\frac{G_y(t)}{G_y(y)}.
    \end{align*}
\end{proof}

\subsection{Canonical CCFPA Equilibrium: Continuity, Concentration, and Computability}
\label{sec:beta-MW-Lipschitch-conentration-bounds}

In this section we establish some key properties of the canonical equilibrium
strategies $\beta$ defined in~\cref{sec:canonical-equilibrium-MW}. We start with
some probability-concentration bounds in~\cref{lemma:bounds-inverse-bidding},
one for the SAPV and one for the IID case. These bounds essentially show that
the probability that a player's equilibrium bids lie within a given interval is
linear with respect to the interval's length, as well as the ``boundedness''
parameters of the underlying prior distribution
(recall~\cref{def:bounded-distributions}). Note that, for these concentration
bounds, we do not make use of any complexity-related notions and the specifics
of our input's representation; the results
in~\cref{lemma:bounds-inverse-bidding} are purely analytical.

We next continue with our complexity considerations, and
in~\cref{lemma:Lipschitz-bounds-beta-SAPV-bounded-IID} we establish the
Lipschitz continuity of the canonical equilibrium strategy, bounding its
Lipschitz constant as a (possibly exponential) function of the auctions
representation. Finally, in~\cref{lemma:inverter-poly-time-oracle-beta-MW} we
show that these equilibrium strategies, which are analytically given by the work
of~\cite{MW82} through~\eqref{eq:MW-CCFPA-PBNE-beta}
and~\eqref{eq:MW-CCFPA-PBNE-L}, can actually also be efficiently computed (with
exponential accuracy).

\begin{lemma}
    \label{lemma:bounds-inverse-bidding}
Consider a CCFPA with $n$ bidders and let $\beta$ be its canonical PBNE (as
described in given in~\cref{sec:canonical-equilibrium-MW}). For any value $v\in
V_1$ in the marginals' support and bids $0\leq b_1\leq b_2 \leq 1$ it holds that
\begin{equation}
    \label{eq:bounds-inverse-bidding-lemma}
\prob{b_1\leq \beta(Y_1) \leq b_2 \fwh{X_1=v} } \leq \gamma \cdot (b_2-b_1),
\end{equation}
where $\gamma\coloneqq 2(n-1)(\priorupperbound/ \priorlowerbound)^2$ for
$(\priorupperbound,\priorlowerbound)$-bounded\footnote{See~\cref{def:bounded-distributions}.}
SAPV settings, and $\gamma\coloneqq n \priorupperbound$ for
$\priorupperbound$-bounded IID settings.
\end{lemma}
\begin{proof}
Fix some $v\in V_1=\support{F_1}$, and let $b_1,b_2\in[0,1]$ with $b_1\leq b2$.
First, we argue that it is enough to prove our lemma for the case where $b_1$
and $b_2$ lie in the image (under the bidding function $\beta$) of the
\emph{same} connected component of the marginals' support. To formalize this,
let $\ssets{I_j}_{j\in[k]}$ be the connected components of $V_1$; that is,
$I_j\subseteq [0,1]$ are intervals, such that $V_1={\dot\union}_{j=1}^k I_j$.
Then, we claim that it is enough to establish
\eqref{eq:bounds-inverse-bidding-lemma} under the assumption that there exists a
$j^*\in[k]$ such that $b_1,b_2\in \beta(I_{j^*})$.
Indeed, assuming this holds, for any $b_1'\leq b_2'$ (not necessarily in the
same connected component of the support) we would have that 
\begin{align*}
\prob{b_1'\leq \beta(Y_1) \leq b_2' \fwh{X_1=v} }
&\leq \sum_{j=1}^k \prob{\beta(Y_1)\in \beta(I_j)\inters[b_1',b_2'] \fwh{X_1=v} }\\ 
&\leq \gamma \cdot \sum_{j=1}^k \mu(I_j\inters[b_1',b_2'])\\
&= \gamma \cdot\mu(V_1\inters[b_1',b_2'])\\
&\leq
\gamma \cdot (b_2'-b_1'),
\end{align*}
where for the first inequality we are using a union bound and the fact that
$\support{G_v}\subseteq V_1$, and in the second and third lines $\mu$
denotes the standard Lebesgue measure in $\R$. Thus, from now on in this proof,
we will assume that there exists an interval $[v_1,v_2]\subseteq V_1$ such that
$b_1,b_2\in \beta([v_1,v_2])$. Recall here, that the canonical equilibrium
strategy $\beta$ is strictly increasing (and thus also invertible) within the
marginals' support\footnote{See~\cref{sec:canonical-equilibrium-MW},
\cref{footnote:beta-strictly-increasing}.}, and therefore this translates to
$\beta^{-1}(b)\in[v_1,v_2]$ for all $b\in[b_1,b_2]$.

\bigskip

We first start with the general SAPV model, assuming further that the priors are
$(\priorlowerbound,\priorupperbound)$-bounded. Observe that, for any $y\in V_1$
and $z\in\interiorset{\support{G_y}}$, we can determine the density $g_y(z)$ by 
\begin{align}
g_y(z) 
&= \lim_{t\to z^-}\frac{G_y(z)-G_y(t)}{z-t}\notag\\
&= \lim_{t\to z^-}\frac{\prob{t\leq Y_1 \leq z \fwh{X_1=y}}}{z-t}\notag\\
&=\frac{1}{f_1(y)}\lim_{t\to z^-}\frac{1}{z-t}\int_{S_z\setminus S_t}f(y,\vec{w})\,\mathrm{d}\vec{w},\label{eq:bounds-density-max-conditional-helper1}
\end{align}
where, for any $x\in[0,1]$, we use $S_x$\label{page:max-order-stat-(n-1)hypercubes} to denote the $(n-1)$-dimensional cube
with edge-length $x$; i.e., $S_x\coloneqq [0,x]^{n-1}$. Note, that the
($(n-1)$-dimensional) volume of this body $S_x$ is equal to $\int_{\vec{w}\in[0,x]^n}1
\,\mathrm{d}\vec{w}=x^{n-1}$, therefore we get the bounds
\begin{align*}
\int_{S_z\setminus S_t}f(y,\vec{w})\,\mathrm{d}\vec{w}
&\leq (z^{n-1}-t^{n-1}) \max_{\vec{w}\in S_z\setminus S_t}f(y,\vec{w})\\
&\leq (z^{n-1}-t^{n-1})\priorupperbound\\
&\leq (n-1)(z-t)z^{n-2} \priorupperbound,
\end{align*}
where for the last inequality we made use
of~\cref{lemma:bounds-difference-of-powers}. Similarly, we can get the lower
bound:
\begin{equation*}
    \label{eq:lower-bound-derivatives-g-integral}
\int_{S_z\setminus S_t}f(y,\vec{w})\,\mathrm{d}\vec{w}
\geq (n-1)(z-t)t^{n-2} \priorlowerbound.
\end{equation*}
Using these, we can bound the density in~\eqref{eq:bounds-density-max-conditional-helper1} by:
\begin{equation}
    \label{eq:bounds-density-max-conditional-helper3}
    \frac{(n-1)z^{n-2}}{f_1(y)}\priorlowerbound 
    = \frac{1}{f_1(y)} \lim_{t\to z^-}\left[(n-1)t^{n-2} \priorlowerbound\right]
    \leq
    g_y(z)
    \leq \frac{1}{f_1(y)} \lim_{t\to z^-}\left[(n-1)z^{n-2} \priorupperbound\right]
    =\frac{(n-1)z^{n-2}}{f_1(y)}\priorupperbound.
\end{equation}

Next, we turn our attention to bounding the derivative of the canonical
equilibrium strategy $\beta$. For any $y\in V_1$ we can get that:
\begin{align}
    \beta'(y) 
    &= [y-\beta(y)]\frac{g_y(y)}{G_y(y)}, &&\text{from~\eqref{eq:MW-beta-differential-equation-def},}\notag\\
    &= \frac{g_y(y)}{G_y(y)}\int_{\underline{v}}^yL_y(t)\, \mathrm{d} t, &&\text{from~\eqref{eq:MW-CCFPA-PBNE-beta},} \label{eq:bounds-density-max-conditional-helper5}\\
    &\geq \frac{g_y(y)}{G_y^2(y)} \int_{\underline{v}}^yG_y(t)\, \mathrm{d} t, &&\text{from Property~\ref{item:lower-bound-MW-L} of \cref{prop:affiliation-properties}},\label{eq:bounds-density-max-conditional-helper4}\\
    &\geq \frac{g_y(y)}{G_y^2(y)} \frac{G^2_y(y)-G^2_y(\underline{v})}{2\max_{t\leq y}{g_y(t)}}, &&\text{from~\cref{lemma:lower-bound-integral-derivative}},\notag\\
    &= \frac{g_y(y)}{2\max_{t\leq y}{g_y(t)}},\label{eq:bounds-density-max-conditional-helper2}
\end{align}
where in the last step we used the fact that $G_y(\underline{v})=0$, since
$\support{G_y}\subseteq V_1$.

Now let's denote
by $h_v$ the density of the random variable $\beta(Y_1)$, conditioned on
$X_1=v$. Then, for any $b\in[b_1,b_2]$, we can write (see, e.g.,
\cite[Theorem~7.1, Sec.~5.7]{ross2010}):
\begin{equation}
    \label{eq:density-beta-of-max}
h_v(b) 
= g_v(\beta^{-1}(b))\frac{\mathrm{d}}{\mathrm{d} b} \beta^{-1}(b)
= \frac{g_v(\beta^{-1}(b))}{\beta'(\beta^{-1}(b))}
= \frac{g_v(y)}{\beta'(y)},
\end{equation}
where we set $y=\beta^{-1}(b)$ for convenience. Therefore, we can utilize the
$(\priorlowerbound,\priorupperbound)$-boundedness of the value distribution, together with the bounds from
\eqref{eq:bounds-density-max-conditional-helper2} and
\eqref{eq:bounds-density-max-conditional-helper3}, to get that, for all
$b\in[b_1,b_2]$ it holds that:
\begin{align*}
h_v(b)
&\leq \frac{g_v(y)}{\frac{g_y(y)}{2\max_{t\leq y}{g_y(t)}}}, &&\text{from~\eqref{eq:density-beta-of-max} and~\eqref{eq:bounds-density-max-conditional-helper2}},\\ 
&= 2 \frac{g_v(y)}{g_y(y)}\max_{t\leq y}{g_y(t)}\\
&\leq 2 \frac{\frac{(n-1)z^{n-2}}{f_1(y)}\priorupperbound}{\frac{(n-1)z^{n-2}}{f_1(y)}\priorlowerbound}\max_{t\leq y}\frac{(n-1)t^{n-2}}{f_1(y)}\priorupperbound, &&\text{from~\eqref{eq:bounds-density-max-conditional-helper3}},\\
&=2\frac{\phi_2}{\phi_1}\cdot \frac{n-1}{f_1(y)}\priorupperbound \max_{t\leq y} t^{n-2}\\
&\leq 2(n-1)\left(\frac{\phi_2}{\phi_1}\right)^2, &&\text{since $f_1(y)\geq \priorlowerbound$ and $y \leq 1$}.
\end{align*}
the last inequality holding due to the fact that
$f_1(y)=\int_{\vec{w}\in[0,1]^{n-1}} f(y,\vec{w})\,\mathrm{d}\vec{w}\geq
\priorlowerbound$ (given the full-support assumption for prior $f$, under our
bounded SAPV model), and due to $y \leq 1$. Using this, we can finally bound the
desired probability in the statement~\eqref{eq:bounds-inverse-bidding-lemma} of
our lemma by
\begin{equation*}
\prob{b_1\leq \beta(Y_1) \leq b_2 \fwh{X_1=v} }
\leq \int_{b_1}^{b_2} h_v(b)\,\mathrm{d} b
\leq (b_2-b_1)\cdot 2(n-1)(\priorupperbound/\priorlowerbound)^2.
\end{equation*}

\bigskip
Now let's consider the remaining case, where the values are IID according to
$\priorupperbound$-bounded value distribution $F_1$.
Recall\footnote{See~\cref{sec:CCFPA-appendix-technical},
Page~\ref{page:iid-max-order-statistics}.} that, in this case, for any $x\in
V_1$ we have that the conditional maximum order statistic distributions have cdf
and pdf, respectively: $G_x(z) = G(z) = F_1^{n-1}(z)$ and $g_x(z)= g(z)
=(n-1)f_1(z)F_1^{n-2}(z)$.
So, we can now improve the bound for the derivative $\beta'$ of the canonical
equilibrium strategy in~\eqref{eq:bounds-density-max-conditional-helper4} by:
\begin{align*}
\beta'(y)&\geq \frac{g(y)}{G^2(y)} \int_{\underline{v}}^yG(t)\, \mathrm{d} t\\
    &= \frac{(n-1)f_1(y)F_1^{n-2}(y)}{F_1^{2(n-1)}(y)} \int_{\underline{v}}^yF_1^{n-1}(t)\, \mathrm{d} t\\
    &= \frac{(n-1)f_1(y)}{n F_1^{n}(y)} \int_{\underline{v}}^y n F_1^{n-1}(t)\, \mathrm{d} t\\
    &\geq \frac{(n-1)f_1(y)}{n \priorupperbound F_1^n(y)} \int_{\underline{v}}^y n f_1(t)F_1^{n-1}(t)\, \mathrm{d} t, &&\text{due to $\priorupperbound$-boundedness},\\
    &= \frac{(n-1)f_1(y)}{n \priorupperbound F_1^n(y)} \int_{\underline{v}}^y \left[F_1^{n}(t)\right]'\, \mathrm{d} t\\
    &= \frac{(n-1)f_1(y)}{n \priorupperbound F_1^n(y)} F_1^{n}(y),\\
    &=\frac{n-1}{n}\frac{f_1(y)}{\priorupperbound}.
\end{align*}
Therefore, we can now bound $h_v(b)$ from~\eqref{eq:density-beta-of-max} by:
$$
h_v(b) 
=
\frac{g(y)}{\beta'(y)}
\leq \frac{(n-1)f_1(y)F_1^{n-2}(y)}{\frac{n-1}{n}\frac{f_1(y)}{\priorupperbound}}
= n \priorupperbound F_1^{n-2}(y)
\leq n \priorupperbound.
$$
\end{proof}

\begin{lemma}
    \label{lemma:Lipschitz-bounds-beta-SAPV-bounded-IID}
    For both the IID and the (full-support) SAPV settings, the canonical CCFPA
    equilibrium strategy (see~\cref{sec:canonical-equilibrium-MW}) is Lipschitz
    continuous,\footnote{Let $E\subseteq \R$ and $c\geq 0$. We say that a
    function $g:E\map\R$ is \emph{Lipschitz} continuous (on $E$) with
    \emph{Lipschitz constant} $c$, if for all $x,y\in E$ it holds that
    $\card{f(x)-f(y)} \leq c\cdot \card{x-y}$. For more background, see
    e.g.~\cite[Sec.~1.6]{RoydenFitzpatrick2010}} with a Lipschitz constant which
    is at most \emph{exponential} in the (binary) representation
    (see~\cref{sec:CFPA-represent}) of the auction.
    \end{lemma}
    \begin{proof}
    See~\cref{app:densification-proofs}.
\end{proof}

\begin{lemma}
    \label{lemma:inverter-poly-time-oracle-beta-MW}
    For both the SAPV and IID settings, the values of the canonical CCFPA
    equilibrium strategy (see~\cref{sec:canonical-equilibrium-MW}) can be
    exactly computed in polynomial time. That is, given an auction $\mathcal A$
    and an $x\in[\underline{v},1]$, (the binary representation of) $\beta(x)$ can
    be computed in time polynomial in the (binary) representations of $x$ and
    $\mathcal{A}$. 
    \end{lemma}
\begin{proof}
    See~\cref{app:densification-proofs}.
\end{proof}

\subsection{From Exact CCFPA Equilibria to Approximate CFPA Equilibria}
We are now ready to finalize the proof of our main ``bid densification'' result,
namely~\cref{th:approximate-PBNE-CFPA-to-CCFPA}. They key idea behind of our
approach is to study CFPAs as discrete approximations of their
bull-bidding-space CCFPA counterparts, the quality of the approximation
depending on the granularity of the original, discrete bidding space. This
technique is formalized in two steps: first,
in~\cref{lemma:approx-equilibrium-discretizing-bidding-CCFPA} we prove that a
bidding strategy which is ``near'' to the canonical \emph{exact} PBNE of a
CCFPA, must itself be an \emph{approximate} equilibrium; then,
in~\cref{lemma:approximate-inverter} we show how such an approximation of the
canonical equilibrium strategy, within the desired discrete subset of bids of
our CFPA, can indeed be constructed in polynomial time.

\begin{definition}[Approximation of bidding strategies]
\label{def:near-bidding-strategies}
    Let $\beta_1,\beta_2:[0,1]\to[0,1]$ be two bidding strategies of a symmetric
    CFPA/CCFPA with marginal prior density $F_1$. We will say that $\beta_2$ is
    an \emph{$\varepsilon$-under\-appro\-ximation} if
    $$
    \beta_1(v)-\varepsilon \leq \beta_2(v) \leq \beta_1(v) \qquad \forall v\in\support{F_1}.
    $$
\end{definition}

\begin{lemma}
    \label{lemma:approx-equilibrium-discretizing-bidding-CCFPA}
    Consider a CCFPA with $n$ bidders and let $\beta$ be its canonical PBNE (as
    described in~\cref{sec:canonical-equilibrium-MW}). Let $\tilde{\beta}$ be an
    other nondecreasing (but possibly discontinuous) bidding strategy, which is
    an $\tilde{\varepsilon}$-underapproximation of $\beta$. Then,
    $\tilde{\beta}$ is a $2\gamma\tilde{\varepsilon}$-approximate PBNE, where
    $\gamma \coloneqq 2(n-1)(\priorupperbound/ \priorlowerbound)^2$ for
    $(\priorupperbound,\priorlowerbound)$-bounded SAPV settings, and $\gamma
    \coloneqq n \priorupperbound$ for $\priorupperbound$-bounded IID settings.
\end{lemma}
\begin{proof}
Throughout this proof, we use $H_v(b,\hat{\beta})$ to denote the probability of
a bidder winning the auction, conditioned on her true value being $v\in
V_1=\support{F_1}$, when she bids $b\in[0,1]$ while all other bidders play
according to the same (possibly discontinuous) nondecreasing bidding strategy
$\hat{\beta}:V_1\to[0,1]$. Also, we let $u_v(b,\hat{\beta}) =
(v-b)H_v(b,\hat{\beta})$ denote the corresponding (interim) expected utility of
the bidder.

We start by observing that, for the (possibly discontinuous) bidding strategy
$\tilde{\beta}$ in the statement of our lemma, we have the bounds
\begin{equation*}
    \prob{\tilde{\beta}(Y_1) < b\fwh{X_1=v}} 
    \leq H_v(b,\tilde{\beta}) 
    \leq \prob{\tilde{\beta}(Y_1) \leq b\fwh{X_1=v}}.
\end{equation*}
These are a direct consequence of the observations that (i) beating all other
bidders is a \emph{sufficient} condition for winning the item, while (ii)
bidding at least as high as the others is a \emph{necessary} condition for
winning.
Since $\beta-\tilde{\varepsilon}\leq \tilde{\beta} \leq \beta$ within the $V_1$
and $\support{G_v}\subseteq V_1$ (i.e., the support of $Y_1$ when conditioned on
$X_1=v$ is a subset of the value marginals' support), the above inequality gives
\begin{equation}
    \label{eq:bounds-H-CCFPA-general}
    G_v(b)=\prob{\beta(Y_1) < b\fwh{X_1=v}} 
    \leq H_v(b,\tilde{\beta}) 
    \leq \prob{\beta(Y_1) + \tilde{\varepsilon} \leq b\fwh{X_1=v}}=G_v(b+\tilde{\varepsilon}).
\end{equation}

Next, recall that the canonical equilibrium strategy $\beta$ is \emph{strictly}
increasing in the support $V_1$ and thus, since $\support{G_v}\subseteq
\support{F_1}$, random variable $\beta(Y_1)$ (when conditioned on $X_1=v$) is
also strictly increasing (in its support). Therefore, for any $b\in [0,1]$ we
can write
\begin{equation}
    \label{eq:bounds-H-CCFPA-continuous}
     H_v(b,\beta) 
    = \prob{\beta(Y_1) \leq b\fwh{X_1=v}}= G_v(b).
\end{equation}

We will now prove that the following two inequalities hold, for any $v\in V_1$:
\begin{align}
u_v(\beta(v),\beta) &\leq u_v(\tilde{\beta}(v),\tilde{\beta}) + \gamma\tilde{\varepsilon}, \label{eq:utilities-approx-densification-a} \\
\intertext{and}
 u_v(b,\beta) &\geq u_v(b,\tilde{\beta}) - \tilde{\varepsilon}, &&\text{for all}\;\; b\in [0,1]. \label{eq:utilities-approx-densification-b} 
\end{align}
For \eqref{eq:utilities-approx-densification-a}, first, we have:
\begin{align*}
    u_v(\beta(v),\beta) - u_v(\tilde{\beta}(v),\tilde{\beta}) 
    &= (v-\beta(v))H_v(\beta(v),\beta) - (v-\tilde{\beta}(v))H_v(\tilde{\beta}(v),\tilde{\beta})\\
    &\leq (v-\beta(v))H_v(\beta(v),\beta) - (v-\beta(v))H_v(\tilde{\beta}(v),\tilde{\beta}), &&\text{since $0\leq\tilde{\beta}\leq \beta$},\\
    &\leq (v-\beta(v))G_v(\beta(v)) - (v-\beta(v))G_v(\tilde{\beta}(v)), &&\text{from \eqref{eq:bounds-H-CCFPA-continuous} and \eqref{eq:bounds-H-CCFPA-general},}\\
    &= (v-\beta(v))[G_v(\beta(v)) - G_v(\tilde{\beta}(v))]\\
    &\leq (v-\beta(v))\cdot \gamma[\beta(v)-\tilde{\beta}(v)], &&\text{from~\cref{lemma:bounds-inverse-bidding}},\\
    &= \gamma\tilde{\varepsilon}, &&\text{since $0\leq v-\beta(v) \leq 1$, $\beta-\tilde{\varepsilon}\leq \tilde{\beta}$}.
\end{align*}

For~\eqref{eq:utilities-approx-densification-b}:
\begin{align*}
u_v(b,\tilde{\beta})- u_v(b,\beta) &= (v-b) H_v(b,\tilde{\beta}) - (v-b)H_v(b,\beta)\\
    &\leq (v-b) \left[ G_v(b+\tilde{\varepsilon}) - G_v(b)\right], &&\text{from \eqref{eq:bounds-H-CCFPA-general}, \eqref{eq:bounds-H-CCFPA-continuous},}\\
    &=(\beta(v)-b) \cdot \gamma\tilde{\varepsilon}, &&\text{from~\cref{lemma:bounds-inverse-bidding}},\\
    &\leq \gamma \tilde{\varepsilon} &&\text{since $\beta(v),b\in [0,1]$}.
\end{align*}

Now we are ready to establish that indeed $\tilde{\beta}$ is a
$2\gamma\tilde{\varepsilon}$-approximate symmetric PBNE. For any $v\in V_1$, and
any $b\in[0,1]$ it is:
\begin{equation*}
u_v(\tilde{\beta(v)},\tilde{\beta}) 
\overset{\eqref{eq:utilities-approx-densification-a}}{\geq} u_v(\beta(v),\beta) -\gamma\tilde{\varepsilon}
\geq u_v(b,\beta) -\gamma\tilde{\varepsilon}
\overset{\eqref{eq:utilities-approx-densification-b}}{\geq}
u_v(b,\tilde{\beta}) -\gamma\tilde{\varepsilon} -\gamma\tilde{\varepsilon},
\end{equation*}
where the middle inequality is due to the fact that $\beta$ is an (exact) PBNE.
\end{proof}

\begin{lemma} [Approximate Inverter]
    \label{lemma:approximate-inverter}
    Consider a CCFPA with $n$ bidders and let $\beta$ be its canonical PBNE (as
    described in~\cref{sec:canonical-equilibrium-MW}). Fix a (nonempty)
    \emph{finite}
    $\delta$-dense\footnote{See~\cref{def:denseness-bidding-space}.} subset of
    bids $B\subseteq [0,1]$. For any $\varepsilon>0$, we can compute (the
    standard, step-function representation of) a nondecreasing bidding strategy
    $\tilde{\beta}:[\underline{v},1] \to B$ which is a
    $(\delta+2\varepsilon)$-underapproximation of $\beta$, in time polynomial in
    $\log(1/\varepsilon)$ and the description of the auction.
\end{lemma}
\begin{proof}
The most critical component of our proof, is establishing that we can
efficiently ``invert'' the canonical equilibrium strategy $\beta$. That is,
given an $\varepsilon>0$ and a possible bid $b\in\beta(V_1)$\footnote{Here we
are using our standard notation of $V_1=\support{F_1}$ denoting the support of
the distribution of the value marginals $F_1$. Recall also that $\underline{v}$
is the leftmost point of the support $V_1$. Finally, notice that, since $\beta$
is continuous in $[\underline{v},1]$ and constant in all intervals of
$[\underline{v},1]\setminus V_1$ that are outside the marginal's support
(see~\pageref{page:beta-absolutely-continuous}), it must be that
$\beta(V_1)=\beta([\underline{v},1])$.}, we can find, in time polynomial in
$\log(1/\varepsilon)$, a value $x\in [0,1]$ such that $\beta(x)$ is
$\varepsilon$-near the given bid $b$. We will achieve that, by performing binary
search in the feasible value domain $[\underline{v},1]$; for convenience, let's
also define $\xi\coloneq (1-\underline{v}) \leq 1$ for the rest of this proof.  

Indeed, given that $\beta$ is nondecreasing in $[\underline{v},1]$ and that,
by~\cref{lemma:inverter-poly-time-oracle-beta-MW}, we have an efficient oracle
for the values of $\beta$, by performing $k$ steps of such binary search,
$k=1,2,3,...$, we can construct a sequence $y_0(b),y_1(b),\dots,y_k(b)$, with
$y_0(b)=\xi/2$, with the property that there is guaranteed to exist a value
$x\in[y_k(b)-\xi/2^{k+1},y_k(b)+\xi/2^{k+1}]$ such that $b\leq \beta(x)\leq
b+\varepsilon$. Thus, if we choose the rightmost point of  this interval, 
$$
s_k(b;k) \coloneqq y_k(b)-\xi/2^{k+1},
$$
due to the monotonicity and the Lipschitz continuity of $\beta$, established
in~\cref{lemma:Lipschitz-bounds-beta-SAPV-bounded-IID}, we can guarantee that 
\begin{align}
   \beta(s(b;k)) -b &\geq \beta(x) -b \geq 0 \label{eq:upper-bound-binary-search-1}\\
\intertext{and}
\beta(s(b;k)) - b
 &\leq \card{\beta(s(b;k))-\beta(x)}+ \card{\beta(x)-b}
\leq L_{\beta} \card{s(b;k)-x} + \varepsilon
\leq L_{\beta} \frac{\xi}{2^k} + \varepsilon, \label{eq:upper-bound-binary-search-2}
\end{align}
where $L_\beta$ is the Lipschitz constant of $\beta$. Recall that,
from~\cref{lemma:Lipschitz-bounds-beta-SAPV-bounded-IID}, the value of $L_\beta$
is at most exponential in the (binary) representation of our auction. Thus, by
choosing a sufficiently large, but still polynomial in the size of our input and
$\log(1/\varepsilon)$, number of steps $k=k_{\varepsilon}$, we can make the
upper bound in~\cref{eq:upper-bound-binary-search-2} to be at most
$2\varepsilon$. In the following let's simply denote by $s(b)\coloneqq
s(b;k_{\varepsilon})$ such an (efficiently computable) value that achieves this
bound (for a sufficiently large $k_{\varepsilon}$); we will refer to this
function $s:\beta(V_1)\map [0,1]$ as \emph{approximate inverter}.

Now let $B=\ssets{0=b_0<b_1<b_2<\dots<b_m}$ be the bids of our finite,
$\delta$-dense bidding set $B$. Using the above efficient algorithm, we can
compute all values $s(b)$, for all bids $b\in B\inters\beta(V_1)$ of our
restricted bidding space $B$, in order to define a step-function bidding
strategy $\tilde{\beta}:[\underline{v},1]\to B$ by using these $s(b)$'s as break
points; that is, under $\tilde{\beta}$, a player switches to bid $b\in B$ when
she reaches value $s(b)$ (recall our bid-function output representation
from~\cref{sec:CFPA-represent}). 

More precisely, let $J\coloneqq\sset{j\in[m]\fwh{b_j\in b(V_1)}}$ be the indices
of all (non-zero probability) positive bids. Notice that, due to the continuity
of the canonical equilibrium strategy $\beta$, $J$ is a set of
\emph{consecutive} integers; let $\underline{m}\coloneqq \min J$ and
$\overline{m}\coloneqq \max J$, i.e., $J=[\underline{m},\overline{m}]$.  For a
$j\in J$, for simplicity we denote $s_j\coloneqq s(b_j)$, where $s$ is the
aforementioned, efficiently computable, approximate inverter function. Then, the
bidding strategy $\tilde{\beta}$ is formally defined by
\begin{equation}
    \label{eq:def-underapproximation-inverter}
\tilde{\beta}(x)= 
\begin{cases}
b_{\underline{m}-1}, & \text{if}\;\; x\in [\underline{v},s_{\underline{m}}],\\
b_{j}, & \text{if}\;\; x\in (s_j,s_{j+1}],\;\; \underline{m}\leq j <\overline{m}, \\
b_{\overline{m}}, & \text{if}\;\; x\in [s_{\overline{m}},1].
\end{cases}
\end{equation}

Next, we argue that $\tilde{\beta}$ is always below the canonical equilibrium
strategy. Indeed, take for example a value $x\in (s_j,s_{j+1}]$ with
$\underline{m}\leq j <\overline{m}$ (the remaining cases of the definition of
$\tilde{\beta}$ in~\eqref{eq:def-underapproximation-inverter} can be handled
similarly). Then 
$$\tilde{\beta}(x)=b_j \leq \beta(s(b_j)) = \beta(s_j) \leq \beta (x),$$ where
the first inequality is due to~\eqref{eq:upper-bound-binary-search-1} and the
second inequality is due to the monotonicity of $\beta$.

Finally, to argue that $\tilde{\beta}$ is a $(\delta+2\varepsilon)$-near the
original bidding strategy $\beta$, we first observe that, since $\tilde{\beta}$
is a step-function, its maximum distance to $\beta$ is realized on the jump
points $\sset{s_j}_{j\in J}$, or at the rightmost point of our $[0,1]$. Also,
since the set $B$ is $\delta$-dense in $[0,1]$, it is also $\delta$-dense in
$\beta([\underline{v},1])\subseteq[0,1]$. Therefore, the distance of the two
functions can be upper-bounded by
$$
\sup_{x\in [\underline{v},1]}\cards{\tilde{\beta}(x)-\beta(x)}
\leq \delta+ \max_{j\in J}\cards{b_j-\beta(s_j)}
\leq \delta + \max_{b\in B}\cards{b-\beta(s(b))}
\leq 2\varepsilon+\delta,
$$
where the last inequality is due to the definition of our approximate inverter
and its approximation upper bound in~\eqref{eq:upper-bound-binary-search-2}.
\end{proof}

\bigskip
We can now put all the pieces together, and prove the main result of this section:
\begin{proof}[Proof of~\cref{th:approximate-PBNE-CFPA-to-CCFPA}]
    Let $B\subseteq$ be the $\delta$-dense bidding set of the CFPA. We first
    use~\cref{lemma:approximate-inverter} in order to construct, in
    $\poly(\log(1/\varepsilon))$ time, a nondecreasing bidding strategy
    $\tilde{\beta}:[0,1]\to B$ which is
    $(\delta+2\varepsilon)$-underapproximation of the canonical equilibrium
    strategy $\beta$ (as described in~\cref{sec:canonical-equilibrium-MW}) of
    the CCFPA extension of our auction (from $B$ to $[0,1]$). Clearly, since
    $\beta$ is no-overbidding, $\tilde{\beta}$ is no-overbidding as well.
    Applying~\cref{lemma:approx-equilibrium-discretizing-bidding-CCFPA} with
    $\tilde{\varepsilon}\gets \delta+2\varepsilon$, gives us our desired
    approximation parameters.
\end{proof}
\begin{remark}
\label{remark:discrete-approx-CCFPA}
    Although our main result of this section, namely
    \cref{th:approximate-PBNE-CFPA-to-CCFPA}, refers to the computability of
    approximate equilibria for the CFPA (i.e., for a discrete bidding space
    $B$), one can use our intermediate tools within its proof, to derive some
    interesting approximation results for continuous bids as well, i.e., for the
    CCFPA. More precisely, our ``approximate inverter''
    \cref{lemma:approximate-inverter} essentially uses the discrete bidding
    space $B$ just to define a piecewise-constant bidding function which
    approximates the \cite{MW82} equilibrium
    formula~\eqref{eq:MW-CCFPA-PBNE-beta}; this is still a ``legitimate''
    strategy for the extended CCFPA that allows bidding in the entire
    interval [0,1], just restricted on $B$. Then, we
    deploy~\cref{lemma:approx-equilibrium-discretizing-bidding-CCFPA} to argue
    that this step-function constitutes an approximate BNE of the
    CFPA with continuous bids in $[0,1]$.

    In our paper, since we are studying the discrete bid setting, we \emph{must}
    use the fixed discrete bidding space $B$ that is given to us as input.
    However, if we are studying the continuous bid setting instead, we can
    choose any discretization (e.g., equidistant bids with granularity $\delta$)
    that we desire, and obtain a polynomial-time algorithm for CCFPA, with
    running time that will depend on $\delta$.
\end{remark}

\section{Discussion and Future Work}\label{sec:conclusion}

In this work, we have made significant progress in understanding the complexity
of computing equilibria in the first-price auction when the values of the
bidders are correlated. We firmly believe that our results bring us a step
closer to the ``holy grail'' of this literature, namely an answer about the
complexity of the auction with IPV. Below we state some perhaps more tangible
open problems that are directly associated with our work. 

The first interesting direction is to consider the computational complexity of
computing MBNE of the DFPA, even non-monotone ones. The existence results
presented in \cref{sec:existence} imply that this is a \emph{total search
problem}, and hence the appropriate candidate classes for its complexity would
be the subclasses of TFNP. While the corresponding problem with subjective
priors is PPAD-complete \citep{fghk24}, in the case of correlated values, the
true complexity of the problem could even lie somewhere lower in the TFNP
hierarchy. We state the associated open problem below.

\begin{opproblem}
What is the computational complexity of computing MBNE of the DFPA?
\end{opproblem}
\noindent One could also ask similar questions about monotone PBNE of the CFPA
or MBNE of the DFPA for affiliated values, but these will naturally be harder to
prove.

The second interesting open problem that stems from our work is whether we can
extend the results obtained via our bid densification technique to obtain
approximation algorithms for the MBNE of the DFPA as well. While our
\cref{lem:discrete-to-continuous-and-back-for-existence} seems like a very
useful tool to achieve that, the continuous distribution that it generates has
parts with zero density, and hence \cref{th:approximate-PBNE-CFPA-to-CCFPA} does
not apply. As we mentioned in \cref{sec:existence}, smoothening the density will
violate the affiliation condition, so it seems that the only way to deal with
this is to extend \cref{th:approximate-PBNE-CFPA-to-CCFPA} to distributions that
do not have full support. In fact, most of the machinery that we develop in
\cref{sec:densification} is already capable of achieving that, but our attempts
of generalization fell short on being able to bound certain quantities that have
to do with the value distribution. We believe that the desired generalization
should be possible, but that it would require rather involved and diligent
technical work. We state the second open problem below.

\begin{opproblem}
Can we extend the algorithm of \cref{th:approximate-PBNE-CFPA-to-CCFPA} to apply
to SAPV settings without full support, and as a result also obtain a similar
approximation algorithm for the MBNE of the DFPA as well?
\end{opproblem}

\clearpage
\appendix

\addcontentsline{toc}{section}{Appendix}
\section*{Appendix}

\section{Proof of \texorpdfstring{\cref{lem:efficient-computation}}{Proposition
\ref*{lem:efficient-computation}}}\label{app:utilities} In this section of the
appendix, we prove \cref{lem:efficient-computation}, namely the efficient
computation of the utilities.

\subsection{Efficient Utility Computation in Discrete First-Price
Auctions}\label{app:dfpa-utility} In this section, we provide the proof of
\cref{lem:efficient-computation} in the DFPA setting.

\begin{lemma}\label{lem:app-efficient-utilities} Fix a DFPA with correlated
    priors. For any bidder $i \in N$, any value $v_i \in V_i$, and any mixed
    strategy profile of the other bidders $\vec{\beta_{-i}}$, the utility of
    $i$, $u_i(\vec{\gamma},\vec{\beta_{-i}};v_i)$, when playing some
    distribution $\vec{\gamma}$ over her bids, is computable in polynomial time.
    Additionally, the utility is still efficiently computable in the $k$-GSAPV (in particular, SAPV)
    setting, where the joint prior is succinctly represented, when $\vec{\beta}_{-i}$ is symmetric.
\end{lemma}
\begin{proof}
    Without loss of generality, consider a reordering of the bidders such that
    we are computing the utility of the last one (bidder $n=\card{N}$). We can
    use the definition of a mixed strategy to express the utility of $n$ when
    picking a distribution $\vec{\gamma}$ over the bids as
    $u_n(\vec{\gamma},\vec{\beta_{-n}};v_n)=\sum_{b \in B}\vec{\gamma}(b)
    u_n(b,\vec{\beta_{-n}};v_n)$. Thus, for the left-hand side to be efficiently
    computable, it suffices to prove that all of the $\card{B}$ many summands
    are efficiently computable. By the definition of utility
    in~\eqref{eq:utility}, we can express the utility of bidder $n$ when playing
    a bid $b$ as
    $u_n(b,\vec{\beta_{-n}};v_n)=(v_n-b)H_{n}(b,\vec{\beta_{-n}};v_n)$. We now
    proceed to show how to efficiently compute $H_{n}$, by first expressing it
    as:
    \begin{equation}\label{eq:h-fcn-explicit}
        H_{n}(b,\vec{\beta_{-n}};v_n)=\sum_{r=0}^{n-1} \frac{1}{r+1} \cdot \sum_{\vec{v_{-n}}\in\support{F_{n\mid v_n}}} f_{n \mid v_n}(\vec{v_{-n}}) \cdot T_n(b,n-1,r,\vec{v_{-n}})
    \end{equation}
    where, for $0 \leq r \leq \ell \leq n-1$, $T_n(b,n-1,r,\vec{v_{-n}})$
    denotes the probability that exactly $r$ out of the remaining $n-1$ bidders
    bid exactly $b$ and all of the others bid strictly less, when the values of
    the bidders other than $n$ are $\vec{v_{-n}}$. The randomness here stems
    from the mixed strategies $\vec{\beta_{-n}}$ of the bidders.

    We can efficiently compute the support of $F_{n\mid v_n}$ by checking
    whether $v_n=t^j_n$ for each $\vec{t^j}$ in the representation of the
    distribution. Additionally, we can efficiently compute $f_{n \mid
    v_n}(\vec{v_{-n}})$ using the definition of the conditional distribution
    in~\eqref{eq:discrete-conditional}, as we can calculate the marginal
    $f_n(v_n)$ by summing over the points in the support of the distribution
    which we just computed. It remains to show how to efficiently compute each
    $T_n(b,n-1,r,\vec{v_{-n}})$. To make notation easier to follow, it is useful
    to define the following two quantities:
    \begin{equation}\label{eq:g-utility-dfpa}
        g_{j,b} \coloneq \beta_j(v_j)(b)
    \end{equation}
    \begin{equation}\label{eq:G-utility-dfpa}
        G_{j,b} \coloneq \sum_{\substack{b'\in B, \\ b'<b}} g_{j,b'}
    \end{equation}
    where, for any bidder $j$ (who must have value $v_j$
    from~\eqref{eq:h-fcn-explicit}), $g_{j,b}$ represents the probability $j$
    picks bid $b$ and $G_{j,b}$ the probability that she picks some bid strictly
    smaller than $b$, always conditioned on the fact that bidder $n$ has value
    $v_n$. We can then express $T_n(b,n-1,r,\vec{v_{-n}})$ as:
    \begin{equation}\label{eq:winning-prob-explicit}
        T_n(b,n-1,r,\vec{v_{-n}}) = \sum_{\substack{S \subseteq [n-1], \\ \card{S}=r}} \prod_{s \in S} g_{s,b} \prod_{s \in [n-1]\setminus S} G_{s,b}
    \end{equation}

    Notice that in~\eqref{eq:winning-prob-explicit} we are summing over all
    possible subsets of $[n-1]$ of size $r$, which are exponentially many.
    Hence, this is not an efficient way to compute the utility. Instead, we
    proceed with a dynamic programming approach, of similar nature to previous
    work in \cite{fghk24}. We define the DP as follows:
    \begin{align*}
    	T_n(b,0,0,\vec{v_{-n}}) &=1;              &\\
    	T_n(b,\ell,k,\vec{v_{-n}}) &=0,&\quad\text{for}\;\; k>\ell;\\
    	T_n(b,\ell+1,0,\vec{v_{-n}}) &=T_n(b,\ell,0,\vec{v_{-n}})G_{\ell+1,b};&\\
    	T_n(b,\ell+1,k+1,\vec{v_{-n}}) &=T_n(b,\ell,k,\vec{v_{-n}})g_{\ell+1,b} + T_n(b,\ell,k+1,\vec{v_{-n}})G_{\ell+1,b}; &\text{for}\;\; k\leq\ell.
    \end{align*}

    Using this, we can compute any $T_n(b,n-1,k,\vec{v_{-n}})$ for $k \in
    0,1,\ldots,n-1$ with a total number of $O(n^2)$ recursive calls. Therefore,
    following the previous steps in this proof we see that
    $u_n(\vec{\gamma},\vec{\beta_{-n}};v_n)$ is computable in polynomial time.

    We now move to the $k$-GSAPV setting.
    Let $\mathcal{R}=\{(\vec{t^1},p_1),\ldots,(\vec{t^\ell},p_\ell)\}$ be the succinct representation of $F$, as described in \cref{sec:representation-dfpa}.
    Let $\vec{t^j_{-n}}$ be the vector corresponding from removing the first entry (among the ones that correspond to bidder $n$'s group) that matches $v_n$ in $\vec{t^j}$.
    Using the properties of the representation, we can express the conditional distribution at some point $\vec{t^j_{-n}}$ of the support as:
    \begin{equation}\label{eq:app-conditional-representative}
        f_{n \mid v_n}(\vec{t^j_{-n}}) = \frac{p_j}{f_n(v_n)}
    \end{equation}
    where $f_n(v_n)$ denotes the marginal distribution of bidder $n$ evaluated at $v_n$. Notice
    that the marginal distribution can be efficiently computed using the
    properties of the succinct representation by appropriately summing over all
    elements in its support where $v_n$ appears in one of the entries corresponding to $n$'s group. 
    We will then rewrite~\eqref{eq:h-fcn-explicit} as:
    \begin{equation}
        H_{n}(b,\vec{\beta_{-n}};v_n)=\sum_{r=0}^{n-1} \frac{1}{r+1} \cdot \sum_{\substack{(\vec{t^j},p_j)\in \mathcal{R}: \\\vec{t^j_{-n}}\in \support{F_{n \mid v_n}} }} m_j \cdot f_{n \mid v_n}(\vec{t^j_{-n}}) \cdot T_n(b,n-1,r,\vec{t^j_{-n}}) 
    \end{equation} 
    where $m_j$ counts the number of group-valid permutations.
    To compute $m_j$, assume, without loss of generality, that the bidder $n$ belongs in the last group (the $k$-th). Then, $m_j = n_1! n_2!\ldots n_{k-1}!(n_k-1)!$.

    In the above, we used the fact that, due to the symmetry of the bidding strategies $\vec{\beta_{-n}}$, $T_n(b,n-1,r,\vec{v})=T_n(b,n-1,r,\vec{v'})$ for any $\vec{v'}$ that is a group-valid permutation of $\vec{v}$.
    The result that the computation of the utilities takes polynomial time comes from the following facts:
    \begin{itemize}
        \item [-] We can check whether $\vec{t^j_{-n}}\in \support{F_{n \mid v_n}}$ in polynomial time, by simply checking whether $v_n$ appears in $\vec{t^j}$ in the entries corresponding to bidder $n$'s group.
        \item [-] We can efficiently compute $T_n(b,n-1,r,\vec{t^j_{-n}})$ using the same dynamic programming approach we presented in the general correlated setting.
        \item [-] We can efficiently compute $m_j \cdot f_{n \mid v_n}(\vec{t^j_{-n}})$, using the definition of $m_j$ from \cref{sec:representation-dfpa} and \cref{eq:app-conditional-representative}. 
    \end{itemize}
\end{proof}

\subsection{Efficient Utility Computation in Continuous First-Price
Auctions}\label{app:cfpa-utility} We will now show that we can also efficiently
compute the utilities in the CFPA, given the representation of
\cref{sec:CFPA-represent}:

\begin{lemma}\label{lem:app-cfpa-efficient-utilities} Fix a CFPA with correlated
    priors. For any bidder $i \in N$, any value $v_i \in V_i$, and any (pure)
    strategy profile of the other bidders $\vec{\beta_{-i}}$, the utility of
    $i$, $u_i(b,\vec{\beta_{-i}};v_i)$, when bidding $b$, is computable in
    polynomial time. Additionally, the utility is still efficiently computable
    in the $k$-GSAPV (in particular, SAPV) setting, where the prior is succinctly represented, when $\vec{\beta}_{-i}$ is symmetric.
\end{lemma}
\begin{proof}
   Similarly to \cref{lem:app-efficient-utilities}, we begin by considering,
   without loss of generality, a reordering of the bidders such that we are
   computing the utility of the last one (bidder $n=\card{N}$). Using the
   definition of utility in~\eqref{eq:utility}, we have
   $u_n(b,\vec{\beta_{-n}};v_n)=(v_n-b)H_{n}(b,\vec{\beta_{-n}};v_n)$, so it
   suffices to show that the $H$ function is efficiently computable. We can
   express this as:
   \begin{equation}\label{eq:app-cfpa-h-fcn}
    H_{n}(b,\vec{\beta_{-n}};v_n)=\sum_{r=0}^{n-1}\frac{1}{r+1}\int_{\vec{v}_{-n}\in \vec{V}_{-n}}f_{n|v_n}(\vec{v_{-n}}) T_n(b,n-1,r,\vec{v_{-n}})\, \mathrm{d}\vec{v_{-n}}
   \end{equation}
   where we define $T_n(b,n-1,r,\vec{v_{-n}})$ to indicate whether exactly $r$
   out of the remaining $n-1$ bidders bid exactly $b$ and all the others bid
   strictly less, when the values of the bidders other than $n$ are
   $\vec{v_{-n}}$. Firstly, $T_n(b,n-1,r,\vec{v_{-n}})$ is easy to compute using
   the same dynamic approach as in \cref{lem:app-efficient-utilities}, with the
   only difference of redefining~\eqref{eq:g-utility-dfpa}
   and~\eqref{eq:G-utility-dfpa} to:

   \begin{equation}\label{eq:g-utility-cfpa}
    g_{j,b} \coloneq \mathbbm{1}\left[\beta_j(v_j)=b\right]
    \end{equation}
    \begin{equation}\label{eq:G-utility-cfpa}
        G_{j,b} \coloneq \mathbbm{1}\left[\beta_j(v_j)<b\right]
    \end{equation}

    It remains to show how to compute the integral. Notice that, given the
    representation of the CFPA, we can replace the integral by a sum over the
    support of $f_{n\mid v_n}$:
   \begin{equation}\label{eq:cfpa-utility-h-via-repr}
    H_{n}(b,\vec{\beta_{-n}};v_n)=\sum_{r=0}^{n-1}\frac{1}{r+1} \cdot \sum_{\vec{v_{-n}}\in \support{F_{n \mid v_n}}} f_{n|v_n}(\vec{v_{-n}}) T_n(b,n-1,r,\vec{v_{-n}})
   \end{equation}

   We can efficiently compute the support of $F_{n \mid v_n}$ by checking
   whether $v_n \in \vec{R^j_n}$ for each hyperrectangle $\vec{R^j}$, adding
   $\vec{R^j_{-n}}$ to the support if that is the case. For the computation of
   $f_{n|v_n}(\vec{v_{-n}})$, notice that is suffices to be able to compute the
   marginal $f_n(v_n)$, which we can efficiently do by summing over all the
   elements in the support of $F_{n \mid v_n}$, which we just computed. Hence,
   there are polynomially many (to the size of the input) summands, each of
   which can be efficiently computed, which concludes the proof for general
   correlated values.

    Moving to the $k$-GSAPV setting, let
    $\mathcal{R}=\{(\vec{R^1},w_1),\ldots,(\vec{R^\ell},w_\ell)\}$ be the
    succinct representation of the joint distribution $F$. In this case, we can
    rewrite the $H$ functions that express the winning probability using the
    succinct representation as follows:
    \begin{equation}\label{eq:cfpa-utility-h-group-symmetric}
    H_{n}(b,\vec{\beta_{-n}};v_n)=\sum_{r=0}^{n-1} \frac{1}{r+1} \cdot \sum_{\substack{(\vec{R^j},w_j)\in \mathcal{R}: \\\vec{R^j_{-n}}\in \support{F_{n \mid v_n}} }} f_{n \mid v_n}(\vec{R^j_{-n}}) \cdot m_j \cdot T_n(b,n-1,r,\vec{R^j_{-n}}), 
    \end{equation} 
    where $m_j$ counts the number of valid permutations (with respect to the
    groups). To compute $m_j$, assume, without loss of generality, that the bidder $n$ belongs in the
    last group (the $k$-th). Then, $m_j = n_1! n_2!\ldots n_{k-1}!(n_k-1)!$.

    To derive the above, we have used the fact that the bidders are symmetric to
    guarantee that $T_n(b,n-1,r,\vec{v})=T_n(b,n-1,r,\vec{v'})$ for any
    $\vec{v'}$ that is a group-valid permutation of $\vec{v}$, and instead of summing over
    the whole support of the conditional, we have only summed over the elements
    that appear in the succinct representation. To prove that the RHS
    of~\eqref{eq:cfpa-utility-h-group-symmetric} (and thus the utility, when the
    other bidders play according to a symmetric strategy $\vec{\beta_{-n}}$) is
    efficiently computable, we show the following facts:
    \begin{itemize}[leftmargin=*]
        \item[-] First of all, notice that it is easy to check whether
        $\vec{R^j_{-n}}\in \support{F_{n \mid v_n}}$, by iterating through the
        intervals representing $\vec{R^j}$ and checking if $v_n$ is in any of
        them (out of the ones corresponding to bidder $n$'s group). 
        \item[-] $T_n(b,n-1,r,\vec{R^j_{-n}})$ is also efficiently computable,
        using the same dynamic approach as in the case of general correlated
        priors.
        \item[-] It remains to show that we can compute $f_{n \mid v_n}(\vec{R^j_{-n}})$ for every $\vec{R^j_{-n}}$.
        To see this, we will express it as follows:
        \begin{equation}\label{eq:cfpa-conditional-symmetric}
            f_{n \mid v_n}(\vec{R^j_{-n}}) = \frac{w_j}{f_n(v_n)} = \frac{w_j}{\int_{\vec{v_{-n}}\in \vec{V_{-n}}} f(v_n,\vec{v_{-n}})\, \mathrm{d}\vec{v_{-n}}} = \frac{w_j}{\sum_{j=1}^{\ell}w_j \cdot \sum_{\pi \in \perm{n_1,n_2,\ldots,n_k}} \mathbbm{1}_{\pi(\vec{R^j})}(v_n,\vec{v_{-n}})} 
        \end{equation}
        where the last step follows from the definition of the representation of
        the CFPA with SAPV. Finally, we need to reason that the denominator in
        the expression of the conditional is efficiently computable. Notice that
        if we naively try to compute it, there is an exponential blow up in the
        computation of all group-valid permutations. Instead, we will describe how to
        efficiently compute it using our succinct representation. To see this,
        notice that for every hyperrectangle $\vec{R^j}$ in the representation
        we can keep a count of the number of times each interval appears in the entries corresponding to $n$'s group.
        Then, we can also efficiently compute $S_n$, which we define to be the set of intervals of $\vec{R^j}$ that correspond to $n$'s group, in which $v_n$ is contained. 
        We can now express the sum over the valid permutations to be equal to $n_1!n_2!\dots n_{k-1}!(n_k-1)! \sum_{s \in S_n}c_s$, where we denote by $c_s$ the number of times interval $s$ appears in $\vec{R^j}$ and we have assumed, without loss of generality, that $n$ belongs to group $k$.
    \end{itemize}
    Using the above properties, we have demonstrated how to efficiently compute
    the RHS of~\eqref{eq:cfpa-utility-h-group-symmetric}, and therefore the
    utility $u_n(b,\vec{\beta_{-n}};v_n)$ for symmetric strategy profiles
    $\vec{\beta_{-n}}$.
\end{proof}

\section{Omitted Proofs from \texorpdfstring{\cref{sec:densification}}{Section \ref*{sec:densification}}}
\label{app:densification-proofs}

\subsection{Proof of~\texorpdfstring{\cref{lemma:inverter-poly-time-oracle-beta-MW}}{Lemma~\ref*{lemma:inverter-poly-time-oracle-beta-MW}}}

    We start with the IID setting. We will use the same notation for the IID
    value distribution representation, as we did in the proof of
    \cref{lemma:Lipschitz-bounds-beta-SAPV-bounded-IID} (see
    Page~\pageref{page:IID-represent-internal-lemma}). Recall that the density
    of the marginal distribution is given by $f_1(x)=p_{j}$ for all
    $x\in(a_{j-1},a_{j})$, $j\in[k]$, and thus its cdf, on any value $x\in
    [\underline{v},1]$, can be recursively computed, in polynomial time, by:
    \begin{equation}
        \label{eq:marginal-cdf-iid-recursive}
    F_1(x)=
    \begin{cases}
        xp_1,&\text{for}\;\; x\in[a_0,a_1],\\
        (x-a_{j-1})p_{j}+F(a_{j-1}), &\text{for}\;\; x\in[a_{j-1},a_{j}],\; j=2,3,\dots,k.
    \end{cases}
    \end{equation}
    In the following, it would also be convenient to have direct access to the
    intervals of $F_1$'s support, so we define $J^*\coloneqq
    \ssets{j\in[k]\fwh{p_k>0}}$ and $\bar{J}^*:=[k]\setminus J^*$. Note that
    sets $J^*$ and $\bar{J}^*$ can be computed in polynomial time from our
    representation of $F_1$.
    
    Now, to compute the canonical equilibrium strategy $\beta$, on any value
    $x\in[\underline{v},1]$, we first observe that by~\eqref{eq:MW-CCFPA-PBNE-beta-IID}:
    $$
    \beta(x) 
    = x-\int_{0}^x \frac{F_1^{n-1}(t)}{F_1^{n-1}(x)} \,\mathrm{d} t
    = x-\frac{1}{F_1^{n-1}(x)}\int_{0}^x F_1^{n-1}(t) \,\mathrm{d} t.
    $$
    Therefore, for any value $x\in V_1$ in the support of the marginal, i.e.,
    such that $x\in[a_\xi,a_{\xi+1}]$ with $\xi+1\in J^*$, we can compute the
    above integral as:
    \begin{align}
    \beta(x) &=x-\frac{1}{F_1^{n-1}(x)}\int_{\underline{v}}^x [F_1(t)]^{n-1} \,\mathrm{d} t \notag\\
        &=x-\frac{1}{F_1^{n-1}(x)}\left[\sum_{j=1}^\xi\int_{a_{j-1}}^{a_j} \left[F_1(a_{j-1})+(t-a_{j-1})p_j\right]^{n-1} \,\mathrm{d} t +\int_{a_{\xi}}^x \left[F(a_\xi)+(t-a_\xi)p_{\xi+1}\right]^{n-1} \,\mathrm{d} t\right] \notag\\
        &=x-\frac{1}{F_1^{n-1}(x)}\left[\sum_{j\in[\xi]\inters J^*}\int_{a_{j-1}}^{a_j} \left[F_1(a_{j-1})+(t-a_{j-1})p_j\right]^{n-1} \,\mathrm{d} t +\int_{a_{\xi}}^x \left[F_1(a_\xi)+(t-a_\xi)p_{\xi+1}\right]^{n-1} \,\mathrm{d} t\right. \notag\\
        &\qquad\qquad\qquad\left.+\sum_{j\in[\xi]\inters \bar{J}^*}\int_{a_{j-1}}^{a_j} \left[F_1(a_{j-1})+(t-a_{j-1})p_j\right]^{n-1} \,\mathrm{d} t \right] \notag\\
        &=x-\frac{1}{F_1^{n-1}(x)}\left[\sum_{j\in[\xi]\inters J^*}\frac{1}{np_{j}}\int_{a_{j-1}}^{a_j} \frac{\mathrm{d}}{\mathrm{d}t}\left[F_1(a_{j-1})+(t-a_{j-1})p_j\right]^{n} \,\mathrm{d} t\right. \notag\\
        &\qquad\qquad\qquad\left.+\frac{1}{np_{\xi+1}}\int_{a_{\xi}}^x \frac{\mathrm{d}}{\mathrm{d}t}\left[F_1(a_\xi)+(t-a_\xi)p_{\xi+1}\right]^{n} \,\mathrm{d} t+\sum_{j\in[\xi]\inters \bar{J}^*}\int_{a_{j-1}}^{a_j} F_1^{n-1}(a_{j-1}) \,\mathrm{d} t \right] \notag\\
        &=x-\frac{1}{F_1^{n-1}(x)}\left[\sum_{j\in[\xi]\inters J^*}\frac{1}{np_{j}}\left[\left(F_1(a_{j-1})+(a_j-a_{j-1})p_j\right)^{n}-F_1^n(a_{j-1})\right] \right.\notag\\ 
        &\qquad\qquad\qquad\left.+\frac{1}{np_{\xi+1}} \left[\left(F_1(a_{\xi})+(x-a_{\xi})p_{\xi+1}\right)^{n}-F_1^n(a_{\xi})\right] +\sum_{j\in[\xi]\inters \bar{J}^*}(a_j-a_{j-1}) F_1^{n-1}(a_{j-1}) \right] \notag\\
        &=x-\frac{1}{F_1^{n-1}(x)}\left[\sum_{j\in[\xi]\inters J^*}\frac{1}{np_{j}}\left[F_1^n(a_{j})-F_1^n(a_{j-1})\right] +\frac{1}{np_{\xi+1}} \left[F_1^n(x)-F_1^n(a_{\xi})\right] \right.\notag\\
        &\qquad\qquad\qquad\left.+\sum_{j\in[\xi]\inters \bar{J}^*}(a_j-a_{j-1}) F_1^{n-1}(a_{j-1}) \right] \label{eq:beta-explicit-iid-piecewise} 
    \end{align}
    Given that, by~\eqref{eq:marginal-cdf-iid-recursive}, we have
    polynomial-time oracle access to the values of the cdf $F_1$, it is not hard
    to see that the expression~\eqref{eq:beta-explicit-iid-piecewise} can be
    (exactly) computed in polynomial time as well (with respect to the binary
    representation of the input value $x$ and the auction's description).

    Finally, note that for the remaining values $x\in [\underline{v},1]\setminus
    V_1$, which are outside the marginal's support, computation is
    straightforward: since $\beta$ is constant in the intervals outside of
    $F_1$'s support, we can simply query the value of $\beta$ on the last point,
    before $x$, in the support. That is, we can compute (in polynomial time)
    index $\xi^*=\max\ssets{j\in J\fwh{a_j\leq x}}$ and then use
    $\beta(x)=\beta(a_{\xi^*})$, where $\beta(a_{\xi^*})$ can be computed, in
    polynomial time, as described above
    (in~\eqref{eq:beta-explicit-iid-piecewise}, by using $x\gets
    \alpha_{\xi^*}$).
    
    \bigskip
    We now move to the SAPV setting, under our standard assumption that the
    joint value $F$ has full support. Similar to our proof
    of~\cref{lemma:Lipschitz-bounds-beta-SAPV-bounded-IID}, we need to now use
    our hyperrectangle representation for $F$ (see~\cref{sec:CFPA-represent}).
    Our proof is split into a series of observations/steps, each of which we
    make sure that can be executed in polynomial time.
    \begin{enumerate}
    \item\label{item:SAPV-marginal-induced-representation} Due to the
    hyperrectangle representation of $F$ in our input, the marginal distribution
    $F_1$ is piecewise constant. That is, there exist pairs
    $\sset{(a_1,p_1),\dots,(a_k,p_k)}\in(0,1]^2$ with
    $0=a_0<a_1<a_2<\dots<a_{k-1}<a_k= 1$ and $\sum_{j=1}^kp_j=1$, such that the
    density $f_1$ of the marginal (which has full-support) is given by
    $f_1(x)=p_j$ if $x\in(a_{j-1},a_j)$ for $j\in[k]$. Furthermore, this
    representation can be constructed in polynomial time (with respect to the
    initial binary representation of $F$ in our input).
    
    \item\label{item:SAPV-marginal-order-statistics-induced-representation}
    Given a value $v\in [0,1]$, the distribution of the maximum value of all
    other players, namely $G_v$, has also piecewise constant density. Again, we
    can construct in polynomial time a list of value-density pairs
    $\sset{(c^v_1,q^v_1),\dots,(c^v_\ell,q^v_\ell)}\in(0,1]^2$ such that
    $g_v(y)=g^v_j$ if $y\in[c^v_{j-1},c^v_j]$. This can be done in two steps:
    \begin{enumerate}
    \item First, we find a succinct, hyperrectangle representation\footnote{That
    is, our standard representation for symmetric priors,
    see~\eqref{eq:represent-CFPA-density-permutations}
    in~\cref{sec:CFPA-represent}.} for the conditional distribution $F_{1|v}$ of
    all other bidders' values $\vec{v}_{-1}$, when the value of bidder $1$ (or
    any other bidder, due to symmetry) is fixed at $v$;
    see~\eqref{eq:cfpa-conditional-symmetric}, for more details about why this
    can be efficiently computed.
    \item Secondly, given a value $y\in [0,1]$, the computation of the density
    $g_v(y)$ boils down to identifying (up to symmetry) all the hyperrectangles
    in $F_{1|v}$'s representation from the first step, that are intersecting the
    outer faces of the hypercube $[0,y]^{n-1}$; i.e., we check for all
    hyperrectangles that, in any of their coordinates, include $y$, and we take
    the sum of their weight/densities, taking all their permutations into
    consideration (due to the symmetric representation;
    see~\eqref{eq:represent-CFPA-density-permutations}). The reason for this is
    that the corresponding cdf of the maximum order statistic $G_v(y)$ of
    $\vec{v}_{-1}$ is equal to the probability $\prob{Y_1 \leq y
    \fwh{X_1=y}}$.\footnote{This is conceptually similar to our arguments for
    the bodies $S_x$ in Page~\pageref{page:max-order-stat-(n-1)hypercubes} in
    the proof of~\cref{lemma:bounds-inverse-bidding}.} Finally, recall that, as
    we have argued multiple times within our proofs in~\cref{sec:densification},
    $\support{G_v}\subseteq \support{F_1}$, and therefore it is enough to only
    consider one value of $y$ for each interval $(a_{j-1},a_j)$ of $F_1$'s
    representation (see Step~\ref{item:SAPV-marginal-induced-representation}
    above). 
    \end{enumerate}
    Then, the corresponding cdf can also be efficiently computed, by
    observing that
    \begin{equation}
        \label{eq:oracle-beta-general-computation-G}
        G_v(y)=(y-c^v_{j-1})q_j^v+\sum_{l=1}^{j-1}(c^v_{l}-c^v_{l-1})q^v_l.
    \end{equation}
    \item\label{item:SAPV-marginal-order-statistics-constant-interval-induced-representation}
    Given a fixed value $y\in [0,1]$, functions $G_v(y), g_v(y)$ are constant,
    with respect to the conditional $v$, within the different intervals in the
    support of the marginal. That is, for any $j\in[k]$ and $v,v'\in
    [a_{j-1},a_j]$, it is $g_v(y)=g_{v'}(y)$ for all $y\in [0,1]$.
    \end{enumerate}

    Now we are ready to show how the values of the canonical bidding strategy
    $\beta$ (given in~\eqref{eq:MW-CCFPA-PBNE-beta}), can be (exactly) computed
    in polynomial time. Fix some $v\in [0,1]$ and let $v\in[a_{k_v-1},a_{k_v}]$
    be the interval of the marginal distribution representation (see
    Point~\ref{item:SAPV-marginal-induced-representation} above) in which it
    lies in. We will now show that function $L_v(y)$, given
    in~\eqref{eq:MW-CCFPA-PBNE-L}, is piecewise constant with respect to $y
    \in[0,v]$, across the intervals of the representation of $G_v$ (see
    Point~\ref{item:SAPV-marginal-order-statistics-induced-representation}
    above). Indeed $L_v(y)$ can be computed in the following recursive way:
    \begin{itemize}
        \item For $y\in[a_{k_v-1},v]$, it is
        \begin{align*}
            L_v(y)&= \exp\left(-\int_{y}^v\frac{g_t(t)}{G_t(t)}\,\mathrm{d}t\right)\\
            &= \exp\left(-\int_{y}^v\frac{g_{a_{k_v}}(t)}{G_{a_{k_v}}(t)}\,\mathrm{d}t\right), &&\text{due Point~\ref{item:SAPV-marginal-order-statistics-constant-interval-induced-representation} above, since}\;\; t\in[a_{k_v-1},a_{k_v}],\\
            &= \exp\left(-\int_{y}^v\left[\frac{\mathrm{d}}{\mathrm{d}t}\ln G_{a_{k_v}}(t)\right]\,\mathrm{d}t\right)\\
            &= \exp\left(\ln G_{a_{k_v}}(y)-\ln G_{a_{k_v}}(v)\right)\\
            &=\frac{G_{a_{k_v}}(y)}{G_{a_{k_v}}(v)},
        \end{align*}
        which can be computed in polynomial time, via~\eqref{eq:oracle-beta-general-computation-G}.
        \item For $y\in[a_{\kappa-1},a_{\kappa}]$, where $1\leq\kappa\leq k_v-1$, it is
        \begin{align*}
            L_v(y)&= \exp\left(-\int_{y}^{a_\kappa}\frac{g_t(t)}{G_t(t)}\,\mathrm{d}t -\int_{a_\kappa}^v\frac{g_t(t)}{G_t(t)}\,\mathrm{d}t\right)\\
            &= \exp\left(-\int_{y}^{a_\kappa}\frac{g_t(t)}{G_t(t)}\,\mathrm{d}t\right)\cdot\exp\left(-\int_{a_\kappa}^v\frac{g_t(t)}{G_t(t)}\,\mathrm{d}t\right)\\
            &= \exp\left(-\int_{y}^{a_\kappa}\frac{g_t(t)}{G_t(t)}\,\mathrm{d}t\right)\cdot L_v(a_\kappa)\\
            &= \exp\left(-\int_{y}^{a_\kappa}\frac{g_{a_\kappa}(t)}{G_{a_\kappa}(t)}\,\mathrm{d}t\right)\cdot L_v(a_\kappa),&&\text{since}\;\; t\in[a_{\kappa-1},a_\kappa],\\
            &= \frac{G_{a_\kappa}(y)}{G_{a_\kappa}(a_\kappa)}\cdot L_v(a_\kappa).
        \end{align*}
    \end{itemize}

\subsection{Proof of~\texorpdfstring{\cref{lemma:Lipschitz-bounds-beta-SAPV-bounded-IID}}{Lemma~\ref*{lemma:Lipschitz-bounds-beta-SAPV-bounded-IID}}}

Recall (see~\cref{sec:canonical-equilibrium-MW},
    page~\pageref{page:beta-absolutely-continuous}) that the canonical
    equilibrium strategy $\beta$ is absolutely continuous and almost everywhere
    differentiable (see~\eqref{eq:MW-beta-differential-equation-def}).
    Therefore, in order to bound its Lipschitz constant, it is enough to bound
    its derivative $\card{\beta'}$.\footnote{More precisely, here we are using
    the fact that, if for an absolutely continuous function $f:[a,b]\map\R$ it
    holds that $\card{f'}\leq c$ a.e. on $[a,b]$, for some $c>0$, then $\beta$
    is Lipschitz continuous with constant at most $c$. This is a direct
    consequence of the fundamental theorem of calculus: for any $x,y\in[a,b]$ we
    get that $\card{f(x)-f(y)}=\card{\int_x^y f'(t)\,\mathrm{d}t} \leq \int_x^y
    \card{f'(t)} \,\mathrm{d}t \leq \int_x^y c \,\mathrm{d}t = c\card{x-y}$.}
    Also recall that $\beta$ is nondecreasing (and thus, its derivative is
    nonnegative) and constant outside the marginals' support $V_1$. Therefore,
    it is enough to simply upper-bound $\beta'$ on $V_1$.

    We start first with the SAPV setting. Then, the value distribution is
    $(\priorlowerbound,\priorupperbound)$-bounded for $\priorlowerbound\geq
    w_{\min}$ and $\priorupperbound\leq \ell n! w_{\max}$
    (see~\eqref{eq:represent-CFPA-density-permutations}), where here we are
    using the parameters of the hyperrectangle representation
    from~\cref{sec:CFPA-represent}, with $w_{\min}\coloneqq\min_{j\in[\ell]}
    w_j$,  $w_{\max}\coloneqq\max_{j\in[\ell]} w_j$, and $\ell$ being the number
    of hyperrectangles in the input. Then, at any $y\in V_1$ we can bound
    $\beta$'s derivative by:
    \begin{align}
    \beta'(y) 
    &= \frac{g_y(y)}{G_y(y)}\int_{\underline{v}}^yL_y(t)\, \mathrm{d} t, &&\text{from~\eqref{eq:bounds-density-max-conditional-helper5},}\\
        &\leq \frac{g_y(y)}{G_y(y)}\int_{\underline{v}}^y 1\, \mathrm{d} t,&&\text{from Property~\ref{item:affiliation-property-measure-L} of~\cref{prop:affiliation-properties},}\label{eq:bounds-density-max-conditional-helper8}\\
        &\leq \frac{g_y(y)}{G_y(y)} y\notag\\
        &= \frac{g_y(y)}{\int_{0}^y g_y(t)\,\mathrm{d}t} y\notag, &&\text{since $G_y$ has full support $[0,1]$,}\\
        &\leq \frac{\frac{n-1}{f_1(y)}\priorupperbound y^{n-2}}{\int_{0}^y\frac{n-1}{f_1(y)}\priorlowerbound t^{n-2}}y\notag, &&\text{from~\eqref{eq:bounds-density-max-conditional-helper3},}\\
        &=\frac{\priorupperbound}{\priorlowerbound}\frac{(n-1)y^{n-2}}{y^{n-1}}y\notag\\
        &=(n-1)(\priorupperbound/\priorlowerbound)\notag\\
        &\leq (n-1)\frac{\ell n! w_{\max}}{w_{\min}}\notag\\
        &= 2^{\Theta(n\log n)}\ell\frac{w_{\max}}{w_{\min}}.\notag
    \end{align}
    
    We next consider IID values. We will make use that now the marginal $F_1$ is
    given to us (see also the representation in
    p.~\pageref{page:IID-represent-internal-lemma}) in the input by a set of
    possible values $0=a_0<a_1<a_2<\dots<a_{k-1}<a_{k}=1$ together with their
    corresponding probabilities $\ssets{p_j}_{j\in[k]}\in [0,1]$, such that
    $\sum_{j=1}^k(a_j-a_{j-1})p_j=1$. Then, the marginal's density is given by
    $f_1(x)=p_{j}$ for all $x\in(a_{j-1},a_{j})$, $j\in[k]$. Clearly, the value
    distribution is $\priorupperbound$-bounded, with
    $\priorupperbound=p_{\max}\coloneq \max_{j\in[k]} p_j>0$. Let us also denote
    $p_{\min}\coloneq \min_{j\in[k]}\ssets{p_j\fwh{p_j>0}}>0$,
    $\bar{a}_{\min}\coloneqq \min_{j\in[k]} (a_j-a_{j-1})>0$ and
    $\bar{a}_{\max}\coloneqq \max_{j\in[k]} (a_j-a_{j-1})>0$. Observe that all
    these quantities have polynomial binary representation.
    
    From~\eqref{eq:bounds-density-max-conditional-helper8}, for the IID setting
    we can bound the derivative of the canonical equilibrium strategy, at any
    value $x\in V_1$ in the support, by
    \begin{align}
    \beta'(x) &\leq \frac{g(x)}{G(x)}\int_{\underline{v}}^x 1\, \mathrm{d} t \notag\\
    &=\frac{g(x)}{G(x)}(x-\underline{v}) \notag\\
    &=\frac{(n-1)f_1(x)F_1^{n-2}(x)}{F_1^{n-1}(x)}(x-\underline{v})\notag\\
    &=(n-1)\priorupperbound\frac{x-\underline{v}}{F_1(x)}\label{eq:bound-Lipschit-beta-IID-helper1},
    \end{align}
    where, recall that, $\underline{v}=\inf V_1$ denotes the leftmost point in
    the support of the marginal value distribution. Taking our input
    representation into consideration, clearly it must be that
    $\underline{v}\in\ssets{a_0,a_1,\dots,a_{k-1}}$. Let $\xi\in[k]$ such that
    $\underline{v}=a_{\xi-1}$. Notice that it must be that $p_\xi>0$. Then, we
    observe that: 
    \begin{itemize}
    \item If $x\in[a_{\xi-1},a_\xi]$, then $F_1(x)=p_\xi(x-a_{\xi-1})=p_\xi(x-\underline{v})$ and therefore $\frac{x-\underline{v}}{F_1(x)}= \frac{1}{p_{\xi}}\leq\frac{1}{p_{\min}}\leq\frac{\bar{a}_{\max}}{\bar{a}_{\min}}\frac{1}{p_{\min}}$.
    \item If $x\in[a_{l-1},a_l]$ for some $\xi<l\leq k$, then it must be that
    $F(x) \geq (a_\xi-a_{\xi-1})p_\xi\geq \bar{a}_{\min}p_{\min} $ and
    $x-\underline{v}=x-a_{\xi-1}\leq a_{l}-a_{\xi-1}\leq \bar{a}_{\max}$.
    Therefore,
    $\frac{x-\underline{v}}{F_1(x)}\leq\frac{\bar{a}_{\max}}{\bar{a}_{\min}}\frac{1}{p_{\min}}$.
    \end{itemize}
    Using this bounds, we can finally upper-bound the derivative of $\beta$
    from~\eqref{eq:bound-Lipschit-beta-IID-helper1} by
    $$\beta'(x) \leq (n-1)\cdot p_{\max} \cdot \frac{\bar{a}_{\max}}{\bar{a}_{\min}}\frac{1}{p_{\min}}\leq n \frac{\bar{a}_{\max}}{\bar{a}_{\min}}\frac{p_{\max}}{p_{\min}}.$$

\section{Technical Lemmas}
\label{app:technical}

\begin{lemma}
\label{lemma:lower-bound-integral-derivative} Let $g:[a,b]\to\R_{\geq 0}$ be a
nondecreasing, absolutely continuous function, whose derivative is at most $L>0$
(almost everywhere). Then,
$$
\int_{a}^b F(t)\, dt \geq \frac{F^2(b)-F^2(a)}{2\cdot L}.
$$
\end{lemma}
\begin{proof}
Since $F$ is absolutely continuous and nondecreasing, its derivative $F'$ exists
almost everywhere and is nonnegative, thus 
$$
\int_{a}^b F(t)\, dt 
\geq \int_{a}^b \frac{F'(t)}{L}F(t)\, dt
= \frac{1}{2L}\int_{a}^b 2 F(t)F'(t)\, dt
= \frac{1}{2L}\int_{a}^b [F^2(t)]'\, dt
= \frac{1}{2L} \left[F^2(b)-F^2(a)\right].
$$
\end{proof}

\begin{lemma}
    \label{lemma:bounds-difference-of-powers} 
    Let $n$ be a positive integer, and $0\leq x \leq y$ reals. Then,
    $$
    n(y-x)x^{n-1}
    \leq
    y^n-x^n
    \leq 
    n(y-x)y^{n-1}.
    $$
\end{lemma}
\begin{proof}
Immediate from the convexity of the function $f(x)=x^n$ and the fact that its
derivative is $f'(x)=n x^{n-1}$.
\end{proof}

\bibliographystyle{plainnat}
\bibliography{refs}

\end{document}